\documentclass[a4paper,USenglish]{article}

\usepackage{tikz}
\usetikzlibrary{angles,quotes}
\usetikzlibrary{decorations}
\usepackage{pgfplots}
\usepackage{xspace}
\usepackage{amssymb,amsmath,amsthm}
\usepackage{thm-restate}
\usepackage{hyperref}
\usepackage{cleveref}
\usepackage[margin=3cm]{geometry}
\usepackage{authblk}

\newtheorem{theorem}{Theorem}
\newtheorem{proposition}[theorem]{Proposition}
\newtheorem{lemma}[theorem]{Lemma}
\newtheorem{definition}[theorem]{Definition}
\newtheorem{observation}[theorem]{Observation}

\title{Incremental Maximization via Continuization}

\author[1]{Yann Disser}
\author[2]{Max Klimm\thanks{Supported by Deutsche Forschungsgemeinschaft under Germany's Excellence Strategy, Berlin Mathematics Research Center (grant EXC-2046/1, Project 390685689).}}
\author[3]{Kevin Schewior\thanks{Supported in part by the Independent Research Fund Denmark, Natural Sciences, grant DFF-0135-00018B.}}
\author[1]{David Weckbecker\thanks{Supported by DFG grant DI 2041/2.}}

\affil[1]{TU Darmstadt, Germany, \tt \{disser|weckbecker\}@mathematik.tu-darmstadt.de}
\affil[2]{TU Berlin, Germany, \tt klimm@math.tu-berlin.de}
\affil[3]{University of Southern Denmark, Denmark, \tt kevs@sdu.dk}

\date{}

\newcommand{\R}{\mathbb{R}}
\newcommand{\N}{\mathbb{N}}
\newcommand{\C}{\mathbb{C}}

\newcommand{\E}{\mathbb{E}}
\newcommand{\randalg}{\textsc{RandomizedScaling}\xspace}

\newcommand{\alg}{\textsc{Alg}}
\newcommand{\opt}{\textsc{Opt}}
\newcommand{\greedycap}{\textsc{GreedyScaling}\xspace}
\newcommand{\incmax}{\textsc{IncMax}\xspace}
\newcommand{\incmaxsep}{\textsc{IncMaxSep}\xspace}
\newcommand{\contincmax}{\textsc{IncMaxCont}\xspace}

\newcommand{\capacity}{C}

\begin{document}

\maketitle

\begin{abstract}
We consider the problem of finding an incremental solution to a cardinality-constrained maximization problem that not only captures the solution for a fixed cardinality, but also describes how to gradually grow the solution as the cardinality bound increases.
The goal is to find an incremental solution that guarantees a good competitive ratio against the optimum solution for all cardinalities simultaneously.
The central challenge is to characterize maximization problems where this is possible, and to determine the best-possible competitive ratio that can be attained.
A lower bound of $2.18$ and an upper bound of $\varphi + 1 \approx 2.618$ are known on the competitive ratio for monotone and accountable objectives [Bernstein~et~al.,~Math.~Prog.,~2022], which capture a wide range of maximization problems.
We introduce a continuization technique and identify an optimal incremental algorithm that provides  strong evidence that $\varphi + 1$ is the best-possible competitive ratio.
Using this continuization, we obtain an improved lower bound of $2.246$ by studying a particular recurrence relation whose characteristic polynomial has complex roots exactly beyond the lower bound.
Based on the optimal continuous algorithm combined with a scaling approach, we also provide a $1.772$-competitive randomized algorithm.
We complement this by a randomized lower bound of~$1.447$ via Yao's principle.
\end{abstract}


\newpage

\section{Introduction}

A classical optimization problem takes as input a single instance and outputs a single solution. While this paradigm can be appropriate in static situations, it fails to capture scenarios that are characterized by perpetual growth, such as growing infrastructure networks, expanding companies, or private households with a steady income. In these cases, a single static solution may be rendered useless unless it can be extended perpetually into larger, more expansive solutions that are adequate for the changed circumstances.

To capture scenarios like this more adequately, we adopt the \emph{incremental optimization} framework formalized as follows. 
An instance of the \textsc{Incremental Maximization} (\incmax) problem is given by a countable set $U$ of elements and a monotone objective function $f \colon 2^U \to \mathbb{R}_{\geq 0}$ that assigns each subset $X \subseteq U$ a value $f(X)$.
A solution for an \incmax instance is an order $\sigma = (e_1,e_2,\dots)$ of the elements of $U$ such that each prefix of~$\sigma$ yields a good solution with respect to the objective function $f$. Formally, for $k \in [n]$, let \mbox{$\opt(k) = \max\{ f(X) : |X| = k, X \subseteq U\}$} denote the optimal value of the problem of maximizing $f(X)$ under the cardinality-constraint $|X| = k$. 
A deterministic solution $\sigma = (e_1,e_2,\dots)$ is called $\alpha$-competitive if $\opt(k) / f(\{e_1,\dots,e_k\}) \leq \alpha$ for all $k \in [n]$.
A randomized solution is a probability distribution $\Sigma = (E_1,E_2,\dots)$ over deterministic solutions (where $E_1,E_2,\dots$ are random variables).
It is called $\alpha$-competitive if $\opt(k) / \mathbb{E}[f(\{E_1,\dots,E_k\})] \leq \alpha$ for all $k \in [n]$.
In both cases, we call the infimum over all $\alpha \geq 1$, such that the solution is $\alpha$-competitive, the \emph{(randomized) competitive ratio} of the solution.
A (randomized) algorithm is called $\alpha$-competitive for some $\alpha \geq 1$ if, for every instance, it produces an $\alpha$-competitive solution,
and its \emph{(randomized) competitive ratio} is the infimum over all such~$\alpha$.
The \emph{(randomized) competitive ratio} of a class of problems (or a problem instance) is the infimum over the competitive ratios of all (randomized) algorithms for it.

Clearly, in this general form, no meaningful results regarding the existence of competitive solutions are possible. 
For illustration consider the instance $U = \{a,b,c\}$ where for some $M \in \mathbb{N}$ we have 
\[
	f(X) = \begin{cases}M, & \textrm{if }\{b,c\}\subseteq X,\\ |\{a\}\cap X|, & \textrm{otherwise}\end{cases} \quad \text{ for all $X \subseteq U$.} 
\]
Then, every solution needs to start with element~$a$ in order to be competitive for $k=1$, but any such order cannot be better than $M$-competitive for $k=2$. The underlying issue is that the optimal solution for $k=2$ given by $\{b,c\}$ does not admit a competitive partial solution of cardinality $k=1$. To circumvent this issue, Bernstein et al.~\cite{BernsteinDisserGrossHimburg/20} consider \emph{accountable} functions,
i.e., functions~$f$, such that, for every $X\subseteq U$, there exists $e\in X$ with $f(X\setminus\{e\})\geq f(X)-f(X)/|X|$.
They further show that many natural incremental optimization problems are monotone and accountable such as the following.
\begin{description}
	\item[Weighted matching:] $U$ is the set of edges of a weighted graph, and $f(X)$ is the maximum weight of a matching contained in $X$;
	\item[Set packing:] $U$ is a set of weighted subsets of a ground set, and $f(X)$ is the maximum weight of a set of mutually disjoint subsets of $X$;
	\item[Submodular function maximization:] $U$ is arbitrary, and $f$ is monotone and submodular;
	\item[(Multi-dimensional) Knapsack:] $U$ is a set of items with (multi-dimensional) sizes and values, and $f(X)$ is the maximum value of a subset of items of $X$ that fits into the knapsack. 
\end{description}

Bernstein et al.~\cite{BernsteinDisserGrossHimburg/20} gave an algorithm to compute a $(1 + \varphi)$-competitive incremental solution and showed that the competitive ratio of the \incmax problem is at least $2.18$. 
Throughout this work, we assume that the objective $f$ is accountable.

%

\vspace*{.2cm}
\textbf{Our results.}
As a first step, we reduce the general \incmax problem to the special case of \incmaxsep, where the elements of the instance can be partitioned into a (countable) set of uniform and modular subsets such that the overall objective is the maximum over the modular functions on the subsets.
We then define the \contincmax problem as a continuization, assuming that there exists one such subset with (fractional) elements of every size~$c\in\R_{>0}$.
The smooth structure of this problem better lends itself to analysis.

We consider the continuous algorithm $\greedycap(c_1,\rho)$ that adds a sequence of these subsets, starting with the subset of size~$c_1>0$ and proceeding along a sequence of subsets of largest possible sizes under the constraint that $\rho$-competitiveness is maintained for as long as possible.
We first show that there always exists an optimal solution of this form.

\begin{restatable}{theorem}{GreedyCapBest}\label{thm:GreedyCap_best_possible}
	For every instance of~\contincmax, there exists a starting value~$c_1$ such that the algorithm $\greedycap(c_1,\rho^*)$ achieves the best-possible competitive ratio~$\rho^*\geq1$.
\end{restatable}

Our continuous embedding allows us to view every algorithm as an increasing sequence of sizes of subsets that are added one after the other.
Using elementary calculus, we can show that, with the golden ratio \mbox{$\varphi := \frac{1}{2}(1+\sqrt{5})\approx1.618$}, $\greedycap(c_1,\rho)$ achieves the known upper bound of $\varphi + 1$ for a range of starting values.
Here, $d(c)$ refers to the density, i.e., value per size, of the subset of size~$c$ (see Sec.~\ref{sec:separability}).

\begin{restatable}{theorem}{GreedyCapComp}\label{thm:GreedyCap_phi+1-comp}
	$\greedycap(c_{1},\varphi+1)$ is $(\varphi+1)$-competitive if and only if $d(c_1)\geq\frac{1}{\varphi+1}$.
\end{restatable}

On the other hand, we are able to, for every starting value~$c_1$, construct an instance of\linebreak \contincmax where $\greedycap(c_1,\rho)$ is not better than~$(\varphi + 1)$-competitive for any $\rho > 1$.
We emphasize that the optimum value of $\varphi + 1$ emerges naturally from the geometry of complex roots.
Based on this evidence, we conjecture that $\varphi + 1$ is the best-possible competitive ratio.

Of course, proving a general lower bound requires to construct a single instance such that $\greedycap$ is not better than~$(\varphi + 1)$-competitive for \emph{every} starting value.
Careful chaining of our construction for a single starting value yields the following.

\begin{restatable}{proposition}{GreedyCapLB}\label{thm:lower_bound_greedycap}
	For every countable set $S\subset\R_{>0}$ of starting values, there exists an instance of \contincmax such that $\greedycap(c_{1},\rho)$ is not $\rho$-competitive for any $c_{1}\in S$ and any $\rho<\varphi+1$.
\end{restatable}

Crucially, while this gives a lower bound if we only allow rational starting values~$c_1 \in \mathbb{Q}$, transferring the lower bound back to \incmax requires excluding all reals.
Even though we are not able to achieve this, we can extrapolate our analysis in terms of complex calculus to any \contincmax algorithm.
With this, we beat the currently best known lower bound of~$2.18$ in~\cite{BernsteinDisserGrossHimburg/20}.

\begin{restatable}{theorem}{detLB}\label{thm:det_lower_bound}
	The \incmax\ problem has a competitive ratio of at least $2.246$.
\end{restatable}

We can also apply our technique, specifically the reduction to separable problem instances and the structure of the \greedycap algorithm, to the analysis of randomized algorithms for\linebreak \incmax.
We employ a scaling approach based on the algorithms in~\cite{BernsteinDisserGrossHimburg/20}, combined with a randomized selection of the starting value~$c_1$ inspired by a randomized algorithm for the \textsc{CowPath} problem in \cite{KaoReifTate/93}.
The resulting algorithm has a randomized competitive ratio that beats our deterministic lower bound.

\begin{restatable}{theorem}{randUB}\label{thm:rand_upper_bound}
	\incmax admits a $1.772$-competitive randomized algorithm.
\end{restatable}

We complement this result with a lower bound via Yao's principle for separable instances of \incmax.

\begin{restatable}{theorem}{randLB}\label{thm:lb}
	Every randomized \incmax algorithm has competitive ratio at least~$1.447$.
\end{restatable}

\textbf{Related work.}
Our work is based on the incremental maximization framework introduced by Bernstein et al.~\cite{BernsteinDisserGrossHimburg/20}. We provide a new structural understanding that leads to a better lower bound and new randomized bounds.

A similar framework is considered for matchings by Hassin and Rubinstein~\cite{HassinR02}. 
Here, the objective~$f$ is the total weight of a set of edges and the solution is additionally required to be a matching.
Hassin and Rubinstein~\cite{HassinR02} show that the competitive ratio in this setting is $\sqrt{2}$ and Matuschke, Skutella, and Soto~\cite{MatuschkeSS18} show that the randomized competitive ratio is $\ln(4) \approx 1.38$.
The setting was later generalized to the intersection of matroids~\cite{FujitaKM13} and to independence systems with bounded exchangeability~\cite{KakimuraM13,Mestre06}.
Note that, while our results hold for a broader class of objective functions, we require monotonicity of the objective and cannot model the constraint that the solution must be a matching.
We can, however, capture the matching problem by letting the objective~$f$ be the largest weight of a matching contained as a subset in the solution (i.e., not all parts of the solution need to be used).
That being said, it is easy to verify that the lower bound of $\sqrt{2}$ on the competitiveness of any deterministic algorithm in the setting of~\cite{HassinR02} also applies in our case.

Hassin and Segev~\cite{HassinS06} studied the problem of finding a small subgraph that contains, for all~$k$, a path (or tree) of cardinality at most $k$ with weight at least $\alpha$ times the optimal solution and show that for this $\alpha|V|/(1 - \alpha^2)$ edges suffice.
There are further results on problems where the items have sizes and the cardinality-constraint is replaced by a knapsack constraint \cite{DisserKMS17,DisserKW21,KlimmK22,MegowM13}.
Goemans and Unda~\cite{GoemansU17} studied general incremental maximization problems with a sum-objective.

Incremental \emph{minimization} problems further been studied for a variety of minimzation problems such as $k$-median \cite{ChrobakKNY08,MettuP03,LinNRW10}, facility location \cite{LinNRW10,Plaxton06}, and $k$-center \cite{Gonzalez85,LinNRW10}.
As noted by Lin et al.~\cite{LinNRW10}, the results for the minimum latency problem in \cite{BlumCCPRS94,GoemansK98} implicitly yield results for the incremental $k$-MST problem.
There are further results on incremental optimization problems where in each step the set of feasible solution increases \cite{HartlineS06,HartlineS07}.

\section{Separability of Incremental Maximization}\label{sec:separability}

As a first step to bound the competitive ratio of \incmax, we introduce a subclass of instances of a relatively simple structure, and show that it has the same competitive ratio as \incmax.
Thus, we can restrict ourselves to this subclass in  our search for bounds on the competitive ratio.

\begin{definition}
	An instance of \incmax with objective $f\colon 2^U\rightarrow\R_{>0}$ is called \emph{separable} if there exist a partition $U=R_1\cup R_2\cup \dots$ of~$U$ and values $d_i>0$ such that
	\[
	f(X)=\max_{i\in\N} \{|X\cap R_i|\cdot d_i\} \quad \text{for all $X\subseteq U$}.
	\]
	We refer to~$d_i$ as the \emph{density} of set~$R_i$ and to $v_i:=|R_i|\cdot d_i$ as the \emph{value} of set~$R_i$.
	The restriction of \incmax to separable instances will be denoted by \incmaxsep.
\end{definition}

We start our analysis of \incmaxsep with the following immediate observation.

\begin{lemma}\label{obs:incmaxsep_assumptions}
	Any instance of \incmaxsep can be transformed into one with the same or a worse competitive, that satisfies the following properties.
	\begin{enumerate}
		\item There is exactly one set of every cardinality, i.e., $|R_i|=i$.
		
		\item Densities are decreasing, i.e., $1 \geq d_1 \geq d_2 \geq \dots$.
		
		\item Values are increasing, i.e., $v_1 \leq v_2 \leq \dots$.
	\end{enumerate}
\end{lemma}

\begin{proof}
	We will show that every instance that does not satisfy the assumptions can be transformed into one that does, without changing the optimum value for any size, and without changing the value of the best incremental solution.
	Thus the competitive ratio of the two instances coincide.
	
	If there are two sets $R_i,R_j$ with $|R_i|=|R_j|$, it only makes sense to consider the one with higher density, as every solution adding elements from the set of lower density can be improved by adding elements from the other set instead.
	If there is $i\in\N_{\geq2}$ such that there is no set with~$i$ elements, we can add a new set with~$i$ elements to the instance.
	This new set will have value $v_{i-1}$ .
	Then, every solution that adds elements from the newly introduced set can be improved by adding elements from set $R_{i-1}$ instead.
	If there is no set $R_1$ with $1$ element, we can introduce it with density~$d_2$.
	Then, every solution that adds this one element can instead also add one element from~$R_2$.
	Thus, the first assumption can be made.
	
	The assumption that $1\geq d_1$ is without loss of generality by rescaling the objective~$f$.
	If there was $i\in\N$ with $d_i<d_{i+1}$, every solution to the problem instance that adds elements from the set~$R_i$ could be improved by adding elements from the set~$R_{i+1}$ instead.
	Since $|R_{i+1}|\geq|R_i|$, this is possible.
	
	The third assumption can be made because, if there was $i\in\N$ with $v_i >v_{i+1}$, a solution that adds elements from~$R_{i+1}$ can be improved by adding elements from~$R_i$ instead.
\end{proof}

In the following, we assume that every instance satisfies the properties from Lemma~\ref{obs:incmaxsep_assumptions}.

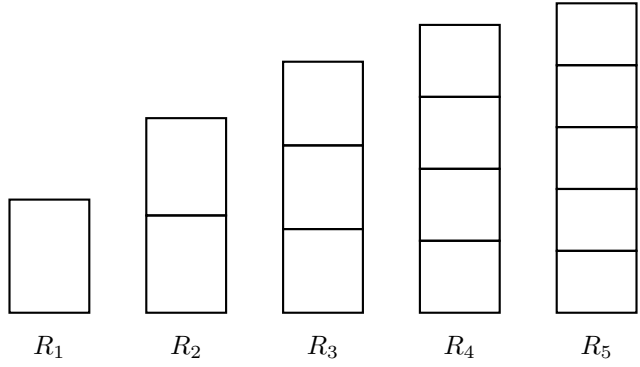
\begin{figure}[t]
	\begin{center}
		\begin{tikzpicture}[scale=1.5]
			\def\columnwidth{0.7}
			\def\factor{0.86}
			\def\heightA{1}
			\def\heightB{\factor}
			\def\heightC{\factor*\factor}
			\def\heightD{\factor*\factor*\factor}
			\def\heightE{\factor*\factor*\factor*\factor}
			
			\def\xOffset{0*(\columnwidth+0.5)}
			\node at ({\xOffset+0.5*\columnwidth},-0.3) {$R_1$};
			\foreach \x in {1}{
				\draw[black,thick] ({\xOffset},{(\x-1)*\heightA}) rectangle ({\xOffset+\columnwidth},{\x*\heightA});
			}
			
			\def\xOffset{1*(\columnwidth+0.5)}
			\node at ({\xOffset+0.5*\columnwidth},-0.3) {$R_2$};
			\foreach \x in {1,2}{
				\draw[black,thick] ({\xOffset},{(\x-1)*\heightB}) rectangle ({\xOffset+\columnwidth},{\x*\heightB});
			}
			
			\def\xOffset{2*(\columnwidth+0.5)}
			\node at ({\xOffset+0.5*\columnwidth},-0.3) {$R_3$};
			\foreach \x in {1,2,3}{
				\draw[black,thick] ({\xOffset},{(\x-1)*\heightC}) rectangle ({\xOffset+\columnwidth},{\x*\heightC});
			}
			
			\def\xOffset{3*(\columnwidth+0.5)}
			\node at ({\xOffset+0.5*\columnwidth},-0.3) {$R_4$};
			\foreach \x in {1,2,3,4}{
				\draw[black,thick] ({\xOffset},{(\x-1)*\heightD}) rectangle ({\xOffset+\columnwidth},{\x*\heightD});
			}
			
			\def\xOffset{4*(\columnwidth+0.5)}
			\node at ({\xOffset+0.5*\columnwidth},-0.3) {$R_5$};
			\foreach \x in {1,2,...,5}{
				\draw[black,thick] ({\xOffset},{(\x-1)*\heightE}) rectangle ({\xOffset+\columnwidth},{\x*\heightE});
			}
		\end{tikzpicture}
	\end{center}
	\caption{Illustration of an instance of \incmaxsep with $N=5$ sets. Each set $R_i$ consists of~$i$ elements. The height of the elements represents their value. As in Lemma~\ref{obs:incmaxsep_assumptions}, the values of the single elements becomes less the larger~$i$ is, while the value of the whole set~$R_i$ increases.}
\end{figure}
\begin{definition}
	We say that a solution for \incmaxsep is \emph{represented} by a sequence of sizes $(c_{1},c_{2},\dots)$ if it first adds all elements from the set $R_{c_1}$, then all elements from the set $R_{c_2}$, and so on.
\end{definition}
A solution of \incmaxsep can only improve if it is altered in a way that it chooses sets $R_{c_1},R_{c_2},\dots$ and adds all elements of these sets, one after the other.
If not all elements of one set are added, the solution does not degrade if a smaller set is added instead because the density of the smaller set is at least as large as the density of the larger set.
Adding all elements of one set consecutively is better because the value of the solution increases faster this way.

\begin{lemma}[\cite{BernsteinDisserGrossHimburg/20}, Observation 2]
	There is an algorithm achieving the best-possible competitive ratio for \incmaxsep such that the solution generated by this algorithm can be represented by a sequence $(c_{1},c_{2},\dots)$.
	We can assume that $v_{c_{i}}<v_{c_{i+1}}$ and thus, since the values $(v_i)_{i\in\N}$ are non-decreasing, $c_{i}<c_{i+1}$ for all $i\in\N$.
\end{lemma}

From now on, we will only consider solutions of this form and denote a solution~$X$ by the sequence it is represented by, i.e., $X=(c_{1},c_{2},\dots)$. For a size \mbox{$\capacity\in\N$}, we denote by~$X(\capacity)$ the first~$C$ elements added by~$X$, i.e., $|X(\capacity)|=\capacity$ and, with $O_i:=\arg\max\{f(S)\mid S\subseteq U,|S|=i\}$, we have $X\bigl(\sum_{i=1}^k c_i\bigr) = \bigcup_{i=1}^k O_{c_i}$.

In order to compare the competitive ratios of \incmax and \incmaxsep, we need the following lemma which gives an intuition why accountability is a desirable property for the objective of \incmax to have.

\begin{lemma}\label{lem:accountable_iff_good_ordering_exists}
	A function $f\colon U\rightarrow\R_{\geq0}$ is accountable if and only if, for every $X\subseteq U$ and with \mbox{$n:=|X|\in\N$}, there exists an ordering $(e_{1},\dots,e_{n})$ of $X$ with $f(\{e_{1},\dots,e_{i}\})\geq\frac{i}{n}f(X)$ for all $i\in\{1,\dots,n\}$.
\end{lemma}

\begin{proof}
	``$\Leftarrow$'': Let $X\subseteq U$.
	There is an ordering $(e_{1},\dots,e_{n})$ of $X$ with $f(\{e_{1},\dots,e_{i}\})\geq\frac{i}{n}f(X)$ for all $i\in\{1,\dots,n\}$.
	Especially, we have
	\[
	f(X\setminus\{e_{n}\})=f(\{e_{1},\dots,e_{n-1}\})\geq\frac{n-1}{n}f(X)=f(X)-\frac{f(X)}{|X|},
	\]
	i.e., $f$ is accountable.
	
	``$\Rightarrow$'': Let $f$ be accountable and $X_{n}:=X$.
	We define $X_{n-1},\dots,X_{1}$ recursively.
	Suppose, for \mbox{$i\in\{1,\dots,n-1\}$}, the set $X_{i+1}$ is defined and $|X_{i+1}|=i+1$.
	By accountability of $f$, there exists $e\in X_{i+1}$ with
	\[
	f(X_{i+1}\setminus\{e\})\geq f(X_{i+1})-\frac{f(X_{i+1})}{|X_{i+1}|}=\frac{i}{i+1}f(X_{i+1}).
	\]
	Let $X_{i}:=X_{i+1}\setminus\{e\}$.
	We define the ordering $(e_{1},\dots,e_{n})$ to be the unique ordering of $X$ such that $X_{i}=\{e_{1},\dots,e_{i}\}$ for all $i\in\{1,\dots,n\}$.
	Then, for $i\in\{1,\dots,n-1\}$, we have
	
	\begin{equation*}
		f(\{e_{1},\dots,e_{i}\}) = f(X_{i}) \geq \frac{i}{i+1}f(X_{i+1}) \geq \dots \geq \frac{i}{i+1}\cdot\dots\cdot\frac{n-1}{n}f(X_{n}) = \frac{i}{n}f(X).\qedhere
	\end{equation*}
\end{proof}

With the help of this lemma, we are ready to show the following.

\begin{proposition}\label{prop:CR_incmax_incmaxsep}
	The competitive ratios of \incmax and \incmaxsep coincide.
\end{proposition}

\begin{proof}
	As \incmaxsep is a subclass of \incmax, the competitive ratio of \incmaxsep is smaller or equal to that of \incmax.
	
	To show the other direction, let $f\colon U\rightarrow\R_{\geq0}$ be the accountable objective of some instance of \incmax.
	We will construct an instance of \incmaxsep such that every $\rho$-competitive solution of this problem instance induces a $\rho$-competitive solution for the initial instance of \incmax with objective~$f$.
	
	Recall that $\opt(i)=\max\{f(X)\mid |X|=i, X\subseteq U\}$. Let $O_i\subseteq U$ with $|O_i|=i$ such that $f(O_i)=\opt(i)$.
	
	Now, we define the instance of \incmaxsep.
	Let $n:=|U|$ and let $R_{1},\dots,R_{n}$ be disjoint sets such that, for $i\in\{1,\dots,n\}$, $|R_{i}|=i$.
	For $i\in\{1,\dots,n\}$, let $d_i:=\opt(i)/i$, i.e., we have $v_i=\opt(i)$.
	Let the objective of this problem be denoted by $f_{\text{sep}}$.
	
	Let $X=(c_{1},\dots,c_{n})$ be a $\rho$-competitive solution of the separable problem instance that we just defined.
	We define a solution~$\tilde{X}$ to the \incmax\ problem as follows.
	First, we add all elements from the set~$O_{c_1}$, then all elements from the set~$O_{c_2}$ and so on, until we added all optimal solutions $O_{c_1},\dots,O_{c_n}$.
	For $i\in\{1,\dots,n\}$, the elements of the set~$O_{c_i}$ are added in the order given by Lemma~\ref{lem:accountable_iff_good_ordering_exists}.
	Fix a size $\capacity\in\{1,\dots,n\}$.
	Let $i\in\{1,\dots,n\}$ be such that the last element added to the solution $\tilde{X}(\capacity)$ is from~$O_{c_i}$.
	Note that, for size~$\capacity$,~$\tilde{X}$ has added the optimal solution of size~$c_{i-1}$ completely and $\smash{\capacity-\sum_{j=1}^{i-1}c_j}$ elements from the optimal solution of size~$c_i$.
	By Lemma~\ref{lem:accountable_iff_good_ordering_exists} and monotonicity of~$f$, the value of the solution~$\tilde{X}(\capacity)$ is
	\begin{equation}\label{eq:lem_incmax_RC_1}
	f(\tilde{X}(\capacity))=\max\Bigl\{ \opt(c_{i-1}),\frac{\capacity-\sum_{j=1}^{i-1}c_j}{c_i}\opt(c_i) \Bigr\}.
	\end{equation}
	Similar to~$\tilde{X}$, the solution~$X(\capacity)$ of the separable problem instance contains all elements from the set~$R_{c_{i-1}}$ and $\capacity-\sum_{j=1}^{i-1}c_j$ elements from the set~$R_{c_i}$.
	Thus, the value of the solution~$X(\capacity)$ is
	\[
	f_{\text{sep}}(X(\capacity)) = \max\Bigl\{ v_{c_{i-1}},\bigl(\capacity-\sum_{j=1}^{i-1}c_j\bigr)d_{c_i} \Bigr\} = \max\Bigl\{ v_{c_{i-1}},\frac{\capacity-\sum_{j=1}^{i-1}c_j}{c_i}v_{c_i} \Bigr\}.
	\]
	Combining this with the fact that $v_{c_j}=\opt(c_j)$ for all $j\in\{1,\dots,n\}$ as well as with~\eqref{eq:lem_incmax_RC_1}, we obtain
	\[
	f(\tilde{X}(\capacity)) = f_{\textsc{RC}}(X(\capacity)) \geq \frac{1}{\rho}v(\capacity) = \frac{1}{\rho}\opt(\capacity).\qedhere
	\]
\end{proof}

\section{Continuization Results}\label{sec:continuization}

In order to find lower bounds on the competitive ratio of \incmaxsep, we transform the problem into a continuous one.

\begin{definition}
	In the \contincmax problem, we are given a \emph{density function} \mbox{$d\colon\R_{\geq0}\rightarrow(0,1]$} and a \emph{value function} $v(c):=cd(c)$.
	As for the discrete problem, we denote an incremental solution~$X$ for \contincmax by a sequence of sizes $X=(c_1,c_2,\dots)$.
	For a given size~$c\geq0$, we denote the solution of this size by $X(c)$.
	With $n\in\N$ such that $\sum_{i=1}^{n-1}c_i<c\leq \sum_{i=1}^{n}c_i$, the value of $X(c)$ is defined as
	\[
	f(X(c)):=\max\Biggl\{ \max_{i\in\{1,\dots,n-1\}}v(c_i),\Biggl(c-\sum_{i=1}^{n-1}c_i\Biggr)d(d_n) \Biggr\}.
	\]
	An incremental solution~$X$ is $\rho$-competitive if $\rho\cdot f(X(c))\geq v(c)$ for all $c>0$.
	The competitive ratio of~$X$ is defined as $\inf\{ \rho\geq1\mid X\text{ is }\rho\text{-competitive} \}$.
\end{definition}

The interpretation of the functions~$d$ and~$v$ is that the instance is partitioned into sets, one for every positive size $c \in \mathbb{R}$, each consisting of~$c$ fractional units with a value of $d(c)$ per unit, for a total value of $v(c)$ for the set.
The solution can be interpreted in the following way:
It starts by adding the set of size~$c_1$, then the set of size~$c_2$, and so on.
With $n\in\N$ such that $\sum_{i=1}^{n-1}c_i<c\leq \sum_{i=1}^{n}c_i$, the solution $X(c)$ has added all of the sets of sizes $c_1,\dots,c_{n-1}$ and $c-\sum_{i=1}^{n-1}c_i$ units of the set of size~$c_n$.
Unlike the \incmaxsep problem, the \contincmax problem includes subsets of all real sizes instead of only integer sizes and, furthermore, allows fractional elements to be added to solutions instead of only an integral number of elements.

As for the discrete version of the problem, without loss of generality, we assume that the density function~$d$ is non-increasing and the value function~$v$ is non-decreasing.
These assumptions imply that~$d$ is continuous:
If this was not the case and~$d$ was not continuous for some size~$c'$, i.e., $\lim_{c\nearrow c'}d(c)>\lim_{c\searrow c'}d(c)$, then $\lim_{c\nearrow c'}v(c)>\lim_{c\searrow c'}v(c)$ by definition of $v$, i.e., $v$ would not be increasing in $c$.
So $d$ is continuous, and, by definition of $v$, also $v$ is continuous.
Furthermore, without loss of generality, we assume that $d(0)=1$.

For a fixed size $c\geq0$, we define $p(c)=\max\{c'\geq0\mid v(c')\leq\rho v(c)\}$.
This value gives the size up to which a solution with value $v(c)$ is $\rho$-competitive.
Throughout our analysis, we assume that $p(c)$ is defined for every $c\geq 0$, i.e., that $\lim_{c\rightarrow\infty}v(c)=\infty$.
Otherwise, any algorithm can terminate when the value of its solution is at least $\frac{1}{\rho}\sup_{c\in\R_{\geq0}}v(c)$.

\begin{proposition}\label{thm:CRC_to_RC}
	The competitive ratio of \incmaxsep is greater or equal to that of \contincmax.
\end{proposition}

\begin{proof}
	Let an instance of the \contincmax problem with value function $v\colon\R_{\geq0}\rightarrow\R_{\geq0}$ and density function $d\colon\R_{\geq0}\rightarrow\R_{\geq0}$ be given, and let $\rho\geq1$, $\epsilon>0$.
	We will construct an instance of the \incmaxsep problem such that every solution for this problem instance with competitive ratio~$\rho$ implies a solution for the \contincmax with competitive ratio $\rho+\epsilon$.
	Let $\epsilon'>0$ be small enough such that
	\begin{equation}
	\epsilon'<\rho+\frac{1}{4}-\sqrt{\bigl(\rho+\frac{1}{4}\bigr)^{2}-\frac{1}{2}}\label{eq:CRC_and_RC_def_eps-prime_1}
	\end{equation}
	and
	\begin{equation}
	\rho<(1-\epsilon')\lfloor1+\rho\rfloor.\label{eq:CRC_and_RC_def_eps-prime_2}
	\end{equation}
	This is possible since $\rho<\lfloor1+\rho\rfloor$.
	Furthermore, let $c_{\textrm{min}}\geq0$ be the largest value with\linebreak \mbox{$d\bigl(c_{\textrm{min}}+\frac{\rho}{1-\epsilon'}v(c_{\textrm{min}})\bigr)=1-\epsilon'$}.
	
	We define the \incmaxsep problem as follows.
	Let $n\in\N$ be large enough such that
	\begin{equation}
	d\Bigl(\frac{1}{n}\Bigr)\geq\max\Bigl\{\frac{\rho}{\rho+\epsilon},1-\epsilon'\Bigr\}\label{eq:CRC_and_RC_def_n_1}
	\end{equation}
	and
	\begin{equation}
	\frac{nc_{\textrm{min}}-1}{nc_{\textrm{min}}}\geq\frac{\rho}{\rho+\epsilon}.\label{eq:CRC_and_RC_def_n_2}
	\end{equation}
	Let $R_1,R_2,\dots$ be disjoint sets of elements with $|R_{i}|=i$.
	For $i\in\N$, let $d_i:=v\bigl(\frac{i}{n}\bigr)/i$ be the density of set~$R_{i}$, i.e., we have
	\begin{equation}
	v_i=v\Bigl(\frac{i}{n}\Bigr)\label{eq:CRC_and_RC_relation_values}
	\end{equation}
	and
	\begin{equation}
	d_i=\frac{d\bigl(\frac{i}{n}\bigr)}{n}.\label{eq:CRC_and_RC_relation_densities}
	\end{equation}
	Let $X=(c_{1},c_{2},\dots)$ be a $\rho$-competitive solution
	for this instance of \incmaxsep.
	Without loss of generality, we can assume that $c_{i}<c_{i+1}$ for all $i\in\N$.
	Furthermore, we can assume that $\smash{d_{c_{2}}<\frac{d_1}{\rho}}$.
	If this was not the case, we could simply consider the solution represented by $(c_{2},c_{3},\dots)$, which is also $\rho$-competitive.
	
	We define the solution $X_{\textrm{c}}=\bigl(\frac{c_{1}}{n},\frac{c_{2}}{n},\dots\bigr)$ for the instance of \contincmax.
	We will show that this solution is $(\rho+\epsilon)$-competitive.
	Fix a size $\capacity\geq0$.
	If $\smash{\capacity<\frac{c_{1}}{n}}$, $X_{\textrm{c}}(\capacity)$ contains only elements from the optimal solution of size $\smash{\frac{c_{1}}{n}}$, i.e., we have
	\begin{align*}
		f(X_{\textrm{c}}(\capacity)) & = \capacity\cdot d\bigl(\frac{c_{1}}{n}\bigr) \overset{\eqref{eq:CRC_and_RC_relation_densities}}{=} \capacity n\cdot d_{c_{1}} \geq \capacity n\frac{1}{\rho}\cdot d_1 \overset{\eqref{eq:CRC_and_RC_relation_densities}}{=} \capacity\frac{1}{\rho}\cdot d\bigl(\frac{1}{n}\bigr) \\
		& \overset{\eqref{eq:CRC_and_RC_def_n_1}}{\geq} \capacity\frac{1}{\rho+\epsilon} \geq \frac{1}{\rho+\epsilon}\capacity\cdot d(\capacity) = \frac{1}{\rho+\epsilon}v(\capacity).
	\end{align*}
	
	Now suppose $\smash{\capacity\geq\frac{c_{1}}{n}}$.
	
	\emph{Claim:} We have
	\begin{equation}
	\frac{c_{1}}{n}\geq c_{\textrm{min}}.\label{eq:CRC_and_RC_claim}
	\end{equation}
	
	\allowdisplaybreaks[1]
	\emph{Proof of Claim}: For the sake of contradiction, assume that
	$\smash{\frac{c_{1}}{n}<c_{\textrm{min}}}$. This yields
	\begin{eqnarray}
	d_{\bigl\lfloor c_{1}+\rho\frac{v_{c_{1}}}{d_1}\bigr\rfloor} & \overset{\eqref{eq:CRC_and_RC_relation_values},\eqref{eq:CRC_and_RC_relation_densities}}{=} & d_{\bigl\lfloor c_{1}+\rho n\frac{v(\frac{c_{1}}{n})}{d(\frac{1}{n})}\bigr\rfloor}\nonumber \\
	& \overset{\eqref{eq:CRC_and_RC_def_n_1}}{\geq} & d_{\bigl\lfloor c_{1}+\rho n\frac{v(\frac{c_{1}}{n})}{1-\epsilon'}\bigr\rfloor}\nonumber \\
	& \overset{\eqref{eq:CRC_and_RC_relation_densities}}{=} & d\Bigl(\Bigl\lfloor c_{1}+\rho n\frac{v(\frac{c_{1}}{n})}{1-\epsilon'}\Bigr\rfloor/n\Bigr)/n\nonumber \\
	& \geq & d\Bigl(\frac{c_{1}}{n}+\frac{\rho}{1-\epsilon'}v\Bigl(\frac{c_{1}}{n}\Bigr)\Bigr)/n\nonumber \\
	& \geq & d\Bigl(c_{\textrm{min}}+\frac{\rho}{1-\epsilon'}v(c_{\textrm{min}})\Bigr)/n\nonumber \\
	& = & \frac{1-\epsilon'}{n}.\label{eq:CRC_and_RC_density_relation_eps-prime}
	\end{eqnarray}
	If $c_{1}=1$, we have 
	\begin{eqnarray*}
		\rho\cdot v_{c_{1}} & \overset{c_{1}=1}{=} & \rho d_1 \overset{\eqref{eq:CRC_and_RC_relation_densities}}{=} \frac{\rho}{n}d\bigl(\frac{1}{n}\bigr) \leq \frac{\rho}{n}\\
		& \overset{\eqref{eq:CRC_and_RC_def_eps-prime_2}}{<} & \frac{1-\epsilon'}{n}\lfloor1+\rho\rfloor\\
		& \overset{\eqref{eq:CRC_and_RC_density_relation_eps-prime}}{\leq} & \lfloor1+\rho\rfloor d_{\bigl\lfloor c_{1}+\rho\frac{v_{c_{1}}}{d_1}\bigr\rfloor}\\
		& \overset{c_{1}=1}{=} & \Bigl\lfloor c_{1}+\rho\frac{v_{c_{1}}}{d_1}\Bigr\rfloor d_{\bigl\lfloor c_{1}+\rho\frac{v_{c_{1}}}{d_1}\bigr\rfloor}\\
		& = & v_{\bigl\lfloor c_{1}+\rho\frac{v_{c_{1}}}{d_1}\bigr\rfloor}.
	\end{eqnarray*}
	Otherwise, if $c_{1}\geq2$, we have
	\begin{eqnarray*}
		\rho\cdot v_{c_{1}} & = & \rho\cdot c_{1}d_{c_{1}} \overset{\eqref{eq:CRC_and_RC_relation_densities}}{=} \rho\frac{c_{1}}{n}d\bigl(\frac{c_{1}}{n}\bigr) \leq \rho\frac{c_{1}}{n}\\
		& = & (1-\epsilon')\frac{c_{1}}{n}+(\rho-1+\epsilon')\frac{c_{1}}{n}\\
		& \overset{\eqref{eq:CRC_and_RC_def_eps-prime_1}}{<} & (1-\epsilon')\frac{c_{1}}{n}+\Bigl((1-\epsilon')\rho-\frac{1}{2}\Bigr)(1-\epsilon')\frac{c_{1}}{n}\\
		& \overset{c_{1}\geq2}{\leq} & (1-\epsilon')\frac{c_{1}}{n}+\Bigl((1-\epsilon')\rho-\frac{1}{c_{1}}\Bigr)(1-\epsilon')\frac{c_{1}}{n}\\
		& = & (c_{1}+(1-\epsilon')\rho c_{1}-1)\frac{1-\epsilon'}{n}\\
		& \leq & \lfloor c_{1}+(1-\epsilon')\rho c_{1}\rfloor\frac{1-\epsilon'}{n}\\
		& \overset{\eqref{eq:CRC_and_RC_density_relation_eps-prime},c_{1}\leq\bigl\lfloor c_{1}+\rho\frac{v(c_{1})}{d(1)}\bigr\rfloor}{\leq} & \lfloor c_{1}+\rho c_{1} d_{c_{1}}n\rfloor\frac{1-\epsilon'}{n}\\
		& = & \lfloor c_{1}+\rho v_{c_{1}}n\rfloor\frac{1-\epsilon'}{n}\\
		& \leq & \Bigl\lfloor c_{1}+\rho\frac{v_{c_{1}}}{d(\frac{1}{n})}n\Bigr\rfloor\frac{1-\epsilon'}{n}\\
		& \overset{\eqref{eq:CRC_and_RC_relation_densities}}{=} & \Bigl\lfloor c_{1}+\rho\frac{v_{c_{1}}}{d_1}\Bigr\rfloor\frac{1-\epsilon'}{n}\\
		& \overset{\eqref{eq:CRC_and_RC_density_relation_eps-prime}}{\leq} & \Bigl\lfloor c_{1}+\rho\frac{v_{c_{1}}}{d_1}\Bigr\rfloor d_{\bigl\lfloor c_{1}+\rho\frac{v_{c_{1}}}{d_1}\bigr\rfloor}\\
		& = & v_{\bigl\lfloor c_{1}+\rho\frac{v_{c_{1}}}{d_1}\bigr\rfloor}.
	\end{eqnarray*}
	In either case, we have
	\[
	\rho\cdot v_{c_{1}}<v_{\bigl\lfloor c_{1}+\rho\frac{v_{c_{1}}}{d_1}\bigr\rfloor}=v_{\bigl\lfloor c_{1}+\rho\frac{v_{c_{1}}}{v_1}\bigr\rfloor}.
	\]
	Since $\smash{d_{c_{2}}<\frac{v_1}{\rho}}$, the value of the solution $X \Bigl(\smash{\Bigl\lfloor c_{1}+\rho\frac{v_{c_{1}}}{v_1}\Bigr\rfloor}\Bigl)$ is $\smash{\max\bigl\{ v_{c_{1}},\rho\frac{v_{c_{1}}}{v_1}d_{c_{2}}\bigr\}=v_{c_{1}}}$, i.e., $X$ is not $\rho$-competitive, which is a contradiction.
	Thus, the claim holds, and we have $\smash{\frac{c_{1}}{n}\geq c_{\textrm{min}}}$.
	
	For the value of the solution $X_{\textrm{c}}$, we have
	\begin{eqnarray*}
		f_{\textrm{c}}(X_{\textrm{c}}(\capacity)) & \geq & f_{\textrm{c}}\Bigl(X_{\textrm{c}}\Bigl(\frac{\lfloor n\capacity\rfloor}{n}\Bigr)\Bigr)\\
		& \overset{\eqref{eq:CRC_and_RC_relation_values}}{=} & f(X(\lfloor n\capacity\rfloor))\\
		& \geq & \frac{1}{\rho}v_{\lfloor n\capacity\rfloor}\\
		& = & \frac{1}{\rho}\lfloor n\capacity\rfloor d_{\lfloor n\capacity\rfloor}\\
		& \overset{\eqref{eq:CRC_and_RC_relation_densities}}{=} & \frac{1}{\rho}\cdot\frac{\lfloor n\capacity\rfloor}{n}d\Bigl(\frac{\lfloor n\capacity\rfloor}{n}\Bigr)\\
		& \geq & \frac{1}{\rho}\cdot\frac{\lfloor n\capacity\rfloor}{n}d(\capacity)\\
		& \geq & \frac{1}{\rho}\cdot\frac{n\capacity-1}{n\capacity}\capacity\cdot d(\capacity)\\
		& \overset{\eqref{eq:CRC_and_RC_claim}}{\geq} & \frac{1}{\rho}\cdot\frac{nc_{\textrm{min}}-1}{nc_{\textrm{min}}}\capacity\cdot d(\capacity)\\
		& \overset{\eqref{eq:CRC_and_RC_def_n_2}}{\geq} & \frac{1}{\rho+\epsilon}\capacity\cdot d(\capacity)\\
		& = & \frac{1}{\rho+\epsilon}v(\capacity),
	\end{eqnarray*}
	i.e., the solution $X_{\textrm{c}}$ is $(\rho+\epsilon)$-competitive.
	By choosing $\epsilon>0$ arbitrarily small, the statement follows.
\end{proof}

This proposition implies that instead of devising a lower bound for the \incmaxsep problem, we can construct a lower bound for the \contincmax problem.

Note that it is not clear whether the competitive ratio of \incmaxsep and \contincmax coincide.
This is due to the fact that a solution to the \contincmax problem may add fractional elements while a solution to the \incmaxsep problem may only add an integral number of items.
There are even discrete instances where every continuization of the instance has a competitive ratio smaller than the initial instance.

\begin{observation}
	There exists an instance of \incmaxsep that has a competitive ratio that is strictly larger than that of every instance of \contincmax that monotonically interpolates the \incmaxsep instance.
\end{observation}

\begin{proof}
	Consider the instance of \incmaxsep with $N=16$ sets and
	\begin{eqnarray*}
		d_{1} & = & 1,\\
		d_{3}=d_{4} & = & \frac{17}{40},\\
		d_{12}=d_{13}=d_{14}=d_{15}=d_{16} & = & \frac{16473}{107200}.
	\end{eqnarray*}
	For $i\in\{2,5,6,7,8,9,10,11\}$, the density is chosen such that $i\cdot d_{i}=v_{i}=v_{i-1}=(i-1)d_{i-1}$.
	We show that every incremental solution $(c_{1},c_{2},\dots)$ has a competitive ratio of at least $1.446$ for this problem instance.
	If $c_{1}\geq2$, then $d_{c_{1}}\leq\frac{1}{2}$, i.e., for size $1$, the solution has value $d_{c_{1}}\leq\frac{1}{2}$ which implies that the solution has a competitive ratio of at least $2$.
	Thus assume that $c_{1}=1$. If $c_{2}\geq5$, we can, without loss of generality, assume that $c_{2}\geq12$. Otherwise we can improve the solution by choosing $c_{2}=4$ instead.
	Then, the value of the solution for size~$4$ is\linebreak $\max\{1,3\cdot d_{c_{2}}\}=\max\{1,3\cdot\frac{16473}{107200}\}=1$, while the optimal solution has value $4d_{4}=\frac{17}{10}$, i.e., the competitive ratio of the solution is at least $1.7$.
	Without loss of generality, we can assume that $v_{c_{2}}>v_{c_{1}}$, i.e., that $c_{2}\geq3$.
	It remains to consider the case that $c_{2}\in\{3,4\}$.
	We can assume that $d_{c_{3}}<d_{c_{2}}$ because otherwise, we could improve the solution by removing $c_{2}$.
	Thus, and because $v_{c_{3}}>v_{c_{2}}$, we have $c_{3}\geq12$, i.e., $d_{c_{3}}=\frac{16473}{107200}$.
	For size $4c_{2}$, the value of the solution is $\max\{\frac{17}{40}c_{2},(4c_{2}-1-c_{2})\frac{16473}{107200}\}=\frac{17}{40}c_{2}$ and the optimal solution has value at least $4c_{2}\cdot\frac{16473}{107200}$.
	Thus, the competitive ratio is $4c_{2}\cdot\frac{16473}{107200}/(\frac{17}{40}c_{2})=\frac{969}{670}>1.446$.
	
	Now, we consider an instance of \contincmax with $d(i)=d_{i}$ for all $i\in\{1,\dots,16\}$. Let $\rho=\frac{57}{40}=1.425$.
	We show that the incremental solution $(c_{1},c_{2},c_{3})=(\frac{1}{\rho},4,12-\frac{1}{\rho})$ is $\rho$-competitive.
	Note that $c_{1}\geq\frac{1}{\rho}$.
	By Lemma~\ref{lem:competitive_equivalence}, it suffices to show that
	\begin{equation}\label{eq:sep_harder_than_cont_requirement}
	d(c_{i})\geq\frac{v(c_{i-1})}{p(c_{i-1})-\sum_{j=1}^{i-1}c_{j}}
	\end{equation}
	for $i\in\{2,3\}$.
	We have $p(c_{1})\leq\frac{\rho v(c_{1})}{d_{3}}=\frac{57}{17}v(c_{1})$ and thus
	\[
	d(c_{2})=d(4)=d_{4}=\frac{17}{40}=\frac{1}{\frac{57}{17}-1}\overset{d(c_{1})\geq d_{1}=1}{\geq}\frac{1}{\frac{57}{17}-\frac{1}{d(c_{1})}}\overset{v(c_{1})=c_{1}d(c_{1})}{=}\frac{v(c_{1})}{p(c_{1})-c_{1}}.
	\]
	We have $p(c_{2})=p(4)=\frac{\rho v(4)}{d_{12}}=\frac{\frac{57}{40}\cdot4\cdot\frac{17}{40}}{\frac{16473}{107200}}=\frac{268}{17}$ and thus
	\[
	d(c_{3})\overset{c_{3}<12}{\geq}d_{12}=\frac{16473}{107200}=\frac{4\cdot\frac{17}{40}}{\frac{268}{17}-\frac{40}{57}-4}=\frac{v(c_{2})}{p(c_{2})-c_{1}-c_{2}}.
	\]
	Therefore, \eqref{eq:sep_harder_than_cont_requirement} holds for $i\in\{2,3\}$, i.e., the solution $(c_{1},c_{2},c_{3})$ is $\rho$-competitive.
\end{proof}

Note that, even though this shows that there are instances where the continuous problem is easier than the discrete one, this does not rule out that the competitive ratios of \incmaxsep and \contincmax coincide.
This is due to the fact that the instance in the proof is not a worst-case instance.

\subsection{Optimal Continuous Online Algorithm}\label{sec:optContOnlineAlg}

In this section, we present an algorithm to solve the \contincmax problem, and analyze it.
To get an idea what the algorithm does, consider the following lemma.
It gives a characterization of a solution $(c_{1},c_{2},\dots)$ being $\rho$-competitive, depending on $(c_{1},c_{2},\dots)$, $v$ and~$d$.

\begin{lemma}\label{lem:competitive_equivalence}
	For a solution $(c_{1},c_{2},\dots)$ for an instance of the \contincmax problem, the following statements are equivalent:
	
	\begin{enumerate}
		\item[(i)] $(c_{1},c_{2},\dots)$ is $\rho$-competitive.
		
		\item[(ii)] We have $d(c_{1})\geq\frac{1}{\rho}$ and, for all $i\in\N$, $d(c_{i+1})\geq\frac{v(c_{i})}{p(c_{i})-\sum_{j=1}^{i}c_{j}}$.
		
		\item[(iii)] We have $d(c_{1})\geq\frac{1}{\rho}$, and, for all $i\in\N$, $p(c_{i})>\sum_{j=1}^{i}c_{j}$ and $\smash{d(c_{i+1})\geq\frac{v(c_{i})}{p(c_{i})-\sum_{j=1}^{i}c_{j}}}$.
	\end{enumerate}
\end{lemma}

\begin{proof}
	$(i)\Rightarrow(iii)$: Since $X:=(c_1,c_2,\dots)$ is $\rho$-competitive, for any size $\capacity\geq0$, we have\linebreak $f(X(\capacity))\geq v(\capacity)/\rho$.
	If $\smash{d(c_{1})<\frac{1}{\rho}}$ was true,~$X$ would not be $\rho$-competitive for all sizes $\capacity\geq0$ with $d(\capacity)>\rho d(c_{1})$, which exists because $d(0)=1$.
	Thus, we have $\smash{d(c_{1})\geq\text{\ensuremath{\frac{1}{\rho}}}}$.
	Now\linebreak suppose $\smash{p(c_{i})<\sum_{j=1}^{i}c_{j}}$.
	By definition of $p$ and monotonicity of $v$, we know that\linebreak $\smash{\rho v(c_{i})=v(p(c_{i}))<v\bigl(\sum_{j=1}^{i}c_{j}\bigr)}$ which means that~$X$ is not $\rho$-competitive for size $\smash{\sum_{j=1}^{i}c_{i}}$.
	This is a contradiction and thus we have $\smash{p(c_{i})\geq\sum_{j=1}^{i}c_{j}}$.
	Suppose $\smash{p(c_{i})=\sum_{j=1}^{i}c_{j}}$.
	Let $\smash{x\in(0,\frac{v(c_{i})}{d(c_{i+1})})}$.
	The value of solution~$X(p(c_{i})+x)$ is $\smash{f(X(p(c_{i})+x))=v(c_{i})}$ but, by definition of $p$ and monotonicity of $v$, we know $\smash{\rho v(c_{i})=v(p(c_{i}))<v\bigl(\bigl(\sum_{j=1}^{i}c_{j}\bigr)+x\bigr)}$, and thus we have $\smash{p(c_{i})>\sum_{j=1}^{i}c_{j}}$.
	Assume, we have $d(c_{i+1})<\frac{v(c_{i})}{p(c_{i})-\sum_{j=1}^{i}c_{j}}$.
	We established $p(c_{i})>\sum_{j=1}^{i}c_{j}$ and because\linebreak $\bigl(p(c_{i})-\sum_{j=1}^{i}c_{j}\bigr)d(c_{i+1})<v(c_{i})$, the value of solution~$X(p(c_{i}))$ is $f(X(p(c_{i})))=v(c_{i})=v(p(c_{i}))/\rho$.
	Furthermore, for the same reason there is $\epsilon>0$ with $\smash{\bigl(p(c_{i})+\epsilon-\sum_{j=1}^{i}c_{j}\bigr)d(c_{i+1})<v(c_{i})}$.
	This implies that the value of solution~$X(p(c_{i})+\epsilon)$ is
	\[
	f(X(p(c_{i})+\epsilon))=\max\{\bigl(p(c_{i})+\epsilon-\sum_{j=1}^{i}c_{j}\bigr)d(c_{i+1}),v(c_{i})\}=v(c_{i}),
	\]
	but by definition of $p$ we know that $v(p(c_{i})+\epsilon)>\rho v(c_{i})$ and thus~$X$ is not $\rho$-competitive.
	This is a contradiction and thus $(iii)$ must hold.
	
	$(iii)\Rightarrow(i)$:
	Suppose $(iii)$ holds but~$X:=(c_1,c_2,\dots)$ was not $\rho$-competitive.
	Then there exist $\capacity\geq0$ and $\epsilon>0$ such that~$X$ is $\rho$-competitive for all sizes in $[0,\capacity]$ and not $\rho$-competitive for all sizes in $(\capacity,\capacity+\epsilon]$ because the value function~$v$ and the value of the solution~$X$ are both continuously changing with the size.
	Let $i\in\N$ and \mbox{$0<x\leq c_{i}$} such that $\smash{\capacity=\bigl(\sum_{j=1}^{i-1}c_{j}\bigr)+x}$.
	If $i=1$, the value of $X(\capacity)$ is $\smash{xd(c_{1})\geq\frac{1}{\rho}x\geq\frac{1}{\rho}xd(x)=\frac{1}{\rho}v(x)}$.
	This is a contradiction to the fact that~$X$ is not \mbox{$\rho$-competitive} for size~$\capacity$ and therefore we have $i\geq2$.
	Assume that $x=c_{i}$ holds.
	Then the value of the solution is $f(X(\capacity))=v(c_{i})$.
	For all $0<x'\leq\min\{\epsilon,\frac{v(c_{i})}{d(c_{i+1})}\}$, the value of the solution is $f(X(\capacity+x'))=v(c_{i})$ and, by definition of~$\capacity$,~$\epsilon$ and~$x'$,~$X$ is not $\rho$-competitive for size $\capacity+x'$, i.e., $v(\capacity+x')>\rho v(c_{i})$.
	Since this holds for arbitrarily small $x'>0$ and $\rho$-competitiveness for size~$\capacity$ of~$X$ implies $v(\capacity)\leq\rho v(c_{i})$, we know that $p(c_{i})=\capacity=\sum_{j=1}^{i}c_{j}$, which is a contradiction to $(iii)$ and thus $x\neq c_{i}$, i.e., $x<c_{i}$.
	Let $x'\in(x,\min\{c_{i},x+\epsilon\})$ be chosen arbitrarily and let $\capacity':=\bigl(\sum_{j=1}^{i-1}c_{j}\bigr)+x'$.
	Now suppose that 
	\begin{equation}
	xd(c_{i})<v(c_{i-1}).\label{eq:assumption_value_solution}
	\end{equation}
	Then the value of the solution is $f(X(\capacity))=v(c_{i-1})$ and since~$X$ is $\rho$-competitive for size $\capacity$, we have $\capacity\leq p(c_{i-1})$.
	But since $\capacity'>p(c_{i-1})$ for any $x'>x$, i.e., for any $\capacity'>\capacity$, we have $\capacity=p(c_{i-1})$.
	By $(iii)$ and non-negativity of $v$, we have $\smash{p(c_{i-1})-\sum_{j=1}^{i-1}c_{j}\geq0}$ and thus
	\[
	x\overset{\eqref{eq:assumption_value_solution}}{<}\frac{v(c_{i-1})}{d(c_{i})}\overset{(iii)}{\leq}p(c_{i-1})-\sum_{j=1}^{i-1}c_{j}
	\]
	or, equivalently, $\capacity=\bigl(\sum_{j=1}^{i-1}c_{j}\bigr)+x<p(c_{i-1})$.
	This is a contradiction and therefore \eqref{eq:assumption_value_solution} does not hold, i.e., we have $xd(c_{i})\geq v(c_{i-1})$ and therefore the value of the solution is $f(X(\capacity))=xd(c_{i})$.
	Since~$X$ is $\rho$-competitive for size~$\capacity$, we have
	\begin{equation}
	v(\capacity)\leq\rho\cdot xd(c_{i}).\label{eq:competitiveness_size_s}
	\end{equation}
	This implies
	\begin{equation}
	d(\capacity)=\frac{v(\capacity)}{\capacity}\overset{\eqref{eq:competitiveness_size_s}}{\leq}\rho\frac{x}{\capacity}d(c_{i})\leq\rho d(c_{i}).\label{eq:inequality_desity-s_and_density-c_i}
	\end{equation}
	We can conclude
	\begin{eqnarray*}
		v(\capacity') & = & \capacity'd(\capacity')\\
		& \overset{d\textrm{ non-inc.}}{\leq} & \capacity'd(\capacity)\\
		& = & \capacity d(s)+(\capacity'-\capacity)d(\capacity)\\
		& = & v(\capacity)+(x'-x)d(\capacity)\\
		& \overset{\eqref{eq:assumption_value_solution},\eqref{eq:inequality_desity-s_and_density-c_i}}{\leq} & \rho\cdot xd(c_{i})+(x'-x)\rho d(c_{i})\\
		& = & \rho\cdot x'd(c_{i}),
	\end{eqnarray*}
	which is a contradiction to the fact that~$X$ is not $\rho$-competitive for size~$\capacity'$ and thus $(i)$ must hold.
	
	$(iii)\Rightarrow(ii)$: This follows immediately.
	
	$(ii)\Rightarrow(iii)$: Suppose $(ii)$ holds.
	We have to show that $p(c_{i})>\sum_{j=1}^{i}c_{j}$ for all $i\in\N$.
	The rest of $(iii)$ follows immediately from $(ii)$.
	We will prove that this is the case by induction on $i$.
	For $i=1$ we have $p(c_{1})>c_{1}$ by definition of $p$, continuity of $v$ and the fact that $c_{1}>0$.
	Now suppose
	\begin{equation}
	p(c_{i})>\sum_{j=1}^{i}c_{j}\label{eq:lemma_competitiveness_induction-statement}
	\end{equation}
	holds for some $i\in\N$.
	If $p(c_{i})\geq\sum_{j=1}^{i+1}c_{j}$, \eqref{eq:lemma_competitiveness_induction-statement} holds for $i+1$ because $p(c_{i+1})>p(c_{i})$.
	So suppose 
	\begin{equation}
	p(c_{i})<\sum_{j=1}^{i+1}c_{j}.\label{eq:assumption_p_less_sum-to-i+1}
	\end{equation}
	In that case, we have
	\begin{eqnarray*}
		& & v\Bigl(\sum_{j=1}^{i+1}c_{j}\Bigr)\\
		& = & \Bigl(\sum_{j=1}^{i+1}c_{j}\Bigr)\cdot d\Bigl(\sum_{j=1}^{i+1}c_{j}\Bigr)\\
		& \overset{\eqref{eq:assumption_p_less_sum-to-i+1},d\textrm{ non-inc.}}{\leq} & \Bigl(\sum_{j=1}^{i+1}c_{j}\Bigr)\cdot d(p(c_{i}))\\
		& = & p(c_{i})d(p(c_{i}))+\Bigl(\Bigl(\sum_{j=1}^{i+1}c_{j}\Bigr)-p(c_{i})\Bigr)d(p(c_{i}))\\
		& = & v(p(c_{i}))+\Bigl(\Bigl(\sum_{j=1}^{i+1}c_{j}\Bigr)-p(c_{i})\Bigr)d(p(c_{i}))\\
		& \overset{\textrm{def of }p}{=} & \rho v(c_{i})+\Bigl(\Bigl(\sum_{j=1}^{i+1}c_{j}\Bigr)-p(c_{i})\Bigr)\frac{\rho v(c_{i})}{p(c_{i})}\\
		& \overset{\eqref{eq:lemma_competitiveness_induction-statement},c_{1}>0}{<} & \rho\cdot\Bigl(p(c_{i})-\sum_{j=1}^{i}c_{j}\Bigr)\frac{v(c_{i})}{p(c_{i})-\sum_{j=1}^{i}c_{j}}+\Bigl(\Bigl(\sum_{j=1}^{i+1}c_{j}\Bigr)-p(c_{i})\Bigr)\frac{\rho v(c_{i})}{p(c_{i})-\sum_{j=1}^{i}c_{j}}\\
		& = & \rho\cdot c_{i+1}\frac{v(c_{i})}{p(c_{i})-\sum_{j=1}^{i}c_{j}}\\
		& \overset{(ii)}{\leq} & \rho\cdot c_{i+1}d(c_{i+1})\\
		& = & \rho v(c_{i+1}).
	\end{eqnarray*}
	Since $v$ is increasing and continuous, this implies that $p(c_{i+1})>\sum_{j=1}^{i+1}c_{j}$.
\end{proof}

The intuition behind the fraction
\[
	\frac{v(c_{i})}{p(c_{i})-\sum_{j=1}^{i}c_{j}}
\]
is the following:
The value of the solution $(c_{1},\dots,c_{i-1},c_{i})$ is $v(c_{i})$ and this value is \mbox{$\rho$-competitive} up to size $p(c_{i})$.
The size required for this solution is $\smash{\sum_{j=1}^{i}c_{j}}$.
Thus, in order to stay competitive, the size added next, namely $c_{i+1}$, needs to be chosen such that $\smash{\bigl(p(c_{i})-\sum_{j=1}^{i}c_{j}\bigr)d(c_{i+1})\geq v(c_{i})}$, i.e., the density $d(c_{i+1})$ is large enough such that the value of the solution of size $p(c_{i})$ is\linebreak $\bigl(p(c_{i})-\sum_{j=1}^{i}c_{j}\bigr)d(c_{i+1})$.

We use this fraction to define an algorithm for solving the \contincmax Problem.
For the algorithm, we assume that $v$ is strictly increasing and $d$ is strictly decreasing to make the choice of our algorithm unique.
Every instance of \contincmax can be transformed to satisfy this with an arbitrarily small loss by simpliy ``tilting'' constant parts of~$d$ and~$v$ by a small amount.
The algorithm $\greedycap(c_{1},\rho)$ starts by adding the optimal solution of size $c_{1}>0$ and chooses the size $c_{i+1}$ such that
\begin{equation}
d(c_{i+1})=\frac{v(c_{i})}{p(c_{i})-\sum_{j=1}^{i}c_{j}},\label{eq:GreedyCap_def}
\end{equation}
i.e., as large as possible while still satisfying the inequality in Lemma~\ref{lem:competitive_equivalence}.
An illustration of the algorithm can be found in Figure~\ref{fig:greedyCap}.

Using the definition of the algorithm in~\eqref{eq:GreedyCap_def} and Lemma~\ref{lem:competitive_equivalence}, we are able to prove the following.

\begin{figure}
	\begin{center}
		\begin{tikzpicture}[scale=1.1]
		\begin{axis}[
		width=11cm,
		height=7cm,
		axis lines=middle,
		xtick={7928,20440,33891,41820},
		ytick={193,522},
		xticklabels={$\sum_{j=1}^i c_j$,$p(c_i)$,$c_{i+1}$,$\sum_{j=1}^{i+1} c_j$},
		yticklabels={$v(c_i)$,$v(c_{i+1})$},
		scaled x ticks=false,
		domain=0:64000
		]
		
		\addplot[samples=2,domain=63999:64000,white] {650};
		\node at (axis cs:58000,30){$c$};
		
		\node at (axis cs:48500,620){\small $v(c)$};
		\node at (axis cs:49100,307){\small $\frac{1}{\rho}v(c)$};
		\node at (axis cs:54000,493){\small $\greedycap(c)$};
		
		\addplot[samples=100,domain=0:45000] {x^0.6};
		\addplot[samples=100,domain=0:45000] {x^0.6/2};
		
		\addplot[thick,samples=2,domain=1:3] {1};
		\addplot[thick,samples=2,domain=3:8] {1/(3-2)*(x-1)};
		\addplot[thick,samples=2,domain=8:22] {3};
		\addplot[thick,samples=2,domain=22:49] {3/(22-8)*(x-8)};
		\addplot[thick,samples=2,domain=49:130] {9};
		\addplot[thick,samples=2,domain=130:275] {9/(130-49)*(x-49)};
		\addplot[thick,samples=2,domain=275:718] {26};
		\addplot[thick,samples=2,domain=718:1490] {26/(718-275)*(x-275)};
		\addplot[thick,samples=2,domain=1490:3858] {71};
		\addplot[thick,samples=2,domain=3858:7928] {71/(3858-1490)*(x-1490)};
		\addplot[thick,samples=2,domain=7928:20440] {193};
		\addplot[thick,samples=2,domain=20440:41820] {193/(20440-7928)*(x-7928)};
		\addplot[thick,samples=2,domain=41820:45000] {522};
		
		\addplot[dashed,samples=2,domain=7928:20440] {193/(20440-7928)*(x-7928)};
		
		\addplot[dotted] coordinates {(7928, 0) (7928, 193)};
		\addplot[dotted] coordinates {(20440, 0) (20440, 193)};
		\addplot[dotted] coordinates {(33891, 0) (33891, 522)};
		\addplot[dotted] coordinates {(33891, 522) (41820, 522)};
		\addplot[dotted] coordinates {(41820, 0) (41820, 522)};
		\end{axis}
		\end{tikzpicture}
	\end{center}
	\caption{Illustration how $\greedycap(c_1,\rho)$ works. Between size $\smash{\sum_{j=1}^{i} c_j}$ and size $\smash{\sum_{j=1}^{i+1} c_j}$, the algorithm adds the optimal solution of size~$c_{i+1}$. This size is chosen in a way that the value of the partially added solution has value $v(c_i)$ exactly at size $p(c_i)$, i.e., when the previously added solution of size~$c_i$ loses $\rho$-competitiveness.}\label{fig:greedyCap}
\end{figure}
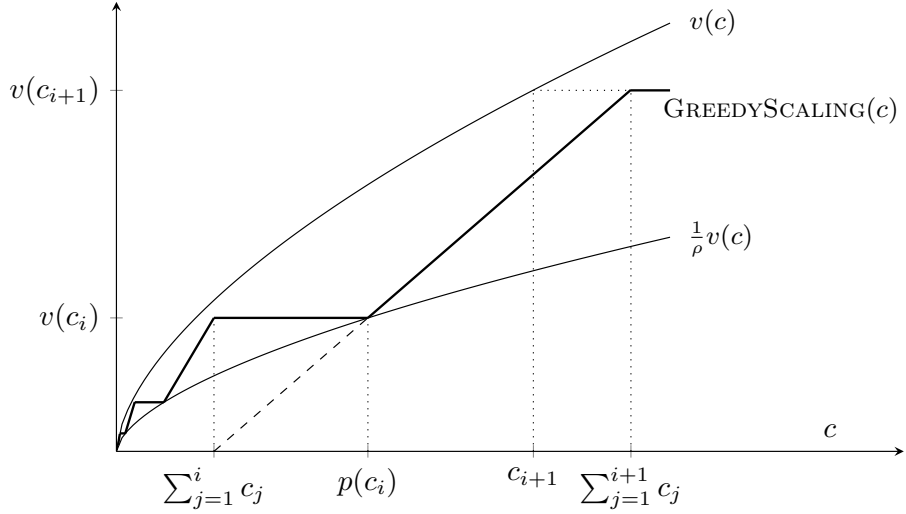

\begin{proposition}\label{prop:GreedyCap_rho-comp_equivalence}
	The algorithm $\greedycap(c_{1},\rho)$ is $\rho$-competitive if and only if it produces a solution $(c_{1},c_{2},\dots)$ with $c_{i}<c_{i+1}$ for all $i\in\N$ and $d(c_{1})\geq\frac{1}{\rho}$.
\end{proposition}

\begin{proof}
	``$\Leftarrow$'': If $c_{i}<c_{i+1}$ for all $i\in\N$ and $d(c_{1})\geq\frac{1}{\rho}$, we can simply apply Lemma~\ref{lem:competitive_equivalence} and obtain that the solution is $\rho$-competitive.
	
	``$\Rightarrow$'': If $d(c_{1})<\frac{1}{\rho}$, Lemma~\ref{lem:competitive_equivalence} yields that the solution is not $\rho$-competitive.
	Now, suppose that $c_{k+1} \leq c_{k}$ for some $k\in\N$.
	If $c_i\leq0$ for some $i\in\N$, then the solution is not valid and thus not $\rho$-competitive.
	Thus, assume that $c_i>0$ for all $i\in\N$.
	We will now iteratively show that, for all $i\in\{k,k+1,\dots\}$, we have $c_{i+1}< c_i$.
	For this, suppose that $c_{i+1}\leq c_i$.
	Then
	\begin{eqnarray*}
		d(c_{i+2}) &=& \frac{v(c_{i+1})}{p(c_{i+1})-\sum_{j=1}^{i+1}c_j} \\
		&=& \frac{1}{\frac{\rho}{d(p(c_{i+1}))}-\frac{1}{v(c_{i+1})}\sum_{j=1}^{i+1}c_j} \\
		&\overset{c_{i+1}\leq c_i}{\geq}& \frac{1}{\frac{\rho}{d(p(c_i))}-\frac{1}{v(c_i)}\sum_{j=1}^{i+1}c_j} \\
		&>& \frac{1}{\frac{\rho}{d(p(c_i))}-\frac{1}{v(c_i)}\sum_{j=1}^{i}c_j} \\
		&=& \frac{v(c_{i})}{p(c_{i})-\sum_{j=1}^{i}c_j} = d(c_{i+1}).
	\end{eqnarray*}
	Because~$d$ is non-increasing, we have $c_{i+2}< c_{i+1}$.
	By an iterative argument it follows that, for all $i\in\{k,k+1,\dots\}$, we have $c_{i+1}< c_i$.
	This implies that the value of the solution $(c_1,c_2,\dots)$ is smaller or equal to $v(c_k)$ for all sizes.
	Yet, for large sizes $\capacity\in\N$, we have $v(\capacity)>\rho v(c_k)$ as $\lim_{c\rightarrow\infty}v(c)=\infty$.
\end{proof}

To show that $\greedycap(c_1,\rho)$ with the correct choice of~$c_1$ and~$\rho$ computes the best-possible solution, we need the following lemma.
It states that if it is possible to find a $\rho$-approximation algorithm for the problem, then, for any $k\in\N$, there exists an algorithm such that we cannot reduce any of the sizes up to size~$k$ in the solution produced by the algorithm without losing $\rho$-competitiveness.
The idea behind its proof is the following.
We start with some $\rho$-competitive solution and iteratively reduce the sizes until we converge to some algorithm and show that this algorithm satisfies the sought property.
To avoid running into some algorithm\linebreak represented by a sequence starting with 0's, we start our search with an algorithm with a minimal number of chosen sizes up to size~$k$.

\begin{lemma}\label{lem:fastest_ALG}
	Let $v$ be strictly increasing, let $\rho$ be an achievable approximation ratio, and let $k\in\N$.
	Then there exists a $\rho$-competitive solution $(c_{1}^{*},c_{2}^{*},\dots)$ such that there exists no other $\rho$-competitive solution $(c_{1}',c_{2}',\dots)$ with $c_{i}'\leq c_{i}^{*}$ for all $i\in\{1,\dots,n-1\}$, \mbox{$c_{i}'=c_{i}^{*}$} for all $i\in\{n,n+1,\dots\}$ and $\smash{\sum_{i=1}^{n}c_{i}'<\sum_{i=1}^{n}c_{i}^{*}}$, where $\smash{n:=\min\{\ell\in\N\mid\sum_{i=1}^{\ell}c_{i}^{*}\geq k\}}$.
	Furthermore, there is some $\capacity>0$ that is independent from~$k$ such that~$c_1^*>\capacity$.
\end{lemma}

\begin{proof}
	We fix a $\rho$-competitive solution $(c_{1},c_{2},\dots)$ with \mbox{$c_{1}<c_{2}<\dots$}.
	We define\linebreak $\smash{n:=\min\{\ell\in\N\mid\sum_{i=1}^{\ell}c_{i}\geq k\}}$.
	For $i\in\N$, let $\mathcal{S}_{i}$ be the set of all $\rho$-competitive solutions 	$(c_{1}',c_{2}',\dots)$ with $c_{n+j}=c_{i+j}'$ for all $j\in\N\cup\{0\}$.
	Let $m:=\min\{i\in\N\mid\mathcal{S}_{i}\neq\emptyset\}$.
	This value is well-defined since $(c_{1},c_{2},\dots)\in\mathcal{S}_{n}$.
	For every solution $(c_{1}',c_{2}',\dots)\in\mathcal{S}_{m}$, we have
	\begin{equation}
	0<c_{1}'<c_{2}'<\dots<c_{m}'=c_{n}\label{eq:sequence_inequalities}
	\end{equation}
	because otherwise it would be possible to skip a size, which is a contradiction to the minimality of~$m$.
	For every $(c_{1}',c_{2}',\dots)\in\mathcal{S}_{m}$, we define the set
	\[
	\mathcal{S}_{m}((c_{1}',c_{2}',\dots)):=\{(c_{1}'',c_{2}'',\dots)\in\mathcal{S}_{m}\mid\forall i\in\N:c_{i}''\leq c_{i}'\}.
	\]
	To prove the lemma, it suffices to show that there exists some solution $(c_{1}^{*},c_{2}^{*},\dots)$ with
	\begin{equation}
	\mathcal{S}_{m}((c_{1}^{*},c_{2}^{*},\dots))=\{(c_{1}^{*},c_{2}^{*},\dots)\}.\label{eq:lex_small_alg_sufficient_condition}
	\end{equation}
	It is easy to see that, for every $(c_{1}'',c_{2}'',\dots)\in\mathcal{S}_{m}((c_{1}',c_{2}',\dots))$,	we have
	\begin{equation}
	\mathcal{S}_{m}((c_{1}'',c_{2}'',\dots))\subseteq\mathcal{S}_{m}((c_{1}',c_{2}',\dots)).\label{eq:sequence-sets_downward_closed}
	\end{equation}
	For $(c_{1}',c_{2}',\dots)\in\mathcal{S}_m$, we define
	\[
	s_{m}((c_{1}',c_{2}',\dots)):=\inf\Bigl\{\sum_{i=1}^{m}c_{i}''\Bigm|(c_{1}'',c_{2}'',\dots)\in\mathcal{S}_{m}((c_{1}',c_{2}',\dots))\Bigr\}.
	\]
	Since \eqref{eq:sequence_inequalities} holds, this value is larger than $0$ and smaller than $\sum_{j=1}^{n}c_{j}$, and is therefore well-defined.
	
	We fix some solution $(c_{1}^{1},c_{2}^{1},\dots)\in\mathcal{S}_{m}$ and recursively define a sequence of solutions such that, for all $i,j\in\N$, we have $(c_{1}^{j+1},c_{2}^{j+1},\dots)\in\mathcal{S}_{m}((c_{1}^{j},c_{2}^{j},\dots))$, and such that
	\begin{equation}
	\bigl(\sum_{i=1}^{m}c_{i}^{j+1}\bigr)-s_{m}((c_{1}^{j},c_{2}^{j},\dots))\leq\Bigl(\frac{1}{2}\Bigr)^{j}.\label{eq:sequence_sum_convergence}
	\end{equation}
	This sequence is well-defined because the infimum can be approximated arbitrarily close.
	
	\emph{Claim 1:} The limit $(c_{1}^{*},c_{2}^{*},\dots):=\lim_{j\rightarrow\infty}(c_{1}^{j},c_{2}^{j},\dots)$ exists.
	
	\emph{Proof of Claim 1:} The sequence $\bigl(s_{m}((c_{1}^{j},c_{2}^{j},\dots))\bigr)_{j\in\N}$	is increasing because of \eqref{eq:sequence-sets_downward_closed}, and the sequence $\smash{\bigl(\sum_{i=1}^mc_{i}^{j}\bigr)_{j\in\N}}$ is decreasing because $(c_{1}^{j+1},c_{2}^{j+1},\dots)\in\mathcal{S}_{m}((c_{1}^{j},c_{2}^{j},\dots))$.
	Furthermore, $\smash{s_{m}((c_{1}^{j},c_{2}^{j},\dots))\leq\sum_{i=1}c_{i}^{j}}$ and we have $\smash{\lim_{j\rightarrow\infty}s_{m}((c_{1}^{j},c_{2}^{j},\dots))=\lim_{j\rightarrow\infty}\sum_{i=1}c_{i}^{j}}$ because of~\eqref{eq:sequence_sum_convergence}.
	Because $c_{i}^{j+1}\leq c_{i}^{j}$ for all $i,j\in\N$, the sequence $\bigl((c_{1}^{j},c_{2}^{j},\dots)\bigr)_{j\in\N}$ converges to some solution $(c_{1}^{*},c_{2}^{*},\dots)$ with
	\begin{equation}
	\sum_{i=1}^{m}c_{i}^{*}=\lim_{j\rightarrow\infty}\sum_{i=1}^{m}c_{i}^{j}=\lim_{j\rightarrow\infty}s_{m}((c_{1}^{j},c_{2}^{j},\dots)).\label{eq:sequence_of_sequences_limit_sum}
	\end{equation}
	
	\emph{Caim 2:} We have $(c_{1}^{*},c_{2}^{*},\dots)\in\mathcal{S}_m$.
	
	\emph{Proof of Claim 2:}
	For every $j\in\N$, we have $c_{m+\ell}^{j}=c_{n+\ell}$ for all $\ell\in\N\cup\{0\}$ and therefore also $c_{m+\ell}^{*}=c_{n+\ell}$ for all $\ell\in\N\cup\{0\}$.
	So, it remains to prove that the solution $(c_{1}^{*},c_{2}^{*},\dots)$ is $\rho$-competitive.
	For all $j\in\N$, we have $\smash{c_{1}^{j}\geq\frac{1}{\rho}}$ and thus $\smash{c_{1}^{*}=\lim_{j\rightarrow\infty}c_{1}^{j}\geq\frac{1}{\rho}}$.
	Next, we show that
	\[
	d(c_{i+1}^{*})\geq\frac{v(c_{i}^{*})}{p(c_{i}^{*})-\sum_{\ell=1}^{i}c_{\ell}^{*}}
	\]
	holds for all $i\in\N$.
	The function $p$ is continuous by continuity and strict monotonicity of~$v$.
	Continuity of $v$ and $p$ imply that $\smash{v(c_{i}')/(p(c_{i}')-\sum_{\ell=1}^ic_{\ell}')}$ is continuous in $c_{1}',\dots,c_{i}'$.
	By Lemma~\ref{lem:competitive_equivalence}, we have
	\[
	d(c_{i+1}^{j})\geq\frac{v(c_{i}^{j})}{p(c_{i}^{j})-\sum_{\ell=1}^{i}c_{\ell}^{j}}
	\]
	for all $j\in\N$ because $(c_{1}^{j},c_{2}^{j},\dots)$ is $\rho$-competitive.
	Both sides of this inequality are continuous in $c_{1}^{j},c_{2}^{j},\dots$, and we have $\lim_{j\rightarrow\infty}c_{i}^{j}=c_{i}^{*}$.
	Those two facts imply that
	\[
	d(c_{i+1}^{*})\geq\frac{v(c_{i}^{*})}{p(c_{i}^{*})-\sum_{\ell=1}^{i}c_{\ell}^{*}}
	\]
	holds.
	To prove $\rho$-competitiveness of $(c_{1}^{*},c_{2}^{*},\dots)$, by Lemma \ref{lem:competitive_equivalence}, it remains to show that\linebreak $0<c_{1}^{*}<c_{2}^{*}<\dots$ holds.
	We know that this holds for all solutions $(c_{1}^{j},c_{2}^{j},\dots)$, $j\in\N$.
	Therefore, we have \mbox{$0\leq c_{1}^{*}\leq\dots\leq c^{*}<c_{+1}^{*}<\dots$}.
	But if some of these inequalities were satisfied with equality, we would be able to leave out sizes until we are left with a solution that satisfies the strict inequalities.
	This solution would need less than $m$ sizes to accumulate a solution of size $c^{*}$, i.e., $m$ would not be minimal.
	Therefore, $0<c_{1}^{*}<c_{2}^{*}<\dots$ holds.
	
	We have established $(c_{1}^{*},c_{2}^{*},\dots)\in\mathcal{S}_{m}$ and therefore $(c_{1}^{*},c_{2}^{*},\dots)\in\mathcal{S}_{m}((c_{1}^{j},c_{2}^{j},\dots))$, which implies that $\mathcal{S}_{m}((c_{1}^{*},c_{2}^{*},\dots))\subseteq\mathcal{S}_{m}((c_{1}^{j},c_{2}^{j},\dots))$.
	Combined with the fact that we have\linebreak \mbox{$(c_{1}^{*},c_{2}^{*},\dots)\in\mathcal{A}_{m}((c_{1}^{*},c_{2}^{*},\dots))$}, this implies
	\[
	\sum_{i=1}^{m}c_{i}^{*}\geq s_{m}((c_{1}^{*},c_{2}^{*},\dots))\geq\lim_{j\rightarrow\infty}s_{m}((c_{1}^{j},c_{2}^{j},\dots))\overset{\eqref{eq:sequence_of_sequences_limit_sum}}{=}\sum_{i=1}^{m}c_{i}^{*},
	\]
	i.e., $s_{m}((c_{1}^{*},c_{2}^{*},\dots))=\sum_{i=1}^{m}c_{i}^{*}$.
	Thus, \eqref{eq:lex_small_alg_sufficient_condition} holds.
	
	It remains to show that there is some~$\capacity>0$ that is independent from~$k$ such that $c_1^*>\capacity$.
	Suppose the contrary, i.e., that $c_1^*$ is not bounded from~$0$ for varying values of~$k$.
	Let $\epsilon>0$ be small enough such that $\frac{\rho}{1-\epsilon}-1\leq\rho$.
	Furthermore, let $k\in\N$ such that the above defined value~$c_1^*$ satisfies $d(p(c_1^*))\geq 1-\epsilon$, which is possible because $d(0)=1$,~$d$ is continuous, and~$c_1^*$ is not bounded from~$0$.
	By minimality of~$m$ and Lemma~\ref{lem:competitive_equivalence}, we have
	\[
	\frac{1}{\rho} > d(c_{2}^{*}) \geq \frac{v(c_1^*)}{p(c_1^*)-c_1^*} = \frac{1}{\frac{\rho}{d(p(c_1^*))}-\frac{1}{d(c_1^*)}} \geq \frac{1}{\frac{\rho}{1-\epsilon}-1} \geq \frac{1}{\rho}
	\]
	which is a contradiction.
	Thus, $c_1^*$ is bounded from~$0$.
\end{proof}

Using this lemma, we can show that $\greedycap(c_1,\rho)$ for the correct choice of~$c_1$ and~$\rho$ can achieve every possible competitive ratio.
This immediately implies Theorem~\ref{thm:GreedyCap_best_possible}.

\begin{lemma}
	Let $v$ be strictly increasing and $d$ be strictly decreasing, and let $\rho$ an achievable approximation ratio.
	There exists a starting value $\smash{c_{1}^{*}\in\bigl[d^{-1}\bigl(\frac{\rho-1}{\rho}\bigr),d^{-1}(\frac{1}{\rho})\bigr]}$ such that the algorithm $\greedycap(c_{1}^{*},\rho)$ is $\rho$-competitive for all sizes in $\R_{\geq0}$.
\end{lemma}

\begin{proof}
	By Lemma \ref{lem:fastest_ALG}, for every $k\in\N$, there exists a $\rho$-competitive solution $(c_{1}^{k},c_{2}^{k},\dots)$ such that, with $n(k):=\min\{\ell\in\N\mid\sum_{i=1}^{\ell}c_{i}\geq k\}$, there exists no other $\rho$-competitive solution $(c_{1}',c_{2}',\dots)$ with $c_{i}'\leq c_{i}^{k}$ for all $i\in\{1,\dots,n(k)-1\}$, $c_{i}'=c_{i}^{k}$ for all $i\in\{n(k),n(k)+1,\dots\}$ and $\smash{\sum_{i=1}^{n(k)}c_{i}'<\sum_{i=1}^{n(k)}c_{i}^{k}}$.
	Furthermore, there is $\capacity>0$ such that $c_{1}^{k}\geq \capacity$ for all $k\in\N$.
	Without loss of generality, we can assume that \mbox{$v(c_{i+1}^{k})\geq v(c_{i}^{k})$} for all $i\in\N$.
	Because $(c_{1}^{k},c_{2}^{k},\dots)$ is $\rho$-competitive, by Lemma \ref{lem:competitive_equivalence}, we know that, for all $i\in\N$, we have
	\begin{equation}
	d(c_{i+1}^{k})\geq\frac{v(c_{i}^{k})}{p(c_{i}^{k})-\sum_{j=1}^{i}c_{j}^{k}}.\label{eq:ALG_iterative_geq-1}
	\end{equation}
	Suppose there was $i'\in\{1,\dots,n(k)-1\}$ such that~\eqref{eq:ALG_iterative_geq-1} does not hold with equality.
	By continuity of $v$, $d$ and $p$, we can find $c_{i'}'<c_{i'}^{k}$ such that
	\begin{equation}
	d(c_{i'+1}^{k})>\frac{v(c_{i'}')}{p(c_{i'}')-c_{i'}'-\sum_{j=1}^{i'-1}c_{j}}>\frac{v(c_{i'}^{k})}{p(c_{i'}^{k})-\sum_{j=1}^{i'}c_{j}^{k}}.\label{eq:smaller_sequence_rho_competitive-1}
	\end{equation}
	The solution $(c_{1}',c_{2}',\dots):=(c_{1}^{k},\dots,c_{i^{'}-1}^{k},c_{i'}',c_{i'+1}^{k},\dots)$ satisfies
	\[
	d(c_{i+1}')\geq\frac{v(c_{i}')}{p(c_{i}')-\sum_{j=1}^{i}c_{j}'}
	\]
	for all $i\in\N$.
	For $i\in\{1,\dots,i'-1\}$, this follows immediately from \eqref{eq:ALG_iterative_geq-1}, for $i=i'$, this follows from~\eqref{eq:smaller_sequence_rho_competitive-1}, and, for $i\in\{i'+1,i'+2,\dots\}$, this is due to \eqref{eq:ALG_iterative_geq-1} and the fact that $c_{i'}'<c_{i'}^{k}$.
	We have $c_{i}'\leq c_{i}^{k}$ for all $i\in\{1,\dots,n(k)-1\}$, $c_{i}'=c_{i}^{k}$ for all $i\in\{n(k),n(k)+1,\dots\}$ and $\smash{\sum_{i=1}^{n(k)}c_{i}'<\sum_{i=1}^{n(k)}c_{i}^{k}}$, which is a contradiction to our initial choice of $(c_{1}^{k},c_{2}^{k},\dots)$.
	Thus, for all $i\in\{1,\dots,n(k)-1\}$,~\eqref{eq:ALG_iterative_geq-1} holds with equality.
	
	The sequence $\bigl(c_{1}^{k}\bigr)_{k\in\N}$ is bounded since $\smash{\capacity\leq c_{1}^{k}\leq d^{-1}(\frac{1}{\rho})}$.
	Therefore, by the Bolzano-Weierstrass theorem, it contains a converging subsequence $\smash{\bigl(c_{1}^{k_{\ell}}\bigr)_{\ell\in\N}}$ with \mbox{$k_{\ell}\in\N$} and $k_{\ell+1}>k_{\ell}$ for all $\ell\in\N$.
	We define $c_{i}^{*}:=\lim_{\ell\rightarrow\infty}c_{i}^{k_{\ell}}$ for all $i\in\N$.
	We have $v(c_{1}^{*})\geq v(d^{-1}(\capacity))>0$ and $v(c_{i+1}^{*})\geq v(c_{i}^{*})$ by continuity of $v$ and because this holds for all sequences $\smash{\bigl(c_{i}^{k}\bigr)_{i\in\N}}$.
	By Lemma \ref{lem:competitive_equivalence}, we have $\smash{d(c_{1}^{k})\geq\frac{1}{\rho}}$ and
	\[
	d(c_{i+1}^{k})\geq\frac{v(c_{i}^{k})}{p(c_{i}^{k})-\sum_{j=1}^{i}c_{j}^{k}}
	\]
	for all $i,k\in\N$.
	Continuity of $d$, $v$ and $p$ yields $\smash{d(c_{1}^{*})\geq\frac{1}{\rho^{*}}}$ and
	\[
	d(c_{i+1}^{*})\geq\frac{v(c_{i}^{*})}{p(c_{i}^{*})-\sum_{j=1}^{i}c_{j}^{*}}.
	\]
	Thus, the sequence $\bigl(c_{1}^{*},c_{2}^{*},\dots\bigr)$ represents
	a $\rho$-competitive algorithm.
	It remains to show that
	\[
	c_{i+1}^{*}=d^{-1}\bigl(\frac{v(c_{i}^{*})}{p(c_{i}^{*})-\sum_{j=1}^{i}c_{j}^{*}}\bigr).
	\]
	Note that $n(k)$ increases in $k$, i.e., for all $N\in\N$, there is $K\in\N$ such that $n(k)\geq N$ for all $k\in\N$ with $k\geq K$.
	This implies that for every $i\in\N$, there exists some $K\in\N$ such that
	\[
	d(c_{i+1}^{k})=\frac{v(c_{i}^{k})}{p(c_{i}^{k})-\sum_{j=1}^{i}c_{j}^{k}}
	\]
	for all $k\in\N$ with $k\geq K$.
	Since $c_{i}^{*}=\lim_{\ell\rightarrow\infty}c_{i}^{k_{\ell}}$, the desired equality follows, i.e., the sequence $(c_{1}^{*},c_{2}^{*},\dots)$ describes the algorithm $\greedycap(c_{1}^{*},\rho)$.
\end{proof}

This result immediately implies Theorem~\ref{thm:GreedyCap_best_possible}.

For a range of starting values~$c_1$, we are able to show the upper bound on the competitive ratio of $\greedycap(c_1,\varphi+1)$ in Theorem~\ref{thm:GreedyCap_phi+1-comp}, where $\varphi=\frac{1}{2}(1+\sqrt{5})\approx1.618$ is the golden ratio.

\GreedyCapComp*

\begin{proof}
	``$\Rightarrow$'': By Lemma~\ref{lem:competitive_equivalence}, $\smash{d(c_1)\geq\frac{1}{\varphi+1}}$ holds because the algorithm is $(\varphi+1)$-competitive.
	
	``$\Leftarrow$'': Let $(c_{1},c_{2},\dots)$ be the solution produced by $\greedycap(c_{1},\varphi+1)$.
	By Proposition~\ref{prop:GreedyCap_rho-comp_equivalence}, to show $(\varphi+1)$-competitiveness, it suffices to show that  $c_{i}\leq c_{i+1}$.
	
	\emph{Claim:} We have $c_{i+1}\geq(\varphi+1)c_{i}$ for all $i\in\N$.
	
	We have \mbox{$p(c_{i})=\max\{c\geq0\mid v(c)\leq(\varphi+1)v(c_{i})\}$}.
	This implies \mbox{$v(p(c_{i}))=(\varphi+1)v(c_{i})$} by continuity of~$v$, and thus
	\begin{equation}
	p(c_{i})=\frac{v(p(c_{i}))}{d(p(c_{i}))}=\frac{(\varphi+1)v(c_{i})}{d(p(c_{i}))}\geq\frac{(\varphi+1)v(c_{i})}{d(c_{i})}=(\varphi+1)c_{i},\label{eq:p_lower_bound}
	\end{equation}
	where the inequality holds because $v(p(c_{i}))=(\varphi+1)v(c_{i})>v(c_{i})$.
	Thus, and because~$v$ is increasing, $p(c_{i})>c_{i}$ , which, implies $d(p(c_{i}))\leq d(c_{i})$ because $d$ is monotone.
	
	\emph{Proof of claim:}
	We will prove the claim by induction.
	Let $i\in\N$ and suppose the claim holds for all $j\in\{1,\dots,i-1\}$.
	Then, $\smash{c_{j}\leq\frac{1}{(\varphi+1)^{i-j}}c_{i}}$ for all $j\in\{1,\dots,i-1\}$, and therefore
	
	\begin{equation}
	\sum_{j=1}^{i}c_{j} \leq c_{i}\sum_{j=1}^{i}(\varphi+1)^{j-i} = \frac{1-(\varphi+1)^{-i}}{1-(\varphi+1)^{-1}}c_{i} < \frac{1}{1-(\varphi+1)^{-1}}c_{i} = \varphi c_{i}.\label{eq:inequality_sum_claim}
	\end{equation}
	This yields
	\begin{align*}
		d(c_{i+1}) &= \frac{v(c_{i})}{p(c_{i})-\sum_{j=1}^{i}c_{j}} \overset{\eqref{eq:inequality_sum_claim}}{<} \frac{v(c_{i})}{p(c_{i})-\varphi c_{i}} \overset{\eqref{eq:p_lower_bound}}{\leq} \frac{v(c_{i})}{p(c_{i})-\frac{\varphi}{\varphi+1}p(c_{i})} \\
		&= (\varphi+1)\frac{v(c_{i})}{p(c_{i})} = \frac{v(p(c_{i}))}{p(c_{i})} = d(p(c_{i})),
	\end{align*}
	which implies $c_{i+1}>p(c_{i})$ because $d$ is decreasing.
	Together with~\eqref{eq:p_lower_bound}, this yields the claim.
\end{proof}

Since $\greedycap(c_1,\rho)$ with the correct starting value~$c_1$ is the best-possible algorithm for a fixed instance, we can give a lower bound of $\rho>1$ for the \contincmax problem by finding an instance that is a lower bound for $\greedycap(c_1,\rho)$ with all starting values~$c_1>0$ that satisfy $d(c_1)\leq 1/\rho$.
In the following, we show that, for every countable set of starting values, there is an instance where $\greedycap(c_1,\rho)$ cannot have a competitive ratio of better than $\varphi+1$ for any of these starting values.
In order to do this, we need the following lemma.

\begin{lemma}\label{lem:recursive_sequence_becomes_negative}
	For $\alpha,\beta,\rho,\epsilon\in\R_{\geq0}$ with $\beta>0$, consider the recursively defined sequence $(t_{n})_{n\in\N}$ with
	\[
	t_{0}=\beta,\quad\quad t_{n+1}=\frac{1}{\frac{\rho}{t_{n}(1-\epsilon)}-\bigl(\sum_{j=0}^{n}\frac{(\rho+\epsilon)^{j-n}}{t_{j}}\bigr)-\frac{\alpha}{(\rho+\epsilon)^{n}}}\quad\textrm{for all }n\in\N\cup\{0\}.
	\]
	If $1<\rho<\varphi+1$, then there exists $\epsilon'>0$ such that, for all $\epsilon\in(0,\epsilon']$, there is $\ell\in\N$ with $t_{\ell}<0$.
\end{lemma}

\begin{proof}
	Rearranging terms, we obtain
	\[
	t_{0}=\beta,\quad\quad\frac{(\rho+\epsilon)^{n+1}}{t_{n+1}}=\frac{\rho(\rho+\epsilon)^{n+1}}{t_{n}(1-\epsilon)}-\Bigl(\sum_{j=0}^{n}\frac{(\rho+\epsilon)^{j+1}}{t_{j}}\Bigr)-\alpha(\rho+\epsilon)\quad\textrm{for all }n\in\N\cup\{0\}.
	\]
	We substitute $a_{n}=1/t_{n}$ for all $n\in\N$ and obtain the recursively defined sequence $(a_{n})_{n\in\N}$ with $a_{0}=1/\beta$ and
	\begin{equation}
	a_{n+1}(\rho+\epsilon)^{n+1}=a_{n}\frac{\rho}{1-\epsilon}(\rho+\epsilon)^{n+1}-\Bigl(\sum_{j=0}^{n}a_{j}(\rho+\epsilon)^{j+1}\Bigr)-\alpha(\rho+\epsilon)\quad\textrm{for all }n\in\N\cup\{0\}.\label{eq:a_n+1-recursive}
	\end{equation}
	This also implies
	\begin{equation}
	a_{n}(\rho+\epsilon)^{n}=a_{n-1}\frac{\rho}{1-\epsilon}(\rho+\epsilon)^{n}-\Bigl(\sum_{j=0}^{n-1}a_{j}(\rho+\epsilon)^{j+1}\Bigr)-\alpha(\rho+\epsilon)\quad\textrm{for all }n\in\N.\label{eq:a_n-recursive}
	\end{equation}
	Subtracting \eqref{eq:a_n-recursive} from \eqref{eq:a_n+1-recursive}, we obtain
	\[
	a_{n+1}(\rho+\epsilon)^{n+1}-a_{n}(\rho+\epsilon)^{n}=a_{n}\frac{\rho}{1-\epsilon}(\rho+\epsilon)^{n+1}-a_{n-1}\frac{\rho}{1-\epsilon}(\rho+\epsilon)^{n}-a_{n}(\rho+\epsilon)^{n+1}
	\]
	for all $n\in\N$, which yields
	\[
	a_{n+1}=a_{n}\Bigl(\frac{1}{\rho+\epsilon}+\frac{\rho}{1-\epsilon}-1\Bigr)-a_{n-1}\frac{\rho}{(1-\epsilon)(\rho+\epsilon)}
	\]
	for all $n\in\N$.
	Together with the start values $a_{0}=1/\beta$ and
	\[
	a_{1}=\frac{1}{t_{1}}=\frac{\rho}{\beta(1-\epsilon)}-\frac{1}{\beta}-\alpha
	\]
	this yields a uniquely defined linear homogeneous recurrence relation with characteristic polynomial
	\[
	0=x^{2}-\Bigl(\frac{1}{\rho+\epsilon}+\frac{\rho}{1-\epsilon}-1\Bigr)x+\frac{\rho}{(1-\epsilon)(\rho+\epsilon)}.
	\]
	Let $\smash{D(\rho,\epsilon)=\bigl(\frac{1}{2(\rho+\epsilon)}+\frac{\rho}{2(1-\epsilon)}-\frac{1}{2}\bigr)^{2}-\frac{\rho}{(1-\epsilon)(\rho+\epsilon)}}$.
	The roots of the characteristic polynomial are then
	\begin{eqnarray}
	x & = & \frac{1}{2(\rho+\epsilon)}+\frac{\rho}{2(1-\epsilon)}-\frac{1}{2}-\sqrt{D(\rho,\epsilon)},\label{eq:roots_characteristic_polynomial_recurrence_relation}\\
	y & = & \frac{1}{2(\rho+\epsilon)}+\frac{\rho}{2(1-\epsilon)}-\frac{1}{2}+\sqrt{D(\rho,\epsilon)}.\nonumber 
	\end{eqnarray}
	We claim that if $\rho<\varphi+1$, then there is $\epsilon>0$ such that $D(\rho,\epsilon)<0$.
	To see this claim, consider the function
	\[
	D(\rho,0)=\Bigl(\frac{1}{2\rho}+\frac{\rho}{2}-\frac{1}{2}\Bigr)^{2}-1.
	\]
	The function $\smash{h(\rho)=\frac{1}{2\rho}+\frac{\rho}{2}-\frac{1}{2}}$ has the derivative $\smash{h'(\rho)=-\frac{1}{2\rho^{2}}+\frac{1}{2}>0}$ for $\rho>1$.
	Thus,~$h$ is strictly increasing for $\rho\in(1,\infty)$, and, hence, $D(\rho,0)$ is also strictly increasing for $\rho\in(1,\infty)$.
	Thus, $D(\rho,0)$ has at most one root $\rho_{0}\in(1,\infty)$.
	This root satisfies
	\[
	\frac{1}{2\rho_{0}}+\frac{\rho_{0}}{2}-\frac{1}{2}=1.
	\]
	Rearranging terms yields $1+\rho_{0}^{2}=3\rho_{0}$.
	The only solution $\rho_{0}>1$ to this equation is $\varphi+1$.
	We have shown that $D(\rho,0)<0$ for all $\rho<\varphi+1$.
	Since $D(\rho,\epsilon)$ is continuous in $\epsilon$, there is $\epsilon'>0$ such that also $D(\rho,\epsilon)<0$ for all $\epsilon\in(0,\epsilon']$.
	For $\rho\in(1,\varphi+1)$ and $\epsilon$ chosen that way, we have that the roots of the characteristic polynomial \eqref{eq:roots_characteristic_polynomial_recurrence_relation} are
	distinct and complex valued.
	We then obtain that the sequence $(a_{n})_{n\in\N}$ has the closed-form expression
	\begin{equation}
	a_{n}=\lambda x^{n}+\mu y^{n}\quad\textrm{for all }n\in\N\cup\{0\},\label{eq:recurrence_relation_closed_form_1}
	\end{equation}
	where the constants $\lambda,\mu\in\C$ are chosen in such a way that the equations for the starting values
	\begin{equation}
	a_{0}=\frac{1}{\beta}=\lambda+\mu\quad\quad\textrm{and}\quad\quad a_{1}=\frac{\rho}{\beta(1-\epsilon)}-\frac{1}{\beta}-\alpha=\lambda x+\mu y\label{eq:recurrence_relation_starting_values}
	\end{equation}
	are satisfied.
	Note that by \eqref{eq:roots_characteristic_polynomial_recurrence_relation}, $x$ and $y$ are complex conjugate, and hence, by \eqref{eq:recurrence_relation_starting_values}, also $\lambda$ and $\mu$ are complex conjugate.
	We can, thus, reformulate \eqref{eq:recurrence_relation_closed_form_1} as
	\begin{eqnarray}
	a_{n} & = & \lambda x^{n}+\bar{\lambda}\bar{x}^{n}\nonumber \\
	& = & \lambda x^{n}+\bar{\lambda x^{n}}\nonumber \\
	& = & 2\mathfrak{R}(\lambda x^{n})\label{eq:recurrence_relation_closed_form_2}
	\end{eqnarray}
	for all $n\in\N$, where for the second equation we used that conjugation is distributive with multiplication and for the third equation we used that for a complex number $z\in\C$ its real part can be computed as $\smash{\mathfrak{R}(z)=\frac{z+\bar{z}}{2}}$.
	In \eqref{eq:recurrence_relation_closed_form_2}, the constant $\lambda\in\C$ satisfies $\lambda+\bar{\lambda}=1/\beta$ and $\smash{\lambda x+\bar{\lambda x}=\frac{\rho}{\beta(1-\epsilon)}-\frac{1}{\beta}-\alpha}$.
	Going to polar coordinates, we obtain
	\[
	\lambda=r_{\lambda}\exp(\mathrm{i}\varphi_{\lambda})\quad\quad\textrm{and}\quad\quad x=r_{x}\exp(\mathrm{i}\varphi_{x})
	\]
	for some $r_{\lambda},r_{x}\in\R_{\geq0}$ and some $\varphi_{\lambda},\varphi_{x}\in[0,2\pi)$.
	By exchanging the roles of $x$ and $y$, it is without loss of generality to assume that $\varphi_{x}\in[0,\pi]$.
	We obtain
	\[
	a_{n}\overset{\eqref{eq:recurrence_relation_closed_form_2}}{=}2\mathfrak{R}(\lambda x^{n})=2\mathfrak{R}\bigl(r_{\lambda}r_{x}^{n}\exp(\mathrm{i}(\varphi_{\lambda}+n\varphi_{x}))\bigr)\quad\quad\textrm{for all }n\in\N.
	\]
	Let $k=\lceil\pi/\varphi_{x}\rceil$.
	We claim that $a_{0},\dots,a_{k}$ are not strictly increasing.
	To see this, note that $a_{0}=1/\beta$, and thus,
	\[
	1=\mathrm{sgn}(a_{0})=\mathrm{sgn}\bigl(2\mathfrak{R}\bigl(r_{\lambda}\exp(\mathrm{i}\varphi_{\lambda})\bigr)\bigr)=\mathrm{sgn}\bigl(2\mathfrak{R}\bigl(\exp(\mathrm{i}\varphi_{\lambda})\bigr)\bigr).
	\]
	On the other hand, we have
	\[
	-1=\mathrm{sgn}\bigl(2\mathfrak{R}\bigl(\exp(\mathrm{i}\varphi_{\lambda}+\pi)\bigr)\bigr).
	\]
	Since $\varphi_{x}\leq\pi$, this implies that either $\mathrm{sgn}(a_{k})=-1$
	or $\mathrm{sgn}(a_{k-1})=-1$ (or both).
	In any case, this implies that there is $\ell\in\N$ with $a_{\ell}<0$.
	Since $t_{n}=1/a_{n}$ for all $n\in\N$, this further implies that $t_{\ell}<0$.
\end{proof}

\begin{figure}
	\begin{center}
		\begin{tikzpicture}
		\def\xd{26}
		\def\xl{1.2}
		\def\lamd{290}
		\def\laml{.8}
		\def\maxIndex{7}
		
		\draw[->] (0,-2.5) -- (0,3.5);
		\draw[->] (-3.5,0) -- (3.5,0);
		\foreach \x in {-3,-2,-1,1,2,3} {
			\draw[] (\x,-.05) -- node[below] {\footnotesize \x} (\x,.05);
		}
		\foreach \x in {-2,2,3} {
			\draw[] (-.05,\x) -- node[left] {\footnotesize \x $i$} (.05,\x);
		}
		\draw[] (-.05,-1) -- node[left] {\footnotesize $-i$} (.05,-1);
		\draw[] (-.05,1) -- node[left] {\footnotesize $i$} (.05,1);
		
		\draw[thick,-stealth,rotate=\xd] (0,0) -- node[right=-1.5pt,at end] {$x$} (\xl,0);
		\draw[thick,-stealth,rotate=\lamd] (0,0) -- node[below=-1.5pt,at end] {$\lambda$} (\laml,0);
		\draw[thick,-stealth,rotate=\lamd+\xd] (0,0) -- node[below=-1.5pt,at end] {$\lambda x$} (\laml*\xl,0);
		\foreach \x in {2,3,...,\maxIndex}{
			\draw[thick,-stealth,rotate=\lamd+\x*\xd] (0,0) -- (\laml*\xl^\x,0);
		}
		\draw[thick,-stealth,rotate=\lamd+\maxIndex*\xd] (0,0) -- node[left,at end] {$\lambda x^\maxIndex$} (\laml*\xl^\maxIndex,0);
		
		\draw [domain=1:\lamd/180*3.14+\maxIndex*(\xd/180)*3.14,variable=\t,smooth,samples=75] plot ({\t r}: {\laml*\xl^((\t-(\lamd/180*3.14))/\xd*180/3.14)});
		\end{tikzpicture}
	\end{center}
	\caption{Multiplying~$\lambda$ repeatedly by $x\in(\C\setminus\R)$ is equivalent to a rotation around the origin that, at some point, reaches the half-plane corresponding to negative real parts.}\label{fig:complex_plane}
\end{figure}
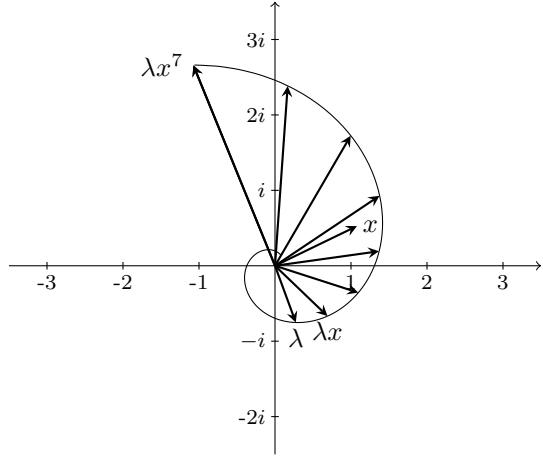

The following lemma shows that, given points $((x_{0},v_{0}),\dots,(x_{k},v_{k}))\in(\R_{>0}\times\R_{>0})^{k+1}$ with $\smash{v_{i}<v_{i+1}<\frac{x_{i+1}}{x_{i}}v_{i}}$ for all $i\in\{0,\dots,k-1\}$, we can construct an instance of \incmax\ with $v(x_{i})=v_{i}$ for all $i\in\{0,\dots,k-1\}$ simply by linearly interpolating between these points.

\begin{lemma}\label{lem:function_from_points}
	Let an instance of \incmax\ with value function $\bar{v}\colon\R_{\geq0}\rightarrow\R_{\geq0}$ and density function $\bar{d}\colon\R_{\geq0}\rightarrow\R_{\geq0}$ be given.
	Let $k\in\N$ and $((x_{0},v_{0}),\dots,(x_{k},v_{k}))\in(\R_{>0}\times\R_{>0})^{k+1}$ with $\bar{v}(x_{0})=v_{0}$ and $\smash{v_{i}<v_{i+1}<\frac{x_{i+1}}{x_{i}}v_{i}}$ for all $i\in\{0,\dots,k-1\}$.
	Then there exist an instance of \incmax\ with value function $v\colon\R_{\geq0}\rightarrow\R_{\geq0}$ and density function $d\colon\R_{\geq0}\rightarrow\R_{\geq0}$ such that $v(x)=\bar{v}(x)$ for all $x\in[0,x_{0}]$ and $v(x_{i})=v_{i}$ for all $i\in\{0,\dots,k\}$.
\end{lemma}

\begin{proof}
	We define 
	\begin{equation}
	v(x):=\bar{v}(x)\quad\textrm{for all }c\in[0,x_{0}]\label{eq:lemma_v_from_0_to_v0}
	\end{equation}
	and, for all $i\in\{0,\dots,k-1\}$,
	\begin{equation}
	v(x):=v_{i}+\frac{x-x_{i}}{x_{i+1}-x_{i}}\bigl(v_{i+1}-v_{i}\bigr)\quad\textrm{for all }x\in(x_{i},x_{i+1}].\label{eq:lemma_v_from_v0_to_vk}
	\end{equation}
	Note that $x_{i+1}>x_{i}$ because $\frac{x_{i+1}}{x_{i}}v_{i}>v_{i+1}$ and $v_{i}<v_{i+1}$.
	Furthermore, for $x>x_{k}$, we define
	\begin{equation}
	v(x):=v_{k-1}+\frac{x-x_{k-1}}{x_{k}-x_{k-1}}(v_{k}-v_{k-1}).\label{eq:lemma_v_from_vk}
	\end{equation}
	We set $d(0):=1$ and $d(x):=\frac{v(x)}{x}$ for all $x>0$.
	
	We have to show that\\
	$(i)$ $v$ is strictly increasing,\\
	$(ii)$ $d$ is strictly decreasing,\\
	$(iii)$ $d(0)=1$,\\
	$(iv)$ $v(x)=xd(x)$ for all $x\in\R_{\geq0}$,\\
	$(v)$ $v(x)=\bar{v}(x)$ for all $x\in[0,x_{0}]$, and\\
	$(vi)$ $v(x_{i})=v_{i}$ for all $i\in\{0,\dots,k\}$.
	
	On the interval $[0,x_{0}]$, $(i)$ holds because of \eqref{eq:lemma_v_from_0_to_v0} and because $\bar{v}$ is strictly increasing.
	For $x>x_{0}$, we have
	\[
	v(x)\overset{\eqref{eq:lemma_v_from_v0_to_vk},\eqref{eq:lemma_v_from_vk}}{=}v_{i}+\frac{x-x_{i}}{x_{i+1}-x_{i}}\bigl(v_{i+1}-v_{i}\bigr)
	\]
	for some $i\in\{0,\dots,k-1\}$, and thus
	\[
	v'(x)\overset{\eqref{eq:lemma_v_from_v0_to_vk},\eqref{eq:lemma_v_from_vk}}{=}\frac{v_{i+1}-v_{i}}{x_{i+1}-x_{i}}\overset{v_{i+1}>v_{i},x_{i+1}>x_{i}}{>}0,
	\]
	i.e., $(i)$ holds for $x>x_{0}$.
	
	On the interval $[0,x_{0}]$, $(ii)$ holds because $\bar{d}$ is strictly increasing and because $d(x)=\bar{d}(x)$ by definition of $d$ and by \eqref{eq:lemma_v_from_0_to_v0}.
	For $x>x_{0}$, we have
	\[
	d(x)\overset{\eqref{eq:lemma_v_from_v0_to_vk},\eqref{eq:lemma_v_from_vk}}{=}\frac{v_{i}}{x}+\frac{1-\frac{x_{i}}{x}}{x_{i+1}-x_{i}}\bigl(v_{i+1}-v_{i}\bigr)
	\]
	for some $i\in\{0,\dots,k-1\}$, and thus
	\[
	d'(x)\overset{\textrm{Def. of }d,\eqref{eq:lemma_v_from_v0_to_vk},\eqref{eq:lemma_v_from_vk}}{=}-\frac{1}{x^{2}}\Bigl(v_{i}-x_{i}\frac{v_{i+1}-v_{i}}{x_{i+1}-x_{i}}\Bigr)\overset{v_{i+1}<\frac{x_{i+1}}{x_{i}}v_{i}}{<}-\frac{1}{x^{2}}\Bigl(v_{i}-\frac{x_{i+1}v_{i}-x_{i}v_{i}}{x_{i+1}-x_{i}}\Bigr)=0,
	\]
	i.e., $(ii)$ holds for $x>x_{0}$.
	
	By definition of $d$, we have $d(0)=1$, i.e., $(iii)$ holds.
	
	We have $v(0)=\bar{v}(0)=0\cdot\bar{d}(0)=0\cdot d(0)$ and $v(x)=xd(x)$ by definition of $d$, i.e., $(iv)$ holds.
	
	By \eqref{eq:lemma_v_from_0_to_v0}, $(v)$ holds.
	
	We have $v(x_{0})\overset{\eqref{eq:lemma_v_from_0_to_v0}}{=}\bar{v}(x_{0})=v_{0}$ and, for $i\in\{1,\dots,k\}$, we have
	\[
	v(x_{i})\overset{\eqref{eq:lemma_v_from_v0_to_vk}}{=}v_{i-1}+\frac{x_{i}-x_{i-1}}{x_{i}-x_{i-1}}\bigl(v_{i}-v_{i-1}\bigr)=v_{i},
	\]
	i.e., $(vi)$ holds.
\end{proof}

The following calculations are needed to show that Lemma~\ref{lem:function_from_points} can be applied to a sequence of points in a later proof.

\begin{lemma}\label{lem:points_satisfy_prerequisites_for_lemma}
	Let $1<\rho<\varphi+1$, $x_{0},v_{0},z>0$,
	\[
	t_{0}=\frac{v_{0}}{x_{0}},\quad\quad t_{n+1}=\frac{1}{\frac{\rho}{t_{n}(1-\epsilon)}-\bigl(\sum_{j=0}^{n}\frac{(\rho+\epsilon)^{j-n}}{t_{j}}\bigr)-\frac{z}{(\rho+\epsilon)^{n}v_{0}}}\quad\textrm{for all }n\in\N\cup\{0\}
	\]
	and let $\epsilon\in(0,1)$ be small enough. By Lemma~\ref{lem:recursive_sequence_becomes_negative}, there exists $\ell'\in\N$ with $t_{\ell'}<0$. Let\linebreak $\ell\in\{0,\dots,\ell'-1\}$ be the smallest index such that $\smash{\frac{1}{t_{\ell}}>\frac{1}{t_{\ell+1}}}$.
	Then,
	
	(i) $(\rho+\epsilon)^{n}v_{0}<\rho(\rho+\epsilon)^{n}v_{0}$ for all $n\in\{0,\dots,\ell\}$,
	
	(ii) $\rho(\rho+\epsilon)^{n}v_{0}<(\rho+\epsilon)^{n+1}v_{0}$ for all
	$n\in\{0,\dots,\ell-1\}$,
	
	(iii) $\rho(\rho+\epsilon)^{n}v_{0}<\frac{\frac{\rho(\rho+\epsilon)^{n}v_{0}}{(1-\epsilon)t_{n}}}{\frac{(\rho+\epsilon)^{n}v_{0}}{t_{n}}}(\rho+\epsilon)^{n}v_{0}$
	for all $n\in\{0,\dots,\ell\}$, and
	
	(iv) $(\rho+\epsilon)^{n+1}v_{0}<\frac{\frac{(\rho+\epsilon)^{n+1}v_{0}}{t_{n+1}}}{\frac{\rho(\rho+\epsilon)^{n}v_{0}}{(1-\epsilon)t_{n}}}\rho(\rho+\epsilon)^{n}v_{0}$
	for all $n\in\{0,\dots,\ell-1\}$.
\end{lemma}

\begin{proof}
	Inequalities $(i)$ and $(ii)$ hold because $\rho>1$, $\epsilon>0$, and $v_{0}>0$.
	
	Let $n\in\{0,\dots,\ell\}$.
	Inequality $(iii)$ is equivalent to $\rho<\frac{\rho}{1-\epsilon}$, which holds because $\epsilon\in(0,1)$.
	
	Let $n\in\{0,\dots,\ell-1\}$.
	Inequality $(iv)$ is equivalent to
	\begin{equation}
	\frac{1}{t_{n}}<(1-\epsilon)\frac{1}{t_{n+1}}\quad\textrm{for all }n\in\{0,\dots,\ell-1\}.\label{eq:lemma_inequality4_equivalent}
	\end{equation}
	For fixed $1<\rho<\varphi+1$, we define the ratio $\smash{r(n,\epsilon):=\frac{t_{n}}{t_{n+1}}}$ for all $n\in\N\cup\{0\}$.
	Note that $\smash{r(0,\epsilon)=\frac{\rho}{1-\epsilon}-1-\frac{z}{v_{0}}t_{0}}$ and, for $n\in\N$,
	\begin{eqnarray}
	r(n,\epsilon) & = & \frac{t_{n}}{t_{n+1}}\nonumber \\
	& = & \Bigl(\frac{\rho}{t_{n}(1-\epsilon)}-\bigl(\sum_{j=0}^{n}\frac{(\rho+\epsilon)^{j-n}}{t_{j}}\bigr)-\frac{z}{(\rho+\epsilon)^{n}v_{0}}\Bigr)t_{n}\nonumber \\
	& = & \Bigl(\frac{\rho}{t_{n}(1-\epsilon)}-\frac{1}{t_{n}}-\frac{\rho}{t_{n-1}(\rho+\epsilon)(1-\epsilon)}\nonumber\\
	&  & \quad+\frac{1}{\rho+\epsilon}\Bigl(\frac{\rho}{t_{n-1}(1-\epsilon)}-\bigl(\sum_{j=0}^{n-1}\frac{(\rho+\epsilon)^{j-n+1}}{t_{j}}\bigr)-\frac{z}{(\rho+\epsilon)^{n-1}v_{0}}\Bigr)\Bigr)t_{n}\nonumber \\
	& = & \Bigl(\frac{\rho}{t_{n}(1-\epsilon)}-\frac{1}{t_{n}}-\frac{\rho}{t_{n-1}(\rho+\epsilon)(1-\epsilon)}+\frac{1}{\rho+\epsilon}\cdot\frac{1}{t_{n}}\Bigr)t_{n}\nonumber \\
	& = & \frac{\rho}{1-\epsilon}-1-\frac{\rho}{(\rho+\epsilon)(1-\epsilon)}\cdot\frac{t_{n}}{t_{n-1}}+\frac{1}{\rho+\epsilon}\nonumber \\
	& = & \frac{\rho}{1-\epsilon}-1+\frac{1}{\rho+\epsilon}-\frac{\rho}{(\rho+\epsilon)(1-\epsilon)}\cdot\frac{1}{r(n-1,\epsilon)}.\label{eq:lemma_ratio_reformulation}
	\end{eqnarray}
	
	\emph{Claim:} For all $n\in\{0,\dots,\ell-1\}$, we have
	\begin{equation}
	r(n,\epsilon)>\frac{r(n,0)}{1-\epsilon}.\label{eq:lemma_claim_ratios_delta_and_0}
	\end{equation}
	
	\emph{Proof of claim:}
	We prove the claim by induction.
	For $n=0$, we have
	\[
	\frac{r(0,\epsilon)}{r(0,0)}=\frac{\frac{\rho}{1-\epsilon}-1-\frac{z}{v_{0}}t_{0}}{\rho-1-\frac{z}{v_{0}}t_{0}}=\frac{\rho-(1-\epsilon)\bigl(1+\frac{z}{v_{0}}t_{0}\bigr)}{\rho-\bigl(1+\frac{z}{v_{0}}t_{0}\bigr)}\cdot\frac{1}{1-\epsilon}\overset{\epsilon>0}{>}\frac{1}{1-\epsilon},
	\]
	i.e., \eqref{eq:lemma_claim_ratios_delta_and_0} holds for $n=0$. 
	
	If $\rho<2$, without loss of generality, we can assume that $\epsilon<1-\frac{\rho}{2}$ because $\epsilon$ is small enough.
	This yields
	\[
	t_{0}=r(0,\epsilon)t_{1}=\Bigl(\frac{\rho}{1-\epsilon}-1-\frac{z}{v_{0}}t_{0}\Bigr)t_{1}<\Bigl(\frac{\rho}{1-\epsilon}-1\Bigr)t_{1}\overset{\epsilon<1-\frac{\rho}{2}}{<}t_{1},
	\]
	i.e., $\ell=0$ and we are done with the proof.
	Thus, assume from now on that $\rho\geq2>\varphi$.
	Suppose~\eqref{eq:lemma_claim_ratios_delta_and_0} holds for some $n\in\{0,\dots,\ell-2\}$.
	Then
	\begin{equation}
	\frac{\rho}{\rho+\epsilon}\cdot\frac{1}{r(n,\epsilon)}\overset{\epsilon>0}{<}\frac{1}{r(n,\epsilon)}\overset{\eqref{eq:lemma_claim_ratios_delta_and_0}}{<}\frac{1-\epsilon}{r(n,0)}\overset{\epsilon>0}{<}\frac{1}{r(n,0)}\label{eq:lemma_ind_ratios_estimate_1}
	\end{equation}
	and
	\[
	\epsilon\rho(\rho+\epsilon)+(1-\epsilon)\rho=\epsilon\rho^{2}+\epsilon^{2}\rho+\rho-\epsilon\rho=(\rho+\epsilon)+(\rho^{2}-\rho-1)\epsilon+\epsilon^{2}\rho\overset{\rho\geq2}{\geq}\rho+\epsilon+\epsilon^{2}\rho>\rho+\epsilon,
	\]
	which is equivalent to
	\begin{equation}
	\epsilon+\frac{1-\epsilon}{\rho+\epsilon}>\frac{1}{\rho}.\label{eq:lemma_ind_ratios_estimate_2}
	\end{equation}
	This yields
	\begin{eqnarray*}
		\frac{r(n+1,\epsilon)}{r(n+1,0)} & \overset{\eqref{eq:lemma_ratio_reformulation}}{=} & \frac{\frac{\rho}{1-\epsilon}-1+\frac{1}{\rho+\epsilon}-\frac{\rho}{(\rho+\epsilon)(1-\epsilon)}\cdot\frac{1}{r(n,\epsilon)}}{\rho-1+\frac{1}{\rho}-\frac{1}{r(n,0)}}\\
		& = & \frac{\rho-1+\epsilon+\frac{1-\epsilon}{\rho+\epsilon}-\frac{\rho}{\rho+\epsilon}\cdot\frac{1}{r(n,\epsilon)}}{\rho-1+\frac{1}{\rho}-\frac{1}{r(n,0)}}\cdot\frac{1}{1-\epsilon}\\
		& \overset{\eqref{eq:lemma_ind_ratios_estimate_1},\eqref{eq:lemma_ind_ratios_estimate_2}}{>} & \frac{\rho-1+\frac{1}{\rho}-\frac{1}{r(n,0)}}{\rho-1+\frac{1}{\rho}-\frac{1}{r(n,0)}}\cdot\frac{1}{1-\epsilon}\\
		& = & \frac{1}{1-\epsilon},
	\end{eqnarray*}
	which proves the claim.
	
	Note that $\ell$ depends on $\epsilon$.
	Thus, we write $\ell(\epsilon)$ from now on.
	By definition of $\ell(0)$, for all $n\in\{0,\dots,\ell-1\}$, we have
	\begin{equation}
	r(n,0)\geq1\label{eq:lemma_ratio_larger_1}
	\end{equation}
	and $r(\ell,0)<1$.
	Note that, by definition, $(t_{n})_{n\in\N\cup\{0\}}$ is continuous in $\epsilon$ for all $n\in\N\cup\{0\}$.
	Thus, if $\epsilon>0$ is small enough (which we assumed), we have $r(\ell,\epsilon)<1$
	and, for all $n\in\{0,\dots,\ell-1\}$,
	\[
	r(n,\epsilon)\overset{\eqref{eq:lemma_claim_ratios_delta_and_0}}{>}\frac{r(n,0)}{1-\epsilon}\overset{\eqref{eq:lemma_ratio_larger_1}}{\geq}\frac{1}{1-\epsilon}.
	\]
	This immediately implies \eqref{eq:lemma_inequality4_equivalent} and thus completes the proof.
\end{proof}

With these preparations, we can now define an instance of \contincmax that excludes one starting value for $\greedycap(c_1,\rho)$.

\begin{theorem}\label{thm:excluding_one_starting_value}
	Let an instance of \incmax\ with value function $\bar{v}\colon\R_{\geq0}\rightarrow\R_{\geq0}$ and density function $\bar{d}\colon\R_{\geq0}\rightarrow\R_{\geq0}$ be given.
	Let $\rho\in(1,\varphi+1)$ and $0<c_{1}<\capacity$.
	Then there exists an instance of \incmax\ with value function $v\colon\R_{\geq0}\rightarrow\R_{\geq0}$ and density function $d\colon\R_{\geq0}\rightarrow\R_{\geq0}$ such that \mbox{$v(c)=\bar{v}(c)$} for all $c\in[0,\capacity]$ and $\greedycap(c_{1},\rho)$ is not $\rho$-competitive for this instance.
	Furthermore, there is some $\overline{\capacity}>0$ such that $v(c)$ can be altered for all $c\geq \overline{\capacity}$ without losing the fact that $\greedycap(c_{1},\rho)$ is not $\rho$-competitive for this instance.
\end{theorem}

\begin{proof}
	If $\greedycap(c_{1},\rho)$ is not $\rho$-competitive for the instance given by $\bar{v}$, we can simply choose $v=\bar{v}$ and ~$\overline{\capacity}$ to be some value for which $\greedycap(c_{1},\rho)$ is not $\rho$-competitive.
	Suppose $\greedycap(c_{1},\rho)$ is $\rho$-competitive for the instance given by $\bar{v}$, and let $(\bar{c}_{1},\bar{c}_{2},\dots)$ with $\bar{c}_{1}=c_{1}$ describe $\greedycap(c_{1},\rho)$ on this instance.
	We will define the function $v$ such that $\greedycap(c_{1},\rho)$ is not $\rho$-competitive on the instance given by $v$.
	In order to do this, we will define a sequence of points and associated values such that $\greedycap(c_{1},\rho)$ is forced to choose these points and Lemma \ref{lem:recursive_sequence_becomes_negative} will show that this sequence is not valid.
	Let $k\in\N_{\geq2}$ such that $\bar{c}_{k-1}\leq \capacity<\bar{c}_{k}$, let $v_k:=\bar{v}(\bar{c}_{k})$, and let $\smash{z:=\sum_{j=1}^{k-1}\bar{c}_{j}}$.
	
	We consider the recursively defined sequence
	\[
	t_{0}=\bar{d}(\bar{c}_{k}),\quad\quad t_{n+1}=\frac{1}{\frac{\rho}{t_{n}(1-\epsilon)}-\bigl(\sum_{j=0}^{n}\frac{(\rho+\epsilon)^{j-n}}{t_{j}}\bigr)-\frac{z}{(\rho+\epsilon)^{n}v_k}}\quad\textrm{for all }n\in\N\cup\{0\}.
	\]
	Since $\rho<\varphi+1$, by Lemma \ref{lem:recursive_sequence_becomes_negative}, there exists $\epsilon'>0$ such that, for all $\epsilon\in(0,\epsilon']$, there is $\ell'\in\N$ with $t_{\ell'}<0$.
	Let $\epsilon\in(0,\epsilon']$ be small enough.
	Since $t_{\ell'}<0$, we have
	\[
	\frac{\rho}{t_{\ell'-1}(1-\epsilon)}-\bigl(\sum_{j=0}^{\ell'-1}\frac{(\rho+\epsilon)^{j-\ell'+1}}{t_{j}}\bigr)-\frac{z}{(\rho+\epsilon)^{\ell'-1}v_k}<0,
	\]
	i.e., we can define $\ell\in\{0,\dots,\ell'-1\}$ to be the smallest index such that
	\begin{equation}
	\frac{1}{t_{\ell}}>\frac{\rho}{t_{\ell}(1-\epsilon)}-\bigl(\sum_{j=0}^{\ell}\frac{(\rho+\epsilon)^{j-\ell}}{t_{j}}\bigr)-\frac{z}{(\rho+\epsilon)^{\ell}v_k}=\frac{1}{t_{\ell+1}}.\label{eq:def_ell_as_smallest_violation}
	\end{equation}
	
	For $n\in\{0,\dots,\ell\}$, let
	\begin{equation}
	x_{2n}:=\frac{(\rho+\epsilon)^{n}v_k}{t_{n}},\quad\quad v_{2n}:=(\rho+\epsilon)^{n}v_k\label{eq:def_points_values_of_alg}
	\end{equation}
	and
	\begin{equation}
	x_{2n+1}:=\frac{\rho(\rho+\epsilon)^{n}v_k}{(1-\epsilon)t_{n}},\quad\quad v_{2n+1}:=\rho(\rho+\epsilon)^{n}v_k.\label{eq:def_points_values_p}
	\end{equation}
	For $c\in[0,x_0]$, we let $v(c)=\bar{v}(c)$ and for $c>x_0$, we let~$v$ be the function with $v(x_n)=v_n$ for all $n\in\{0,\dots,2\ell\}$ that linearly interpolates between these points.
	By Lemmas~\ref{lem:function_from_points} and~\ref{lem:points_satisfy_prerequisites_for_lemma}, this is a valid value function for an instance of \incmax.
	It remains to show that $\greedycap(c_{1},\rho)$ is not $\rho$-competitive for the instance given by $v$.
	
	Let $(c_{1},c_{2},\dots)$ be the solution generated by $\greedycap(c_{1},\rho)$ on the instance given by~$v$.
	Note that $c_{i}=\bar{c}_{i}$ for all $i\in\{1,\dots,k\}$ since $v(c)=\bar{v}(c)$ for all $c\in[0,\bar{c}_{k})]$, i.e., $z=\sum_{j=1}^{k-1}c_{j}$.
	
	\emph{Claim:} For all $n\in\{0,\dots,\ell+1\}$, we have
	\begin{equation}
	d(c_{k+n})=t_{n}.\label{eq:ind_hyp_d-ck+n_equals_tn}
	\end{equation}
	
	\emph{Proof of Claim:}
	We prove this by induction.
	For $n=0$, we have
	\begin{equation}
	d(c_{k})\overset{\eqref{eq:GreedyCap_def}}{=}\bar{d}(\bar{c}_{k})=t_{0},\label{eq:density_k_equals_t0}
	\end{equation}
	where we use the fact that $v(c)=\bar{v}(c)$ for all $c\in[0,\bar{c}_{k}]$.
	
	Suppose, for some $n\in\{0,\dots,\ell\}$, the claim holds for all $i\in\{0,\dots,n\}$.
	By \eqref{eq:def_points_values_of_alg} and because $v(c)=cd(d)$ for all $c\in\R_{\geq0}$, for all $j\in\{0,\dots,n\}$, we have
	\[
	d(x_{2j})=\frac{v_{2j}}{c_{2j}}\overset{\eqref{eq:def_points_values_of_alg}}{=}t_{j}\overset{\eqref{eq:ind_hyp_d-ck+n_equals_tn}}{=}d(c_{k+j}),
	\]
	i.e.,
	\begin{equation}
	c_{k+j}=x_{2j}\label{eq:ck+j_equals_x2j+1}
	\end{equation}
	because $d$ is strictly decreasing.
	Furthermore, we have
	\[
	v(x_{2j+1})\overset{\eqref{eq:def_points_values_p}}{=}\rho(\rho+\epsilon)^{j}v_k\overset{\eqref{eq:def_points_values_of_alg}}{=}\rho v(x_{2j})\overset{\eqref{eq:ck+j_equals_x2j+1}}{=}\rho v(c_{k+j}),
	\]
	i.e.,
	\begin{equation}
	p(c_{k+j})=x_{2j+1}.\label{eq:pk+j_equals_x2j+2}
	\end{equation}
	because $v$ is strictly increasing.
	This yields
	\begin{eqnarray*}
		d(c_{k+n+1}) & \overset{\eqref{eq:GreedyCap_def}}{=} & \frac{v(c_{k+n})}{p(c_{k+n})-\sum_{j=1}^{k+n}c_{j}}\\
		& \overset{\eqref{eq:ck+j_equals_x2j+1},\eqref{eq:pk+j_equals_x2j+2}}{=} & \frac{v(x_{2n})}{x_{2n+1}-z-\sum_{j=0}^{n}x_{2j}}\\
		& \overset{\eqref{eq:def_points_values_of_alg},\eqref{eq:def_points_values_p}}{=} & \frac{(\rho+\epsilon)^{n}v_k}{\frac{\rho(\rho+\epsilon)^{n}v_k}{(1-\epsilon)t_{n}}-z-\sum_{j=0}^{n}\frac{(\rho+\epsilon)^{j}v_k}{t_{j}}}\\
		& = & \frac{1}{\frac{\rho}{(1-\epsilon)t_{n}}-\frac{z}{(\rho+\epsilon)^{n}v_k}-\sum_{j=0}^{n}\frac{(\rho+\epsilon)^{j-n}}{t_{j}}}\\
		& = & t_{n+1},
	\end{eqnarray*}
	which proves the claim.
	
	This implies
	\[
	\frac{1}{d(c_{k+\ell+1})}\overset{\eqref{eq:ind_hyp_d-ck+n_equals_tn}}{=}\frac{1}{t_{\ell+1}}\overset{\eqref{eq:def_ell_as_smallest_violation}}{<}\frac{1}{t_{\ell}}\overset{\eqref{eq:ind_hyp_d-ck+n_equals_tn}}{=}\frac{1}{d(c_{k+\ell})},
	\]
	i.e., either $d(c_{k+\ell+1})>d(c_{k+\ell})$ and thus $c_{k+\ell+1}<c_{k+\ell}$ because $d$ is strictly decreasing, or $d(c_{k+\ell+1})<0$.
	In both cases, by Proposition \ref{prop:GreedyCap_rho-comp_equivalence}, this implies that $\greedycap(c_{1},\rho)$ is not $\rho$-competitive for the instance given by $v$.
	
	Note that the values $v(c)$ for $c>x_{2\ell+1}$ are never used throughout this proof.
	Thus, setting $\overline{\capacity}=x_{2\ell+1}+1$ is a valid choice.
\end{proof}

As mentioned before, we can apply Theorem \ref{thm:excluding_one_starting_value} repeatedly in order to exclude every countable set of starting values and obtain Proposition~\ref{thm:lower_bound_greedycap}.

\GreedyCapLB*

\subsection{General Lower Bound}\label{sec:detLB}

Now we want to employ the techniques we used to prove Lemma~\ref{lem:recursive_sequence_becomes_negative} and Proposition~\ref{thm:lower_bound_greedycap} in order to prove a lower bound on the competitive ratio of \contincmax.
Let $\rho^{*}$ be the unique real root $\rho\geq1$ of the polynomial $-4\rho^{6}+24\rho^{4}-\rho^{3}-30\rho^{2}+31\rho-4$.
As before, we need to show that a recursively defined sequence becomes negative at some point.

\begin{restatable}{lemma}{becomesNegII}\label{lem:recursive_sequence_becomes_negative_2}
	For $\rho\in\R_{\geq0}$ and $\epsilon>0$, consider the recursively defined sequence $(t_{n})_{n\in\N}$ with
	\[
	t_{0}=1,\quad\quad t_{1}=\frac{1-\epsilon}{\rho},\quad\quad t_{n}=\frac{1-\epsilon}{\frac{\rho}{t_{n-1}}-\frac{1}{t_{n-2}}-\frac{1}{\rho}\bigl(\sum_{j=0}^{n-3}\frac{(\rho+\epsilon)^{j+2-n}}{t_{j}}\bigr)}\quad\textrm{for all }n\in\N_{\geq2}.
	\]
	If $1<\rho<\rho^{*}$, then there exists $\epsilon'>0$ such that, for all $\epsilon\in[0,\epsilon']$, there is $\ell\in\N$ with $t_{\ell}<0$.
\end{restatable}

\begin{proof}
	Let $1<\rho<\rho^{*}$. Rearranging terms, we obtain
	\[
	\frac{1-\epsilon}{t_{n}}=\frac{\rho}{t_{n-1}}-\frac{1}{t_{n-2}}-\frac{1}{\rho}\bigl(\sum_{j=0}^{n-3}\frac{(\rho+\epsilon)^{j+2-n}}{t_{j}}\bigr)\quad\textrm{for all }n\in\N_{\geq2}.
	\]
	We substitute $a_{n}=1/t_{n}$ for all $n\in\N\cup\{0\}$ and obtain the recursively defined sequence $(a_{n})_{n\in\N}$ with $a_{0}=1/\beta$ and
	\begin{equation}
	(1-\epsilon)a_{n}=\rho a_{n-1}-a_{n-2}-\frac{1}{\rho}\bigl(\sum_{j=0}^{n-3}(\rho+\epsilon)^{j+2-n}a_{j}\bigr)\quad\textrm{for all }n\in\N_{\geq2}.\label{eq:a_n-recursive_2}
	\end{equation}
	This also implies
	\begin{equation}
	(1-\epsilon)(\rho+\epsilon)a_{n+1}=\rho(\rho+\epsilon)a_{n}-(\rho+\epsilon)a_{n-1}-\frac{1}{\rho}\bigl(\sum_{j=0}^{n-2}(\rho+\epsilon)^{j+2-n}a_{j}\bigr)\quad\textrm{for all }n\in\N.\label{eq:a_n+1-recursive_2}
	\end{equation}
	Subtracting \eqref{eq:a_n-recursive_2} from \eqref{eq:a_n+1-recursive_2}, we obtain
	\[
	(1-\epsilon)(\rho+\epsilon)a_{n+1}-(1-\epsilon)a_{n}=\Bigl(\rho(\rho+\epsilon)a_{n}-(\rho+\epsilon)a_{n-1}-\frac{1}{\rho}a_{n-2}\Bigr)-\Bigl(\rho a_{n-1}-a_{n-2}\Bigr)
	\]
	for all $n\in\N_{\geq2}$, which yields
	\[
	(1-\epsilon)(\rho+\epsilon)a_{n+1}=(\rho^{2}+1+\rho\epsilon-\epsilon)a_{n}-(2\rho+\epsilon)a_{n-1}+\Bigl(1-\frac{1}{\rho}\Bigr)a_{n-2}
	\]
	for all $n\in\N_{\geq2}$.
	Together with the start values
	\begin{equation}
	a_{0}=1,\quad\quad a_{1}=\frac{\rho}{1-\epsilon},\quad\textrm{and}\quad a_{2}=\frac{\rho a_{1}-a_{0}}{1-\epsilon}=\frac{\rho^{2}-1+\epsilon}{(1-\epsilon)^{2}},\label{eq:recurrence_relation_starting_values_2}
	\end{equation}
	this yields a uniquely defined linear homogeneous recurrence relation with characteristic polynomial
	\begin{equation}
	0=x^{3}-\frac{\rho^{2}+1+\rho\epsilon-\epsilon}{(1-\epsilon)(\rho+\epsilon)}x^{2}+\frac{2\rho+\epsilon}{(1-\epsilon)(\rho+\epsilon)}x-\frac{1-\frac{1}{\rho}}{(1-\epsilon)(\rho+\epsilon)}.\label{eq:recurrence_relation_char_poly}
	\end{equation}
	Using
	\begin{align*}
		a & =-\frac{\rho^{2}+1+\rho\epsilon-\epsilon}{(1-\epsilon)(\rho+\epsilon)},\\
		b & =\frac{2\rho+\epsilon}{(1-\epsilon)(\rho+\epsilon)},\\
		c & =-\frac{1-\frac{1}{\rho}}{(1-\epsilon)(\rho+\epsilon)},
	\end{align*}
	the discriminant of this polynomial is
	\[
	D(\rho,\epsilon)=\Bigl(\frac{a^{3}}{27}-\frac{ab}{6}+\frac{c}{2}\Bigr)^{2}+\Bigl(\frac{b}{3}-\frac{a^{2}}{9}\Bigr)^{3}.
	\]
	In particular, we have
	\[
	D(\rho,0)=\frac{-4\rho^{6}+24\rho^{4}-\rho^{3}-30\rho^{2}+31\rho-4}{108\rho^{5}}>0
	\]
	because $1<\rho<\rho^{*}$.
	Note that $a,b,c$ are all continuous in $\epsilon$ and, thus, so is $D(\rho,\epsilon)$.
	Therefore, there is $\epsilon'>0$ such that, for all $\epsilon\in[0,\epsilon']$, we have $D(\rho,\epsilon)>0$.
	The fact that for the discriminant of the polynomial we have $D(\rho,\epsilon)>0$ implies that \eqref{eq:recurrence_relation_char_poly} has one real root $r_{1}$ and two complex conjugate roots $r_{2}$ and $r_{3}=\overline{r_{2}}$.
	We want to express the recurrence relation in terms of the roots, i.e., we want to find $\lambda_{1},\lambda_{2},\lambda_{3}\in\C$ such that
	\begin{equation}
	a_{n}=\lambda_{1}r_{1}^{n}+\lambda_{2}r_{2}^{n}+\lambda_{3}r_{3}^{n}.\label{eq:recurrence_relation_ansatz}
	\end{equation}
	We will now show that $\lambda_{3}=\overline{\lambda_{2}}$.
	Using the starting values \eqref{eq:recurrence_relation_starting_values_2} together with \eqref{eq:recurrence_relation_ansatz}, we obtain
	\begin{align*}
		1 & =\lambda_{1}+\lambda_{2}+\lambda_{3},\\
		\frac{\rho}{1-\epsilon} & =\lambda_{1}r_{1}+\lambda_{2}r_{2}+\lambda_{3}r_{3},\\
		\frac{\rho^{2}-1+\epsilon}{(1-\epsilon)^{2}} & =\lambda_{1}r_{1}^{2}+\lambda_{2}r_{2}^{2}+\lambda_{3}r_{3}^{2}.
	\end{align*}
	This implies
	\begin{align}
		0 & =\mathfrak{Im}(\lambda_{1})+\mathfrak{Im}(\lambda_{2})+\mathfrak{Im}(\lambda_{3}),\label{eq:recurrence_relation_Im_equals_0_1}\\
		0 & =\mathfrak{Im}(\lambda_{1}r_{1})+\mathfrak{Im}(\lambda_{2}r_{2})+\mathfrak{Im}(\lambda_{3}r_{3}),\label{eq:recurrence_relation_Im_equals_0_2}\\
		0 & =\mathfrak{Im}(\lambda_{1}r_{1}^{2})+\mathfrak{Im}(\lambda_{2}r_{2}^{2})+\mathfrak{Im}(\lambda_{3}r_{3}^{2}).\label{eq:recurrence_relation_Im_equals_0_3}
	\end{align}
	We obtain
	\begin{eqnarray}
	&  & \mathfrak{Re}(r_{1}-r_{2})\mathfrak{Im}(\lambda_{2}+\lambda_{3})+\mathfrak{Im}(r_{1}-r_{2})\mathfrak{Re}(\lambda_{2}-\lambda_{3})\nonumber \\
	& = & \mathfrak{Re}(r_{1}-r_{2})\mathfrak{Im}(\lambda_{2})+\mathfrak{Im}(r_{1}-r_{2})\mathfrak{Re}(\lambda_{2})\nonumber \\
	& & +\mathfrak{Re}(r_{1}-r_{2})\mathfrak{Im}(\lambda_{3})-\mathfrak{Im}(r_{1}-r_{2})\mathfrak{Re}(\lambda_{3})\nonumber \\
	& \overset{r_{1}\in\R,r_{3}=\overline{r_{2}}}{=} & \mathfrak{Re}(r_{1}-r_{2})\mathfrak{Im}(\lambda_{2})+\mathfrak{Im}(r_{1}-r_{2})\mathfrak{Re}(\lambda_{2})\nonumber \\
	& & +\mathfrak{Re}(r_{1}-r_{3})\mathfrak{Im}(\lambda_{3})+\mathfrak{Im}(r_{1}-r_{3})\mathfrak{Re}(\lambda_{3})\nonumber \\
	& = & \mathfrak{Im}\bigl((r_{1}-r_{2})\lambda_{2}+(r_{1}-r_{3})\lambda_{3}\bigr)\nonumber \\
	& \overset{r_{1}\in\R}{=} & r_{1}\mathfrak{Im}(\lambda_{2})-\mathfrak{Im}(\lambda_{2}r_{2})+r_{1}\mathfrak{Im}(\lambda_{3})-\mathfrak{Im}(\lambda_{3}r_{3})\nonumber \\
	& \overset{\eqref{eq:recurrence_relation_Im_equals_0_1},\eqref{eq:recurrence_relation_Im_equals_0_2}}{=} & -r_{1}\mathfrak{Im}(\lambda_{1})+\mathfrak{Im}(\lambda_{1}r_{1})\nonumber \\
	& \overset{r_{1}\in\R}{=} & 0\label{eq:recurrence_relation_LGS_1}
	\end{eqnarray}
	and
	\begin{eqnarray}
	&  & \mathfrak{Re}(r_{1}^{2}-r_{2}^{2})\mathfrak{Im}(\lambda_{2}+\lambda_{3})+\mathfrak{Im}(r_{1}^{2}-r_{2}^{2})\mathfrak{Re}(\lambda_{2}-\lambda_{3})\nonumber \\
	& = & \mathfrak{Re}(r_{1}^{2}-r_{2}^{2})\mathfrak{Im}(\lambda_{2})+\mathfrak{Im}(r_{1}^{2}-r_{2}^{2})\mathfrak{Re}(\lambda_{2})\nonumber \\
	& & +\mathfrak{Re}(r_{1}^{2}-r_{2}^{2})\mathfrak{Im}(\lambda_{3})-\mathfrak{Im}(r_{1}^{2}-r_{2}^{2})\mathfrak{Re}(\lambda_{3})\nonumber \\
	& \overset{r_{1}\in\R,r_{3}=\overline{r_{2}}}{=} & \mathfrak{Re}(r_{1}^{2}-r_{2}^{2})\mathfrak{Im}(\lambda_{2})+\mathfrak{Im}(r_{1}^{2}-r_{2}^{2})\mathfrak{Re}(\lambda_{2})\nonumber \\
	& & +\mathfrak{Re}(r_{1}^{2}-r_{3}^{2})\mathfrak{Im}(\lambda_{3})+\mathfrak{Im}(r_{1}^{2}-r_{3}^{2})\mathfrak{Re}(\lambda_{3})\nonumber \\
	& = & \mathfrak{Im}\bigl((r_{1}^{2}-r_{2}^{2})\lambda_{2}+(r_{1}^{2}-r_{3}^{2})\lambda_{3}\bigr)\nonumber \\
	& \overset{r_{1}\in\R}{=} & r_{1}^{2}\mathfrak{Im}(\lambda_{2})-\mathfrak{Im}(\lambda_{2}r_{2}^{2})+r_{1}^{2}\mathfrak{Im}(\lambda_{3})-\mathfrak{Im}(\lambda_{3}r_{3}^{2})\nonumber \\
	& \overset{\eqref{eq:recurrence_relation_Im_equals_0_1},\eqref{eq:recurrence_relation_Im_equals_0_3}}{=} & -r_{1}^{2}\mathfrak{Im}(\lambda_{1})+\mathfrak{Im}(\lambda_{1}r_{1}^{2})\nonumber \\
	& \overset{r_{1}\in\R}{=} & 0.\label{eq:recurrence_relation_LGS_2}
	\end{eqnarray}
	Note that $(r_{1}^{2}-r_{2}^{2})=(r_{1}-r_{2})(r_{1}+r_{2})$ and the fact that $(r_{1}+r_{2})\notin\R$ imply that the matrix
	\[
	M:=\begin{pmatrix}\mathfrak{Re}(r_{1}-r_{2}) & \mathfrak{Im}(r_{1}-r_{2})\\
	\mathfrak{Re}(r_{1}^{2}-r_{2}^{2}) & \mathfrak{Im}(r_{1}^{2}-r_{2}^{2})
	\end{pmatrix}
	\]
	has rank $2$, i.e., the system of equations $M\cdot(\mathfrak{Im}(\lambda_{2}+\lambda_{3}),\mathfrak{Re}(\lambda_{2}-\lambda_{3}))^{T}=(0,0)^{T}$ given by \eqref{eq:recurrence_relation_LGS_1} and \eqref{eq:recurrence_relation_LGS_2} only has the solution $(0,0)^{T}$.
	This yields that $\mathfrak{Im}(\lambda_{2})=-\mathfrak{Im}(\lambda_{3})$ and $\mathfrak{Re}(\lambda_{2})=\mathfrak{Re}(\lambda_{3})$, i.e., we have $\lambda_{3}=\overline{\lambda_{2}}$, and, by \eqref{eq:recurrence_relation_Im_equals_0_1}, we have $\lambda_{1}\in\R$.
	Thus, the recurrence relation can be written as
	\[
	a_{n}=\lambda_{1}r_{1}^{n}+\lambda_{2}r_{2}^{n}+\overline{\lambda_{2}r_{2}^{n}}=\lambda_{1}r_{1}^{n}+2\mathfrak{Re}(\lambda_{2}r_{2}^{n}).
	\]
	We have $|r_{1}|<1$ and $|r_{2}|>1$. There exists $\ell\in\N$ such that $\mathfrak{Re}(\lambda_{2}r_{2}^{\ell})<-\frac{\lambda_{1}}{2}$ because $r_{2}\notin\R$ and $|r_{2}|>1$.
	Thus, $a_{\ell}=\lambda_{1}r_{1}^{\ell}+2\mathfrak{Re}(\lambda_{2}r_{2}^{\ell})<\lambda_{1}r_{1}^{\ell}-\lambda_{1}<0$ where the last inequality follows from the fact that $|r_{1}|<1$.
\end{proof}

With this lemma, we are ready to construct our lower bound on the competitive ratio of \contincmax and thus, via Propositions~\ref{prop:CR_incmax_incmaxsep} and~\ref{thm:CRC_to_RC}, of \incmax.
Recall that $\rho^{*}\approx2.246$ is the unique solution $\rho\geq1$ to the equation $-4\rho^{6}+24\rho^{4}-\rho^{3}-30\rho^{2}+31\rho-4$.

\detLB*

\begin{proof}
	Let $\rho<\rho^{*}$.
	By Lemma \ref{lem:recursive_sequence_becomes_negative_2}, there is $\epsilon'>0$ such that, for all $\epsilon\in(0,\epsilon']$, the recursively defined sequence $(t_{n})_{n\in\N}$ with
	\[
	t_{0}=1,\quad\quad t_{1}=\frac{1-\epsilon}{\rho},\quad\quad t_{n}=\frac{1-\epsilon}{\frac{\rho}{t_{n-1}}-\frac{1}{t_{n-2}}-\frac{1}{\rho}\bigl(\sum_{j=0}^{n-3}\frac{(\rho+\epsilon)^{j+2-n}}{t_{j}}\bigr)}\quad\textrm{for all }n\in\N_{\geq2}
	\]
	becomes negative at some point.
	Thus, for $\epsilon>0$, we can define $\ell(\epsilon)\in\N$ to be the smallest value such that $\smash{\frac{1}{t_{\ell(\epsilon)}}\geq\frac{1}{t_{\ell(\epsilon)+1}}}$.
	Note that this is the case when either $t_{\ell(\epsilon)+1}<0$ or $t_{\ell(\epsilon)+1}\geq t_{\ell(\epsilon)}$.
	This implies that $t_{0}>t_{1}>\dots>t_{\ell(\epsilon)}$.
	Let $v_{i}=(\rho+\epsilon)^{i}$, and let $\epsilon\in(0,\epsilon']$ be small enough such that $\ell:=\ell(\epsilon)=\ell(0)$ and such that, for all $n\in\{0,\dots,\ell-1\}$, we have
	\begin{equation}
	\frac{1}{t_{n+1}}>\frac{\epsilon}{\rho+\epsilon}\cdot\Bigl(\frac{1}{t_{n}}+\frac{1}{v_{n}}\sum_{j=0}^{n-1}\frac{v_{j}}{t_{j}}\Bigr),\label{eq:det_lower_bound_eps_small_enough_1}
	\end{equation}
	which is possible since the inequality holds for $\epsilon=0$.
	
	We consider the instance of \contincmax with the value function that satisfies
	\[
	v(0)=0,\quad v\Bigl(\frac{v_{n}}{t_{n}}\Bigr)=v_{n},\quad\textrm{and}\quad v\Bigl(\frac{v_{n}}{t_{n+1}}\Bigr)=v_{n}\quad\textrm{for all }n\in\{0,1,\dots,\ell\}
	\]
	and linearly interpolates between those points.
	This means that, for $0\leq c\leq1$, we have $d(c)=1$, for $\smash{\frac{v_{n}}{t_{n}}\leq c\leq\frac{v_{n}}{t_{n+1}}}$, we have $v(c)=v_{n}$, and, for $\smash{\frac{v_{n}}{t_{n+1}}\leq c\leq\frac{v_{n+1}}{t_{n+1}}}$, we have $d(c)=t_{n+1}$.
	
	Suppose there was a $\rho$-competitive solution $(c_{0},\dots,c_{N})$, for some $N\in\N$, for this problem instance.
	Without loss of generality, we can assume that, for all $n\in\{0,\dots,N\}$, we have
	\begin{equation}
	d(c_{n})\in\{t_{0},t_{1},\dots,t_{\ell}\}.\label{eq:det_lower_bound_cn_from_sequence}
	\end{equation}
	If this is not the case, we have $\smash{\frac{v_{i}}{t_{i}}<c_{n}<\frac{v_{i}}{t_{i+1}}}$ for some $i\in\{0,\dots,N-1\}$.
	Then, we can improve the solution by setting $c_{n}=\frac{v_{i}}{t_{i}}$ because $v(c_{n})=v_{i}=v\bigl(\frac{v_{i}}{t_{i}}\bigr)$ and $\frac{v_{i}}{t_{i}}<c_{n}$, i.e., the solution obtains the same value faster and can start adding the next size earlier.
	Furthermore, we can assume that
	\begin{equation}
	d(c_{n})>d(c_{n+1})\label{eq:det_lower_bound_dn_decreasing}
	\end{equation}
	for all $n\in\{0,\dots,N-1\}$.
	Otherwise we can improve the solution by removing the smaller of~$c_{n}$ and~$c_{n+1}$.
	This also implies that $N\leq\ell$.
	
	We will now show by induction, that, for $n\in\{0,\dots,N\}$, we have
	\begin{equation}
	d(c_{n})>t_{n+1},\label{eq:det_lower_bound_ind-hyp_1}
	\end{equation}
	and, for $n\in\{0,\dots,N-1\}$, we have
	\begin{equation}
	c_{n}\geq\frac{1}{\rho}\cdot\frac{v_{n}}{t_{n}}.\label{eq:det_lower_bound_ind-hyp_2}
	\end{equation}
	For $n=0,$ by Lemma \ref{lem:competitive_equivalence}, we have $d(c_{0})\geq\frac{1}{\rho}>\frac{1-\epsilon}{\rho}$, i.e., \eqref{eq:det_lower_bound_ind-hyp_1} holds.
	If $c_{0}<\frac{1}{\rho}\cdot\frac{v_{0}}{t_{0}}=\frac{1}{\rho}$, then the solution achieved by only adding $c_{0}$ is $\rho$-competitive up to size $p(c_{0})=\rho c_{0}$.
	By Lemma~\ref{lem:competitive_equivalence}, we have
	\[
	d(c_{1})\geq\frac{v(c_{0})}{p(c_{0})-\sum_{j=0}^{0}c_{j}}=\frac{c_{0}t_{0}}{\rho c_{0}-c_{0}}=\frac{1}{\rho-1}>\frac{1-\epsilon}{\rho}=t_{1},
	\]
	i.e., using \eqref{eq:det_lower_bound_cn_from_sequence}, we obtain $d(c_{1})=t_{0}$, which is a contradiction to \eqref{eq:det_lower_bound_dn_decreasing}.
	Thus, $c_{0}\geq\frac{1}{\rho}\cdot\frac{v_{0}}{t_{0}}$, i.e., also \eqref{eq:det_lower_bound_ind-hyp_2} holds.
	Now, assume that, for some $n\in\{0,\dots,N-1\}$, \eqref{eq:det_lower_bound_ind-hyp_1} and \eqref{eq:det_lower_bound_ind-hyp_2} hold for all $i\in\{0,\dots,n\}$.
	As $c_{n}\geq\frac{1}{\rho}\cdot\frac{v_{n}}{t_{n}}$, we have 
	\begin{equation}
	p(c_{n})=\frac{\rho v(c_{n})}{t_{n+1}}.\label{eq:det_lower_bound_val_of_p(cn)}
	\end{equation}
	Note that we have $d(c_{i})\in\{t_{0},\dots,t_{\ell}\}$ for all $i\in\{0,\dots,N\}$ and that $d(c_{i+1})<d(c_{i})$.
	We have $t_{0}>\dots>t_{\ell}$, which, together with \eqref{eq:det_lower_bound_ind-hyp_1} for $i\in\{0,\dots,n\}$, implies that $d(c_{i})=t_{i}$ for all $i\in\{0,\dots,n\}$.
	Thus, $c_{n}\leq\frac{v_{n}}{t_{n}}$ and therefore
	\begin{equation}
	v(c_{n})\leq v\bigl(\frac{v_{n}}{t_{n}}\bigr)=v_{n}.\label{eq:det_lower_bound_v(c_n)_leq_v_n}
	\end{equation}
	By Lemma \ref{lem:competitive_equivalence}, $c_{n+1}$ has to satisfy
	\begin{eqnarray*}
		d(c_{n+1}) & \geq & \frac{v(c_{n})}{p(c_{n})-\sum_{j=0}^{n}c_{j}}\overset{\eqref{eq:det_lower_bound_val_of_p(cn)}}{=}\frac{v(c_{n})}{\frac{\rho v(c_{n})}{t_{n+1}}-\sum_{j=0}^{n}c_{j}}\\
		& \overset{\eqref{eq:det_lower_bound_ind-hyp_2}}{\geq} & \frac{v(c_{n})}{\frac{\rho v(c_{n})}{t_{n+1}}-c_{n}-\frac{1}{\rho}\sum_{j=0}^{n-1}\frac{v_{j}}{t_{j}}}\\
		& = & \frac{1}{\frac{\rho}{t_{n+1}}-\frac{1}{t_{n}}-\frac{1}{\rho v(c_{n})}\sum_{j=0}^{n-1}\frac{(\rho+\epsilon)^{j}}{t_{j}}}\\
		& \overset{\eqref{eq:det_lower_bound_v(c_n)_leq_v_n}}{\geq} & \frac{1}{\frac{\rho}{t_{n+1}}-\frac{1}{t_{n}}-\frac{1}{\rho v_{n}}\bigl(\sum_{j=0}^{n-1}\frac{(\rho+\epsilon)^{j}}{t_{j}}\bigr)}\\
		& = & \frac{1}{\frac{\rho}{t_{n+1}}-\frac{1}{t_{n}}-\frac{1}{\rho}\bigl(\sum_{j=0}^{n-1}\frac{(\rho+\epsilon)^{j-n}}{t_{j}}\bigr)}\\
		& = & \frac{1}{1-\epsilon}t_{n+2}>t_{n+2}.
	\end{eqnarray*}
	As $t_{0}>\dots>t_{\ell}$, we have $d(c_{n+1})\notin\{t_{n+2},\dots,t_{\ell}\}$, which implies that $d(c_{n+1})=t_{n+1}$.
	Thus,~\eqref{eq:det_lower_bound_ind-hyp_1} holds for $n+1$.
	Now, assume that, for some $n\in\{0,\dots,N-2\}$, \eqref{eq:det_lower_bound_ind-hyp_1} and \eqref{eq:det_lower_bound_ind-hyp_2} hold for all $i\in\{0,\dots,n\}$ and suppose that \eqref{eq:det_lower_bound_ind-hyp_2} does not hold for $n+1$, i.e., we have $c_{n+1}<\frac{1}{\rho}\cdot\frac{v_{n+1}}{t_{n+1}}$.
	Then, $v(c_{n+1})=c_{n+1}t_{n+1}<\frac{1}{\rho}v_{n+1}$, i.e., the solution $(c_{0},\dots,c_{n+1})$ is $\rho$-competitive up to size $\smash{p(c_{n+1})=\frac{\rho v(c_{n+1})}{t_{n+1}}=\rho c_{n+1}}$.
	By Lemma \ref{lem:competitive_equivalence}, the next size in the solution has to satisfy
	\begin{eqnarray*}
		d(c_{n+2}) & \geq & \frac{v(c_{n+1})}{p(c_{n+1})-\sum_{j=0}^{n+1}c_{j}}=\frac{v(c_{n+1})}{\rho c_{n+1}-\sum_{j=0}^{n+1}c_{j}}\\
		& \overset{\eqref{eq:det_lower_bound_ind-hyp_2}}{\geq} & \frac{v(c_{n+1})}{\rho c_{n+1}-c_{n+1}-c_{n}-\frac{1}{\rho}\sum_{j=0}^{n-1}\frac{v_{j}}{t_{j}}}\\
		& = & \frac{1}{\frac{\rho}{t_{n+1}}-\frac{1}{t_{n+1}}-\frac{c_{n}}{v(c_{n+1})}-\frac{1}{\rho v(c_{n+1})}\sum_{j=0}^{n-1}\frac{v_{j}}{t_{j}}}\\
		& > & \frac{1}{\frac{\rho}{t_{n+1}}-\frac{1}{t_{n+1}}-\frac{\rho c_{n}}{v_{n+1}}-\frac{1}{v_{n+1}}\sum_{j=0}^{n-1}\frac{v_{j}}{t_{j}}}\\
		& \geq & \frac{1}{\frac{\rho}{t_{n+1}}-\frac{1}{t_{n+1}}-\frac{\rho v_{n}}{v_{n+1}t_{n+1}}-\frac{1}{v_{n+1}}\sum_{j=0}^{n-1}\frac{v_{j}}{t_{j}}}\\
		& = & \frac{1}{\frac{\rho}{t_{n+1}}-\frac{1}{t_{n+1}}-\frac{\rho}{\rho+\epsilon}\cdot\frac{1}{t_{n+1}}-\frac{\rho}{\rho+\epsilon}\frac{1}{\rho v_{n}}\sum_{j=0}^{n-1}\frac{v_{j}}{t_{j}}}\\
		& \geq & \frac{1}{\frac{\rho}{t_{n+1}}-\frac{1}{t_{n+1}}-\frac{\rho}{\rho+\epsilon}\cdot\frac{1}{t_{n}}-\frac{\rho}{\rho+\epsilon}\frac{1}{\rho v_{n}}\sum_{j=0}^{n-1}\frac{v_{j}}{t_{j}}}\\
		& \overset{\eqref{eq:det_lower_bound_eps_small_enough_1}}{\geq} & \frac{1}{\frac{\rho}{t_{n+1}}-\frac{1}{t_{n}}-\frac{1}{\rho v_{n}}\sum_{j=0}^{n-1}\frac{v_{j}}{t_{j}}}\\
		& = & \frac{1}{\frac{\rho}{t_{n+1}}-\frac{1}{t_{n}}-\frac{1}{\rho}\sum_{j=0}^{n-1}\frac{(\rho+\epsilon)^{j-n}}{t_{j}}}\\
		& = & \frac{1}{1-\epsilon}t_{n+2}>t_{n+2}.
	\end{eqnarray*}
	As $t_{0}>\dots>t_{\ell}$, we have $d(c_{n+2})\notin\{t_{n+2},\dots,t_{\ell}\}$.
	But, for all $i\in\{0,\dots,n+1\}$, we also have $d(c_{i})=t_{i}$ and thus $d(c_{n+2})\notin\{t_{0},\dots,t_{n+1}\}$.
	This is a contradiction to \eqref{eq:det_lower_bound_cn_from_sequence} and therefore \eqref{eq:det_lower_bound_ind-hyp_2} holds for $n+1$.
	
	We have established that \eqref{eq:det_lower_bound_ind-hyp_1} holds for all $n\in\{0,\dots,N\}$.
	Together with the fact that \mbox{$d(c_{n})\in\{t_{0},\dots,t_{\ell}\}$} for \mbox{$n\in\{0,\dots,N\}$}, $d(c_{0})>\dots>d(c_{N})$ and $t_{0}>\dots>t_{\ell}$, we obtain $N\geq\ell-1$ and $d(c_{n})=t_{n}$ for all \mbox{$n\in\{0,\dots,\ell-1\}$}.
	If $N=\ell-1$, the solution obtains a value of $v(c_{\ell-1})$ for size $\smash{\frac{v_{\ell}}{t_{\ell}}}$.
	Yet, the optimum solution for this size has value $\smash{v(\frac{v_{\ell}}{t_{\ell}})=v_{\ell}=(\rho+\epsilon)v_{\ell-1}\geq(\rho+\epsilon)v(c_{\ell-1})}$.
	Thus, $(c_{0},\dots,c_{N})$ would not be $\rho$-competitive.
	Therefore, we have $N=\ell$.
	By a similar argument, we find that $\smash{c_{\ell}\geq\frac{1}{\rho}\cdot\frac{v_{\ell}}{t_{\ell}}}$.
	By \eqref{eq:det_lower_bound_ind-hyp_1}, we know that $d(c_{\ell})>t_{\ell+1}$.
	If $t_{\ell+1}\geq t_{\ell}$, we know that $d(c_{\ell})\neq t_{\ell}$.
	But we also have $d(c_{\ell})\notin\{t_{0},\dots,t_{\ell-1}\}$ as $d(c_{\ell})<d(c_{\ell-1})=t_{\ell-1}$ and $t_{0}<\dots<t_{\ell-1}$.
	This is a contradiction to the assumption that $d(c_{\ell})\in\{t_{0},\dots,t_{\ell}\}$.
	Therefore, $t_{\ell+1}<t_{\ell}$.
	By definition of $\ell$, we have $\smash{\frac{1}{t_{\ell+1}}<\frac{1}{t_{\ell}}}$, which implies that
	\[
	0>t_{\ell+1}=\frac{1-\epsilon}{\frac{\rho}{t_{\ell}}-\frac{1}{t_{\ell-1}}-\frac{1}{\rho}\sum_{j=0}^{\ell-2}\frac{(\rho+\epsilon)^{j+1-\ell}}{t_{j}}}>\frac{1}{\frac{\rho}{t_{\ell}}-\frac{1}{t_{\ell-1}}-\frac{1}{\rho}\sum_{j=0}^{\ell-2}\frac{(\rho+\epsilon)^{j+1-\ell}}{t_{j}}}.
	\]
	This is equivalent to
	\begin{equation}
	\frac{\rho}{t_{\ell}}-\frac{1}{t_{\ell-1}}-\frac{1}{\rho}\sum_{j=0}^{\ell-2}\frac{(\rho+\epsilon)^{j+1-\ell}}{t_{j}}<0.\label{eq:det_lower_bound_tl+1_smaller_0}
	\end{equation}
	We have $p(c_{\ell-1})=\frac{\rho v(c_{\ell-1})}{t_{\ell}}$ and thus
	\begin{eqnarray*}
		\frac{1}{v_{\ell-1}}\Bigl(p(c_{\ell-1})-\sum_{j=0}^{\ell-1}c_{j}\Bigr) & = & \frac{1}{v_{\ell-1}}\Bigl(\frac{\rho v(c_{\ell-1})}{t_{\ell}}-c_{\ell-1}-\sum_{j=0}^{\ell-2}c_{j}\Bigr)\\
		& \overset{\eqref{eq:det_lower_bound_ind-hyp_2}}{\leq} & \frac{1}{v_{\ell-1}}\Bigl(\frac{\rho v(c_{\ell-1})}{t_{\ell}}-c_{\ell-1}-\frac{1}{\rho}\sum_{j=0}^{\ell-2}\frac{v_{j}}{t_{j}}\Bigr)\\
		& = & \frac{c_{\ell-1}}{v_{\ell-1}}\Bigl(\rho\frac{t_{\ell-1}}{t_{\ell}}-1\Bigr)-\frac{1}{\rho}\sum_{j=0}^{\ell-2}\frac{(\rho+\epsilon)^{j+1-\ell}}{t_{j}}.
	\end{eqnarray*}
	As $c_{\ell-1}\leq\frac{v_{\ell-1}}{t_{\ell-1}}$, we obtain
	\begin{eqnarray*}
		\frac{1}{v_{\ell-1}}\Bigl(p(c_{\ell-1})-\sum_{j=0}^{\ell-1}c_{j}\Bigr) & \leq & \frac{1}{t_{\ell-1}}\Bigl(\rho\frac{t_{\ell-1}}{t_{\ell}}-1\Bigr)-\frac{1}{\rho}\sum_{j=0}^{\ell-2}\frac{(\rho+\epsilon)^{j+1-\ell}}{t_{j}}\\
		& = & \frac{\rho}{t_{\ell}}-\frac{1}{t_{\ell-1}}-\frac{1}{\rho}\sum_{j=0}^{\ell-2}\frac{(\rho+\epsilon)^{j+1-\ell}}{t_{j}}\\
		& \overset{\eqref{eq:det_lower_bound_tl+1_smaller_0}}{<} & 0.
	\end{eqnarray*}
	As $v_{\ell-1}=(\rho+\epsilon)^{\ell-1}>0$, we find that $p(c_{\ell-1})-\sum_{j=0}^{j-1}c_{j}<0$, which is a contradiction to Lemma~\ref{lem:competitive_equivalence}~(iii) and the fact that $(c_{0},\dots,c_{N})$ is $\rho$-competitive.
	Thus, a $\rho$-competitive solution cannot exist.
\end{proof}

\begin{figure}
	\begin{center}
		\begin{tikzpicture}[scale=1.1]
		\begin{axis}[
		axis lines=middle,
		xtick={0},
		ytick={1,2.1,4.41,9.261,19.448,40.841},
		xticklabels={},
		yticklabels={,$\rho+\epsilon$,$(\rho+\epsilon)^2$,$(\rho+\epsilon)^3$,$(\rho+\epsilon)^4$,$(\rho+\epsilon)^5$},
		scaled x ticks=false,
		domain=0:320
		]
		
		\def\rhoOne{2.1}			
		\def\rhoTwo{4.41}			
		\def\rhoThree{9.261}		
		\def\rhoFour{19.448}		
		\def\rhoFive{40.841}		
		
		\node at (axis cs:315,2){$c$};
		\addplot[samples=2,domain=319:320,white] {\rhoFive+4};
		
		\addplot[samples=2,domain=0:1] {x};
		\addplot[samples=2,domain=1:2.1] {1};
		\addplot[samples=2,domain=2.1:4.4] {\rhoOne/4.4*x};
		\addplot[samples=2,domain=4.4:7.2] {\rhoOne};
		\addplot[samples=2,domain=7.2:15] {\rhoTwo/15*x};
		\addplot[samples=2,domain=15:21.3] {\rhoTwo};
		\addplot[samples=2,domain=21.3:44.8] {\rhoThree/44.8*x};
		\addplot[samples=2,domain=44.8:57] {\rhoThree};
		\addplot[samples=2,domain=57:119.8] {\rhoFour/119.8*x};
		\addplot[samples=2,domain=119.8:137] {\rhoFour};
		\addplot[samples=2,domain=137:287.7] {\rhoFive/287.7*x};
		
		\addplot[dotted] coordinates {(287.7, 0) (287.7, \rhoFive)};
		\addplot[dotted] coordinates {(137, 0) (137, \rhoFour)};
		\addplot[dotted] coordinates {(119.8, 0) (119.8, \rhoFour)};
		\addplot[dotted] coordinates {(57, 0) (57, \rhoThree)};
		\addplot[dotted] coordinates {(44.8, 0) (44.8, \rhoThree)};
		\addplot[dotted] coordinates {(21.3, 0) (21.3, \rhoTwo)};
		\end{axis}
		
		\draw [decorate,decoration={brace,amplitude=5pt}] (.95,-.03) -- node[below=5pt] {\footnotesize $d(c)=t_3$} (.45,-.03);
		\draw [decorate,decoration={brace,amplitude=5pt}] (2.57,-.03) -- node[below=5pt] {\footnotesize $d(c)=t_4$} (1.22,-.03);
		\draw [decorate,decoration={brace,amplitude=5pt}] (6.17,-.03) -- node[below=5pt] {\footnotesize $d(c)=t_5$} (2.93,-.03);
		\end{tikzpicture}
	\end{center}
	\caption{Lower bound construction for $\rho=2.1$.}\label{fig:detLB}
\end{figure}
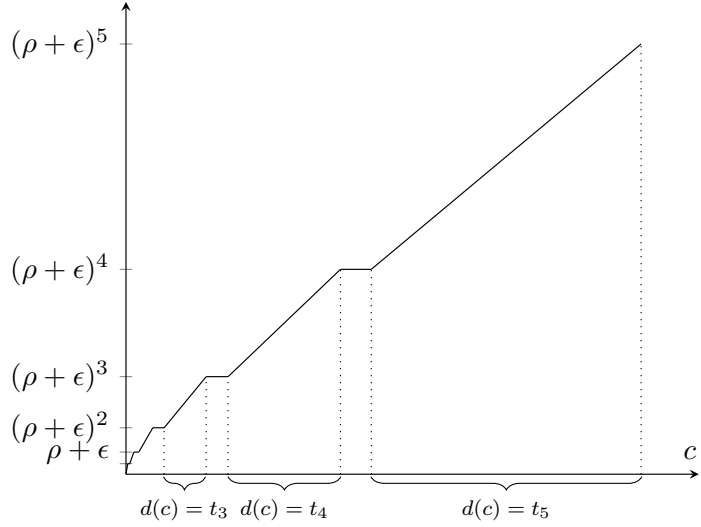

\section{Randomized Incremental Maximization}\label{sec:randIncMax}

We turn to analyzing randomized algorithms to solve the (discrete) \incmaxsep problem.
In contrast to deterministic algorithms, we do not compare the value obtained by the algorithm to an optimum solution, but rather the expected value obtained by the algorithm.
This enables us to find an algorithm with randomized competitive ratio smaller than the lower bound of $2.24$ on the competitive ratio of deterministic algorithms in Theorem~\ref{thm:det_lower_bound}.

\subsection{Randomized Algorithm}\label{sec:randAlg}

Scaling algorithms, i.e., algorithms where the size $c_i$ is chosen such that $c_i=\delta c_{i-1}$ with an appropriate scaling factor $\delta>1$, have been proven to perform well for the deterministic version of the problem.
The best known algorithm is, in fact, a scaling algorithm~\cite{BernsteinDisserGrossHimburg/20}.
In the analysis, it turns out that, on average, a scaling algorithm performs better than the actual competitive ratio, which is only tight for few sizes.
By randomizing the initial size~$c_0$, we manage to average out the worst-case sizes in the analysis.

We describe the randomized algorithm $\randalg$ for \incmaxsep.
For this, we consider the function
\begin{eqnarray*}
	g(x) & = & \frac{1-\sqrt{\bigl(\frac{x^{3}-1}{x-1}x^{z}-1\bigr)^{2}+4x^{5+2z}}}{2\log(x)x^{3+z}} -(1-\delta)\frac{1-x^{-3}}{x-1}+z -\frac{1-x^{-3}}{2(x-1)\log(x)} \\
	&  & -\Bigl(\frac{1-x^{-3}}{x-1}-\frac{1}{x^{3+z}}\Bigr)\Bigl(\log_{x}\Bigl(\sqrt{\bigl(\frac{x^{3}-1}{x-1}x^{z}-1\bigr)^{2}+4x^{5+2z}}-\frac{x^{3}-1}{x-1}x^{z}+1\Bigr)\\
	&  & -\log_{x}(2)-3\Bigr)-\frac{2x^{2+z}}{\Bigl(\sqrt{\bigl(\frac{x^{3}-1}{x-1}x^{z}-1\bigr)^{2}+4x^{5+2z}}-\frac{x^{3}-1}{x-1}x^{z}+1\Bigr)\log(x)} \\
	&  & +\frac{2}{\log(x)}-\Bigl(1+\frac{1}{x^{3+z}}\Bigr)\Bigl(\log_{x}(x^{3+z}+1)+\log_{x}(x-1)-\log_{x}(x^{4}-1)\Bigr)
\end{eqnarray*}
with $z=\log_{x}\Bigl( \frac{x^4-1}{x-1}-1 \Bigr)-3$.
Let $r\approx5.1646$ be the unique maximum of~$g$.
The algorithm \randalg starts by choosing $\epsilon\in(0,1)$ uniformly at random.
For all $i\in\N_{0}$, it calculates $\tilde{c}_{i} := r^{i+\epsilon}$ and $c_{i}:=\lfloor\tilde{c}_{i}\rfloor$ and returns the solution $(c_{0},c_{1},c_{2},\dots)$.\footnote{With this definition, the algorithm does not terminate on finite instances. To avoid this, it suffices to stop calculating the sizes~$c_{i}$ until they are larger than the number of elements in the instance.}
This approach is similar to a randomized algorithm to solve the \textsc{CowPath} problem in~\cite{KaoReifTate/93}, which also calculates such a sequence with a different choice of $r\in\R$ in order to explore a star graph.

We define
\begin{equation*}
\tilde{t}_{i}:=\sum_{j=0}^{i}\tilde{c}_{j}=r^{\epsilon}\frac{r^{i+1}-1}{r-1}
\quad\quad\quad \textrm{and} \quad\quad\quad
t_{i}:=\sum_{j=0}^{i}c_{j}.
\end{equation*}
For better readability, we let $\tilde{c}_{-1}=c_{-1}=\tilde{t}_{-1}=t_{-1}=0$.
Note that, for all $i\in\N_{0}$, we have
\begin{equation}
t_{i-1}\leq\tilde{t}_{i-1}=r^{\epsilon}\frac{r^{i}-1}{r-1}\overset{r>2}{\leq}r^{i+\epsilon}-r^{\epsilon}\leq r^{i+\epsilon}-1=\tilde{c}_{i}-1\leq c_{i}.\label{eq:t_i-1_leq_c_i}
\end{equation}
For every size $c\in\N_{0}$, we denote the solution created by the algorithm $\randalg$ by $X_{\alg}(c)$.
Note that the optimum solution of size $c\in\N_{0}$ is given by the set~$R_c$ because $v_1\leq v_2\leq \dots$ and $d_1\geq d_2\geq \dots$.
Thus, the value of the optimum solution of size~$c$ is $v_c$.

In order to find an upper bound on the randomized competitive ratio of \randalg, we need the following lemma. It gives an estimate on the expected value of the solution for a fixed size~$\capacity\in\N$ of \randalg depending on the interval in which~$\capacity$ falls.

\begin{lemma}\label{lem:estimates_ratio_alg_opt}
	Let $\capacity\in\N$.
	\begin{enumerate}
		\item For $i\in\N\cup\{0\}$ with $\mathbb{P}[\capacity\in(c_{i-1},c_{i}]]>0$, we have
		\[
		\E\bigl[f(X_{\alg}(\capacity))\mid\capacity\in(c_{i-1},c_{i}]\bigr]\geq\E\Bigl[\max\Bigl\{\frac{c_{i-1}}{\capacity},\frac{\capacity-t_{i-1}}{\max\{\capacity,c_{i}\}}\Bigr\}\Bigm|\capacity\in(c_{i-1},c_{i}]\Bigr]\cdot v_\capacity.
		\]
		
		\item For $i\in\N$ with $\mathbb{P}[\capacity\in(\tilde{c}_{i},\tilde{t}_{i}-1]]>0$, we have
		\[
		\E\bigl[f(X_{\alg}(\capacity))\mid\capacity\in(\tilde{c}_{i},\tilde{t}_{i}-1]\bigr]\geq\E\Bigl[1-\frac{\tilde{t}_{i-1}}{\capacity}\Bigm|\capacity\in(\tilde{c}_{i},\tilde{t}_{i}-1]\Bigr]\cdot v_\capacity.
		\]
		
		\item For $i\in\N$ with $\mathbb{P}[\capacity\in(\tilde{t}_{i-1}-1,\tilde{c}_{i}]]>0$, we have
		\[
		\E\bigl[f(X_{\alg}(\capacity))\mid\capacity\in(\tilde{t}_{i-1}-1,\tilde{c}_{i}]\bigr]\geq\E\Bigl[\max\Bigl\{\frac{\tilde{c}_{i-1}-1}{\capacity},\frac{\capacity-\tilde{t}_{i-1}}{\tilde{c}_{i}}\Bigr\}\Bigm|\capacity\in(\tilde{t}_{i-1}-1,\tilde{c}_{i}]\Bigr]\cdot v_\capacity.
		\]
	\end{enumerate}
\end{lemma}

\begin{proof}
	We prove $(i)$.
	For $\capacity\in(t_{i-1},c_{i}]$ with $i\in\N_{0}$, the algorithm $\randalg$ has finished adding the optimal solution of size $c_{i-1}$ and is adding the optimal solution of size~$c_{i}$.
	Thus, 
	\begin{eqnarray}
	f(X_{\alg}(\capacity)) & = & \max\{v_{c_{i-1}},(\capacity-t_{i-1})d_{c_{i}}\}\nonumber \\
	& = & \max\Bigl\{ c_{i-1}d_{c_{i-1}},(\capacity-t_{i-1})\frac{v_{c_{i}}}{c_{i}}\Bigr\}\nonumber \\
	& \overset{d\textrm{ dec.},v\textrm{ inc.}}{\geq} & \max\Bigl\{ c_{i-1}d_{\capacity},(\capacity-t_{i-1})\frac{v_{\capacity}}{c_{i}}\Bigr\}\nonumber \\
	& = & \max\Bigl\{\frac{c_{i-1}}{\capacity},\frac{\capacity-t_{i-1}}{c_{i}}\Bigr\}\cdot v_{\capacity}.\label{eq:lem_estimate_alg_opt_1}
	\end{eqnarray}
	Now, assume $\capacity\in(c_{i},t_{i}]$ with $i\in\N$.
	By~\eqref{eq:t_i-1_leq_c_i}, we have $t_{i-1}\leq c_{i}$, i.e., for size~$\capacity$, the algorithm $\randalg$ has finished adding the optimal solution of size~$c_{i-1}$ and is adding the optimal solution of size~$c_{i}$.
	Thus, we have
	\begin{eqnarray}
	f(X_{\alg}(\capacity)) & \geq & \max\{v_{c_{i-1}},(\capacity-t_{i-1})d_{c_{i}}\}\nonumber \\
	& = & \max\{c_{i-1}d_{c_{i-1}},(\capacity-t_{i-1})d_{c_{i}}\}\nonumber \\
	& \overset{d\textrm{ dec.}}{\geq} & \max\{c_{i-1}d_{\capacity},(\capacity-t_{i-1})d_{\capacity}\}\nonumber \\
	& = & \max\Bigl\{\frac{c_{i-1}}{\capacity},\frac{\capacity-t_{i-1}}{\capacity}\Bigr\}\cdot v_{\capacity}.\label{eq:lem_estimate_alg_opt_2}
	\end{eqnarray}
	Combining~\eqref{eq:lem_estimate_alg_opt_1} and~\eqref{eq:lem_estimate_alg_opt_2}, for $\capacity\in(c_{i-1},c_{i}]$, we obtain
	\[
	f(X_{\alg}(\capacity))\geq\max\Bigl\{\frac{c_{i-1}}{\capacity},\frac{\capacity-t_{i-1}}{\max\{\capacity,c_{i}\}}\Bigr\}\cdot v_{\capacity}.
	\]
	
	We prove $(ii)$.
	For $\capacity\in(\tilde{c}_{i},\tilde{t}_{i}-1]$ with $i\in\N$, we have $\capacity>\tilde{c}_{i}\geq c_{i}\overset{\eqref{eq:t_i-1_leq_c_i}}{\geq}t_{i-1}$, i.e., the algorithm $\randalg$ has finished adding the optimal solution of size~$c_{i-1}$ and is either adding the optimal solution of size~$c_{i}$ or has finished adding the optimal solution of size~$c_{i}$.
	Thus, we have
	\begin{eqnarray*}
		f(X_{\alg}(\capacity)) & \geq & \min\{(\capacity-t_{i-1})d_{c_{i}},v_{c_{i}}\}\\
		& \geq & \min\{\capacity-\tilde{t}_{i-1},c_{i}\}d_{c_{i}}\\
		& \overset{\capacity\leq\tilde{t}_{i}-1}{=} & (\capacity-\tilde{t}_{i-1})d_{c_{i}}\\
		& = & \Bigl(1-\frac{\tilde{t}_{i-1}}{\capacity}\Bigr)\frac{d_{c_{i}}}{d_{\capacity}}\cdot v_{\capacity}\\
		& \overset{d\textrm{ dec.}}{\geq} & \Bigl(1-\frac{\tilde{t}_{i-1}}{\capacity}\Bigr)\cdot v_{\capacity}.
	\end{eqnarray*}
	
	Lastly, we prove $(iii)$.
	For $\capacity\in(\tilde{t}_{i-1}-1,\tilde{c}_{i}]$ with $i\in\N$, we have $\capacity>\tilde{t}_{i-1}-1\geq t_{i-1}-1$.
	Since $\capacity,t_{i-1}\in\N$, this implies $\capacity\geq t_{i-1}$, i.e., the algorithm $\randalg$ has finished adding the optimal solution of size~$c_{i-1}$ and is adding size~$c_{i}$.
	Thus, we have
	\begin{eqnarray*}
		f(X_{\alg}(\capacity)) & = & \max\{v_{c_{i-1}},(\capacity-t_{i-1})d_{c_{i}}\}\\
		& = & \max\Bigl\{ c_{i-1}d_{c_{i-1}},(\capacity-t_{i-1})\frac{v_{c_{i}}}{c_{i}}\Bigr\}\\
		& \overset{d\textrm{ dec.},v\textrm{ inc.}}{\geq} & \max\Bigl\{ c_{i-1}d_{\capacity},(\capacity-t_{i-1})\frac{v_{\capacity}}{c_{i}}\Bigr\}\\
		& = & \max\Bigl\{\frac{c_{i-1}}{\capacity},\frac{\capacity-t_{i-1}}{c_{i}}\Bigr\}\cdot v_{\capacity}\\
		& \geq & \max\Bigl\{\frac{\tilde{c}_{i-1}-1}{\capacity},\frac{\capacity-\tilde{t}_{i-1}}{c_{i}}\Bigr\}\cdot v_{\capacity},
	\end{eqnarray*}
	where we used the fact that $\capacity\leq\tilde{c}_{i}$ and $\capacity\in\N$ and thus $\capacity\leq c_{i}=\lfloor\tilde{c}_{i}\rfloor$.
\end{proof}

In the analysis of the algorithm, additionally, the following estimate is needed.

\begin{lemma}\label{lem:estimate_sum-integrals_g(r)}
	Let $k\in\N$ and $\delta\in(0,1]$ such that $r^{k+\delta}\geq\sum_{i=0}^{3}r^{i}$.
	Then
	\begin{alignat*}{3}
	g(r) \leq I(k,\delta) := &\int_{\min\bigl\{1,\mu(k-1)\bigr\}}^{1}1-\frac{\tilde{t}_{k-2}}{r^{k+\delta}}\,d\epsilon
	&& +\int_{\min\{1,\nu(k-1)\}}^{\min\{1,\mu(k-1)\}}\frac{\tilde{c}_{k-1}-1}{r^{k+\delta}}\,d\epsilon\\
	& +\int_{\delta}^{\min\{1,\nu(k-1)\}}\frac{r^{k+\delta}-\tilde{t}_{k-1}}{\tilde{c}_{k}}\,d\epsilon
	&& +\int_{\max\bigl\{0,\mu(k)\bigr\}}^{\delta}1-\frac{\tilde{t}_{k-1}}{r^{k+\delta}}\,d\epsilon\\
	& +\int_{\max\{0,\nu(k)\}}^{\max\{0,\mu(k)\}}\frac{\tilde{c}_{k}-1}{r^{k+\delta}}\,d\epsilon
	&& +\int_{0}^{\max\{0,\nu(k)\}}\frac{r^{k+\delta}-\tilde{t}_{k}}{\tilde{c}_{k+1}}\,d\epsilon
	\end{alignat*}
	where
	\begin{align*}
	\mu(i) = & \log_{r}(r^{k+\delta}+1)+\log_{r}(r-1)-\log_{r}(r^{i+1}-1),\\
	\nu(i) = & \log_{r}\Bigl(\sqrt{\Bigl(r^{k+\delta}\frac{1-r^{-(i+1)}}{r-1}\!-\!1\Bigr)^{2}\!+4r^{2k+2\delta-1}}-r^{k+\delta}\frac{1-r^{-(i+1)}}{r-1}+1\Bigr)\!-\!\log_{r}(2)-i.
	\end{align*}
\end{lemma}

\begin{proof}
	Depending on $\delta$ and $k$, only 3 or 4 of these integrals are non-zero.
	We distinguish between the different possibilities.
	
	Case 1 ($\nu(k)>0$):
	Then
	\begin{eqnarray*}
		& & \nu(k-1)+\log_{r}(2)+(k-1)\\
		& = & \log_{r}\Bigl(\sqrt{\Bigl(r^{k+\delta}\frac{1-r^{-k}}{r-1}-1\Bigr)^{2}+4r^{2k+2\delta-1}}-r^{k+\delta}\frac{1-r^{-k}}{r-1}+1\Bigr)\\
		& \geq & \log_{r}\Bigl(\sqrt{\Bigl(r^{k+\delta}\frac{1-r^{-(k+1)}}{r-1}-1\Bigr)^{2}+4r^{2k+2\delta-1}}-r^{k+\delta}\frac{1-r^{-(k+1)}}{r-1}+1\Bigr)\\
		& = & \nu(k)+\log_{r}(2)+k,
	\end{eqnarray*}
	i.e., we have $\nu(k-1)\geq\nu(k)+1>1$.
	For the inequality, we used the fact that, for $a>b>0$ and $x>0$, we have $\sqrt{b^{2}+x}-b>\sqrt{a^{2}+x}-a$.
	We obtain
	\begin{eqnarray*}
		&  & I(k,\delta)\\
		& = & \int_{\delta}^{1}\frac{\bar{c}-\frac{r^{k}-1}{r-1}r^{\epsilon}}{r^{k+\epsilon}}\mathrm{d}\epsilon\\
		&  & +\int_{\log_{r}(r^{k+\delta}+1)+\log_{r}(r-1)-\log_{r}(r^{k+1}-1)}^{\delta}1-\frac{\frac{r^{k}-1}{r-1}r^{\epsilon}}{\bar{c}}\mathrm{d}\epsilon\\
		&  & +\int_{\log_{r}\Bigl(\sqrt{\bigl(\frac{r^{k}-r^{-1}}{r-1}r^{\delta}-1\bigr)^{2}+4r^{2k+2\delta-1}}-\frac{r^{k}-r^{-1}}{r-1}r^{\delta}+1\Bigr)-\log_{r}(2)-k}^{\log_{r}(r^{k+\delta}+1)+\log_{r}(r-1)-\log_{r}(r^{k+1}-1)}\frac{r^{k+\epsilon}-1}{\bar{c}}\mathrm{d}\epsilon\\
		&  & +\int_{0}^{\log_{r}\Bigl(\sqrt{\bigl(\frac{r^{k}-r^{-1}}{r-1}r^{\delta}-1\bigr)^{2}+4r^{2k+2\delta-1}}-\frac{r^{k}-r^{-1}}{r-1}r^{\delta}+1\Bigr)-\log_{r}(2)-k}\frac{\bar{c}-\frac{r^{k+1}-1}{r-1}r^{\epsilon}}{r^{k+1+\epsilon}}\mathrm{d}\epsilon\\
		& = & \Bigl[-\frac{\bar{c}}{\log(r)r^{k+\epsilon}}-\frac{1-r^{-k}}{r-1}\epsilon\Bigr]_{\delta}^{1} + \Bigl[\epsilon-\frac{(r^{k}-1)r^{\epsilon}}{\log(r)(r-1)\bar{c}}\Bigr]_{\log_{r}(r^{k+\delta}+1)+\log_{r}(r-1)-\log_{r}(r^{k+1}-1)}^{\delta}\\
		&  & +\Bigl[\frac{r^{k+\epsilon}}{\log(r)\bar{c}}-\frac{\epsilon}{\bar{c}}\Bigr]_{\log_{r}\Bigl(\sqrt{\bigl(\frac{r^{k}-r^{-1}}{r-1}r^{\delta}-1\bigr)^{2}+4r^{2k+2\delta-1}}-\frac{r^{k}-r^{-1}}{r-1}r^{\delta}+1\Bigr)-\log_{r}(2)-k}^{\log_{r}(r^{k+\delta}+1)+\log_{r}(r-1)-\log_{r}(r^{k+1}-1)}\\
		&  & +\Bigl[-\frac{\bar{c}}{\log(r)r^{k+1+\epsilon}}-\frac{1-r^{-(k+1)}}{r-1}\epsilon\Bigr]_{0}^{\log_{r}\Bigl(\sqrt{\bigl(\frac{r^{k}-r^{-1}}{r-1}r^{\delta}-1\bigr)^{2}+4r^{2k+2\delta-1}}-\frac{r^{k}-r^{-1}}{r-1}r^{\delta}+1\Bigr)-\log_{r}(2)-k}\\
		& = & \delta+\frac{2+r^{-(k+\delta)}}{\log(r)}-\frac{2r^{k-1+\delta}}{\log(r)\Bigl(\sqrt{\bigl(\frac{r^{k}-r^{-1}}{r-1}r^{\delta}-1\bigr)^{2}+4r^{2k+2\delta-1}}-\frac{r^{k}-r^{-1}}{r-1}r^{\delta}+1\Bigr)}\\
		&  & -\frac{\sqrt{\bigl(\frac{r^{k}-r^{-1}}{r-1}r^{\delta}-1\bigr)^{2}+4r^{2k+2\delta-1}}-\frac{r^{k}-r^{-1}}{r-1}r^{\delta}+1}{2\log(r)r^{k+\delta}}-\Bigl(\frac{1-r^{-(k+1)}}{r-1}-\frac{1}{r^{k+\delta}}\Bigr)\\
		&  & \Bigl(\log_{r}\Bigl(\sqrt{\bigl(\frac{r^{k}-r^{-1}}{r-1}r^{\delta}-1\bigr)^{2}+4r^{2k+2\delta-1}}-\frac{r^{k}-r^{-1}}{r-1}r^{\delta}+1\Bigr)-\log_{r}(2)-k\Bigr)\\
		& & -\frac{1-r^{-k}}{r-1}\bigl(1-\delta+\frac{1}{\log(r)}\bigr)-\Bigl(1+\frac{1}{r^{k+\delta}}\Bigr)\bigl(\log_{r}(r^{k+\delta}+1)+\log_{r}(r-1)-\log_{r}(r^{k+1}-1)\bigr)\\
		& =: & f_{1}(k,\delta).
	\end{eqnarray*}
	By analyzing this expression, one can see that it is minimized for $k$ and $\delta$ small.
	
	We have $\nu(k)>0$, which only holds for $\delta>\log_{r}\bigl(\frac{r-r^{-k}}{2(r-1)}+\sqrt{\bigl(\frac{r-r^{-k}}{2(r-1)}\bigr)^{2}+r+r^{1-k}}\bigr)$.
	Thus, $f_{1}(k,\delta)$ is minimized for $k=3$ and $\smash{\delta=\log_{r}\bigl(\frac{r-r^{-3}}{2(r-1)}+\sqrt{\bigl(\frac{r-r^{-3}}{2(r-1)}\bigr)^{2}+r+r^{-2}}\bigr)}$.
	Calculating the value, we obtain
	\[
	I(k,\delta)\geq f_{1}(k,\delta)\geq f_{1}\Bigl(3,\log_{r}\Bigl(\frac{r-r^{-3}}{2(r-1)}+\sqrt{\bigl(\frac{r-r^{-3}}{2(r-1)}\bigr)^{2}+r+r^{-2}}\Bigr)\Bigr)>0.566>g(r).
	\]
	
	Case 2 ($\nu(k)\leq0$ and $\nu(k-1)>1$):
	Then
	\begin{eqnarray*}
		&  & I(k,\delta)\\
		& = & \int_{\delta}^{1}\frac{\bar{c}-\frac{r^{k}-1}{r-1}r^{\epsilon}}{r^{k+\epsilon}}\mathrm{d}\epsilon\\
		&  & +\int_{\log_{r}(r^{k+\delta}+1)+\log_{r}(r-1)-\log_{r}(r^{k+1}-1)}^{\delta}1-\frac{\frac{r^{k}-1}{r-1}r^{\epsilon}}{\bar{c}}\mathrm{d}\epsilon\\
		&  & +\int_{0}^{\log_{r}(r^{k+\delta}+1)+\log_{r}(r-1)-\log_{r}(r^{k+1}-1)}\frac{r^{k+\epsilon}-1}{\bar{c}}\mathrm{d}\epsilon\\
		& = & \Bigl[-\frac{\bar{c}}{\log(r)r^{k+\epsilon}}-\frac{1-r^{-k}}{r-1}\epsilon\Bigr]_{\delta}^{1}\\
		&  & +\Bigl[\epsilon-\frac{(r^{k}-1)r^{\epsilon}}{\log(r)(r-1)\bar{c}}\Bigr]_{\log_{r}(r^{k+\delta}+1)+\log_{r}(r-1)-\log_{r}(r^{k+1}-1)}^{\delta}\\
		&  & +\Bigl[\frac{r^{k+\epsilon}}{\log(r)\bar{c}}-\frac{\epsilon}{\bar{c}}\Bigr]_{0}^{\log_{r}(r^{k+\delta}+1)+\log_{r}(r-1)-\log_{r}(r^{k+1}-1)}\\
		& = & -\frac{r^{\delta}}{\log(r)r}-\frac{1-r^{-k}}{r-1}(1-\delta+\frac{1}{\log(r)})+\frac{2+r^{-(k+\delta)}}{\log(r)}-\frac{1}{\log(r)r^{\delta}}\\
		&  & +\delta-\bigl(1+\frac{1}{r^{k+\delta}}\bigr)\bigl(\log_{r}(r^{k+\delta}+1)+\log_{r}(r-1)-\log_{r}(r^{k+1}-1)\bigr)\\
		& =: & f_{2}(k,\delta)
	\end{eqnarray*}
	By analyzing this expression, one can see that it is minimized for $k$ small and $\delta$ large.
	
	We have $\nu(k)\leq0$ and $\nu(k-1)>1$, which only holds for $\delta\in(0,1]$ satisfying
	\[
	\delta > \log_{r}\Bigl(\frac{r^{2}-r^{2-k}}{2(r-1)}+\sqrt{\bigl(\frac{r^{2}-r^{2-k}}{2(r-1)}\bigr)^{2}+r^{3}+r^{3-k}}\Bigr)
	\]
	and
	\[
	\delta \leq \log_{r}\Bigl(\frac{r-r^{-k}}{2(r-1)}+\sqrt{\bigl(\frac{r-r^{-k}}{2(r-1)}\bigr)^{2}+r+r^{1-k}}\Bigr).
	\]
	Thus, $f_{2}(k,\delta)$ is minimized for $k=3$ and $\delta=\log_{r}\Bigl(\frac{r-r^{-3}}{2(r-1)}+\sqrt{\bigl(\frac{r-r^{-3}}{2(r-1)}\bigr)^{2}+r+r^{-2}}\Bigr)$.
	Calculating the value, we obtain
	\[
	I(k,\delta)\geq f_{2}(k,\delta)\geq f_{2}\Bigl(3,\log_{r}\Bigl(\frac{r-r^{-3}}{2(r-1)}+\sqrt{\bigl(\frac{r-r^{-3}}{2(r-1)}\bigr)^{2}+r+r^{-2}}\Bigr)\Bigr)>0.566>g(r).
	\]
	
	Case 3 ($\nu(k-1)\leq1$ and $\mu(k)>0$):
	Then
	\begin{eqnarray}
	\mu(k-1) & = & \mu(k)+\log_{r}(r^{k+1}-1)-\log_{r}(r^{k}-1)\nonumber \\
	& = & \mu(k)+1+\log_{r}\Bigl(\frac{r^{k}-r^{-1}}{r^{k}-1}\Bigr)\nonumber \\
	& \geq & \mu(k)+1\nonumber \\
	& > & 1.\label{eq:estimate_muk-i_muk}
	\end{eqnarray}
	We obtain
	\begin{eqnarray*}
		&  & I(k,\delta)\\
		& = & \int_{\log_{r}\Bigl(\sqrt{\bigl(\frac{r^{k}-1}{r-1}r^{\delta}-1\bigr)^{2}+4r^{2k+2\delta-1}}-\frac{r^{k}-1}{r-1}r^{\delta}+1\Bigr)-\log_{r}(2)-k+1}^{1}\frac{r^{k-1+\epsilon}-1}{\bar{c}}\mathrm{d}\epsilon\\
		&  & +\int_{\delta}^{\log_{r}\Bigl(\sqrt{\bigl(\frac{r^{k}-1}{r-1}r^{\delta}-1\bigr)^{2}+4r^{2k+2\delta-1}}-\frac{r^{k}-1}{r-1}r^{\delta}+1\Bigr)-\log_{r}(2)-k+1}\frac{\bar{c}-\frac{r^{k}-1}{r-1}r^{\epsilon}}{r^{k+\epsilon}}\mathrm{d}\epsilon\\
		&  & +\int_{\log_{r}(r^{k+\delta}+1)+\log_{r}(r-1)-\log_{r}(r^{k+1}-1)}^{\delta}1-\frac{\frac{r^{k}-1}{r-1}r^{\epsilon}}{\bar{c}}\mathrm{d}\epsilon\\
		&  & +\int_{0}^{\log_{r}(r^{k+\delta}+1)+\log_{r}(r-1)-\log_{r}(r^{k+1}-1)}\frac{r^{k+\epsilon}-1}{\bar{c}}\mathrm{d}\epsilon\\
		& = & \Bigl[\frac{r^{k-1+\epsilon}}{\log(r)\bar{c}}-\frac{\epsilon}{\bar{c}}\Bigr]_{\log_{r}\Bigl(\sqrt{\bigl(\frac{r^{k}-1}{r-1}r^{\delta}-1\bigr)^{2}+4r^{2k+2\delta-1}}-\frac{r^{k}-1}{r-1}r^{\delta}+1\Bigr)-\log_{r}(2)-k+1}^{1}\\
		&  & +\Bigl[-\frac{\bar{c}}{\log(r)r^{k+\epsilon}}-\frac{1-r^{-k}}{r-1}\epsilon\Bigr]_{\delta}^{\log_{r}\Bigl(\sqrt{\bigl(\frac{r^{k}-1}{r-1}r^{\delta}-1\bigr)^{2}+4r^{2k+2\delta-1}}-\frac{r^{k}-1}{r-1}r^{\delta}+1\Bigr)-\log_{r}(2)-k+1}\\
		&  & +\Bigl[\epsilon-\frac{(r^{k}-1)r^{\epsilon}}{\log(r)(r-1)\bar{c}}\Bigr]_{\log_{r}(r^{k+\delta}+1)+\log_{r}(r-1)-\log_{r}(r^{k+1}-1)}^{\delta}\\
		&  & +\Bigl[\frac{r^{k+\epsilon}}{\log(r)\bar{c}}-\frac{\epsilon}{\bar{c}}\Bigr]_{0}^{\log_{r}(r^{k+\delta}+1)+\log_{r}(r-1)-\log_{r}(r^{k+1}-1)}\\
		& = & -\frac{1-r^{-k}}{2(r-1)\log(r)}-\frac{2r^{k-1+\delta}}{\Bigl(\sqrt{\bigl(\frac{r^{k}-1}{r-1}r^{\delta}-1\bigr)^{2}+4r^{2k+2\delta-1}}-\frac{r^{k}-1}{r-1}r^{\delta}+1\Bigr)\log(r)}\\
		&  & +\frac{1-\sqrt{\bigl(\frac{r^{k}-1}{r-1}r^{\delta}-1\bigr)^{2}+4r^{2k+2\delta-1}}}{2\log(r)r^{k+\delta}}-\Bigl(\frac{1-r^{-k}}{r-1}-\frac{1}{r^{k+\delta}}\Bigr)-(1-\delta)\frac{1-r^{-k}}{r-1}+\delta\\
		& & \bigl(\log_{r}\Bigl(\sqrt{\bigl(\frac{r^{k}-1}{r-1}r^{\delta}-1\bigr)^{2}+4r^{2k+2\delta-1}}-\frac{r^{k}-1}{r-1}r^{\delta}+1\Bigr)-\log_{r}(2)-k\bigr)\\
		&  & +\frac{2}{\log(r)}-\Bigl(1+\frac{1}{r^{k+\delta}}\Bigr)\Bigl(\log_{r}(r^{k+\delta}+1)+\log_{r}(r-1)-\log_{r}(r^{k+1}-1)\Bigr)\\
		& & \\
		& =: & f_{3}(k,\delta).
	\end{eqnarray*}
	By analyzing this expression, one can see that it is minimized for $k$ and $\delta$ small.
	
	We have $\nu(k-1)\leq1$ and $\mu(k)>0$, which only holds for
	\[
	\log_{r}\Bigl(\frac{r^{k+1}-1}{r-1}-1\Bigr)-k<\delta\leq\log_{r}\Bigl(\frac{r^{2}-r^{2-k}}{2(r-1)}+\sqrt{\bigl(\frac{r^{2}-r^{2-k}}{2(r-1)}\bigr)^{2}+r^{3}+r^{3-k}}\Bigr).
	\]
	Thus, $f_{3}(k,\delta)$ is minimized for $k=3$ and $\delta=\log_{r}\bigl(\frac{r^{4}-1}{r-1}-1\bigr)-3$.
	Calculating the value, we obtain
	\[
	I(k,\delta)\geq f_{3}(k,\delta)\geq f_{3}\Bigl(3,\log_{r}\Bigl(\frac{r^{4}-1}{r-1}-1\Bigr)-3\Bigr)=g(r).
	\]
	
	Case 4 ($\mu(k)\leq0$ and $\mu(k-1)>1$):
	Then
	\begin{eqnarray*}
		&  & I(k,\delta)\\
		& = & \int_{\log_{r}\Bigl(\sqrt{\bigl(\frac{r^{k}-1}{r-1}r^{\delta}-1\bigr)^{2}+4r^{2k+2\delta-1}}-\frac{r^{k}-1}{r-1}r^{\delta}+1\Bigr)-\log_{r}(2)-k+1}^{1}\frac{r^{k-1+\epsilon}-1}{\bar{c}}\mathrm{d}\epsilon\\
		&  & +\int_{\delta}^{\log_{r}\Bigl(\sqrt{\bigl(\frac{r^{k}-1}{r-1}r^{\delta}-1\bigr)^{2}+4r^{2k+2\delta-1}}-\frac{r^{k}-1}{r-1}r^{\delta}+1\Bigr)-\log_{r}(2)-k+1}\frac{\bar{c}-\frac{r^{k}-1}{r-1}r^{\epsilon}}{r^{k+\epsilon}}\mathrm{d}\epsilon\\
		&  & +\int_{0}^{\delta}1-\frac{\frac{r^{k}-1}{r-1}r^{\epsilon}}{\bar{c}}\mathrm{d}\epsilon\\
		& = & \Bigl[\frac{r^{k-1+\epsilon}}{\log(r)\bar{c}}-\frac{\epsilon}{\bar{c}}\Bigr]_{\log_{r}\Bigl(\sqrt{\bigl(\frac{r^{k}-1}{r-1}r^{\delta}-1\bigr)^{2}+4r^{2k+2\delta-1}}-\frac{r^{k}-1}{r-1}r^{\delta}+1\Bigr)-\log_{r}(2)-k+1}^{1}\\
		&  & +\Bigl[-\frac{\bar{c}}{\log(r)r^{k+\epsilon}}-\frac{1-r^{-k}}{r-1}\epsilon\Bigr]_{\delta}^{\log_{r}\Bigl(\sqrt{\bigl(\frac{r^{k}-1}{r-1}r^{\delta}-1\bigr)^{2}+4r^{2k+2\delta-1}}-\frac{r^{k}-1}{r-1}r^{\delta}+1\Bigr)-\log_{r}(2)-k+1}\\
		&  & +\Bigl[\epsilon-\frac{(r^{k}-1)r^{\epsilon}}{\log(r)(r-1)\bar{c}}\Bigr]_{0}^{\delta}\\
		& = & \frac{1}{\log(r)r^{\delta}}-\frac{1}{r^{k+\delta}}-\frac{1}{2\log(r)r^{k+\delta}}+\frac{1}{\log(r)}+\frac{1-r^{-k}}{r-1}\delta\\
		&  & -\frac{2r^{k-1+\delta}}{\log(r)\Bigl(\sqrt{\bigl(\frac{r^{k}-1}{r-1}r^{\delta}-1\bigr)^{2}+4r^{2k+2\delta-1}}-\frac{r^{k}-1}{r-1}r^{\delta}+1\Bigr)}-\Bigl(\frac{1-r^{-k}}{r-1}-\frac{1}{r^{k+\delta}}\Bigr)\\
		&  & \Bigl(\log_{r}\Bigl(\sqrt{\bigl(\frac{r^{k}-1}{r-1}r^{\delta}-1\bigr)^{2}+4r^{2k+2\delta-1}}-\frac{r^{k}-1}{r-1}r^{\delta}+1\Bigr)-\log_{r}(2)-k+1\Bigr)\\
		& & -\frac{\sqrt{\bigl(\frac{r^{k}-1}{r-1}r^{\delta}-1\bigr)^{2}+4r^{2k+2\delta-1}}}{2\log(r)r^{k+\delta}}+\delta-\bigl(\frac{1}{2}-\frac{1}{r^{\delta}}\bigr)\frac{1-r^{-k}}{\log(r)(r-1)}\\
		& =: & f_{4}(k,\delta).
	\end{eqnarray*}
	This expression is minimized for $k$ small.
	
	Note that we have $0\geq\mu(k)=\log_{r}(r^{k+\delta}+1)+\log_{r}(r-1)-\log_{r}(r^{k+1}-1)$, which is equivalent to
	\[
	r^{k+\delta}+1\leq\frac{r^{k+1}-1}{r-1}=\sum_{i=0}^{k}r^{i}.
	\]
	Since $r^{k+\delta}\geq\sum_{i=0}^{3}r^{i}$, we have $k\geq4$, i.e., $f_{4}(k,\delta)$ is minimized for $k=4$.
	By analyzing this function for $k=4$, we obtain
	\[
	I(k,\delta)\geq f_{4}(k,\delta)\geq f_{4}(4,\delta)>0.566>g(r).
	\]
	
	Case 5 ($\mu(k-1)\leq1$):
	Then
	\begin{eqnarray*}
		&  & I(k,\delta)\\
		& = & \int_{\log_{r}(r^{k+\delta}+1)+\log_{r}(r-1)-\log_{r}(r^{k}-1)}^{1}1-\frac{\frac{r^{k-1}-1}{r-1}r^{\epsilon}}{\bar{c}}\mathrm{d}\epsilon\\
		&  & +\int_{\log_{r}\Bigl(\sqrt{\bigl(\frac{r^{k}-1}{r-1}r^{\delta}-1\bigr)^{2}+4r^{2k+2\delta-1}}-\frac{r^{k}-1}{r-1}r^{\delta}+1\Bigr)-\log_{r}(2)-k+1}^{\log_{r}(r^{k+\delta}+1)+\log_{r}(r-1)-\log_{r}(r^{k}-1)}\frac{r^{k-1+\epsilon}-1}{\bar{c}}\mathrm{d}\epsilon\\
		&  & +\int_{\delta}^{\log_{r}\Bigl(\sqrt{\bigl(\frac{r^{k}-1}{r-1}r^{\delta}-1\bigr)^{2}+4r^{2k+2\delta-1}}-\frac{r^{k}-1}{r-1}r^{\delta}+1\Bigr)-\log_{r}(2)-k+1}\frac{\bar{c}-\frac{r^{k}-1}{r-1}r^{\epsilon}}{r^{k+\epsilon}}\mathrm{d}\epsilon\\
		&  & +\int_{0}^{\delta}1-\frac{\frac{r^{k}-1}{r-1}r^{\epsilon}}{\bar{c}}\mathrm{d}\epsilon\\
		& = & \Bigl[\epsilon-\frac{(r^{k-1}-1)r^{\epsilon}}{\log(r)(r-1)\bar{c}}\Bigr]_{\log_{r}(r^{k+\delta}+1)+\log_{r}(r-1)-\log_{r}(r^{k}-1)}^{1}\\
		&  & +\Bigl[\frac{r^{k-1+\epsilon}}{\log(r)\bar{c}}-\frac{\epsilon}{\bar{c}}\Bigr]_{\log_{r}\Bigl(\sqrt{\bigl(\frac{r^{k}-1}{r-1}r^{\delta}-1\bigr)^{2}+4r^{2k+2\delta-1}}-\frac{r^{k}-1}{r-1}r^{\delta}+1\Bigr)-\log_{r}(2)-k+1}^{\log_{r}(r^{k+\delta}+1)+\log_{r}(r-1)-\log_{r}(r^{k}-1)}\\
		&  & +\Bigl[-\frac{\bar{c}}{\log(r)r^{k+\epsilon}}-\frac{1-r^{-k}}{r-1}\epsilon\Bigr]_{\delta}^{\log_{r}\Bigl(\sqrt{\bigl(\frac{r^{k}-1}{r-1}r^{\delta}-1\bigr)^{2}+4r^{2k+2\delta-1}}-\frac{r^{k}-1}{r-1}r^{\delta}+1\Bigr)-\log_{r}(2)-k+1}\\
		&  & +\Bigl[\epsilon-\frac{(r^{k}-1)r^{\epsilon}}{\log(r)(r-1)\bar{c}}\Bigr]_{0}^{\delta}\\
		& = & 1-\Bigl(1+\frac{1}{\bar{c}}\Bigr)\bigl(\log_{r}(r^{k+\delta}+1)+\log_{r}(r-1)-\log_{r}(r^{k}-1)\bigr)-\frac{\sqrt{\bigl(\frac{r^{k}-1}{r-1}r^{\delta}-1\bigr)^{2}+4r^{2k+2\delta-1}}}{2\log(r)\bar{c}}\\
		&  & +\delta+\frac{2}{\log(r)}+\frac{\log_{r}\Bigl(\sqrt{\bigl(\frac{r^{k}-1}{r-1}r^{\delta}-1\bigr)^{2}+4r^{2k+2\delta-1}}-\frac{r^{k}-1}{r-1}r^{\delta}+1\Bigr)-\log_{r}(2)-k+1}{\bar{c}}\\
		&  & -\log_{r}(2)-\frac{2\bar{c}}{\Bigl(\sqrt{\bigl(\frac{r^{k}-1}{r-1}r^{\delta}-1\bigr)^{2}+4r^{2k+2\delta-1}}-\frac{r^{k}-1}{r-1}r^{\delta}+1\Bigr)\log(r)r}+\frac{3}{2\log(r)\bar{c}}\\
		&  & -\frac{1-r^{-k}}{r-1}\bigl(\log_{r}\Bigl(\sqrt{\bigl(\frac{r^{k}-1}{r-1}r^{\delta}-1\bigr)^{2}+4r^{2k+2\delta-1}}-\frac{r^{k}-1}{r-1}r^{\delta}+1\Bigr)+\frac{1}{2\log(r)}-k+1-\delta\bigr)\\
		& =: & f_{5}(k,\delta).
	\end{eqnarray*}
	This expression is minimized for $k$ small and $\delta$ large.
	By~\eqref{eq:estimate_muk-i_muk}, we have \mbox{$\mu(k)\leq\mu(k-1)-1\leq0$}.
	Thus, similarly to case 4, we have $k\geq4$.
	
	We have $\nu(k-1)\leq1$, which only holds for $\smash{\delta\leq\log_{r}\bigl(r\frac{r^{k}-1}{r-1}-1\bigr)-k}$.
	Thus, $f_{5}(k,\delta)$ is minimized for $k=4$ and $\smash{\delta=\log_{r}\bigl(r\frac{r^{4}-1}{r-1}-1\bigr)-4}$.
	Calculating the value, we obtain
	\[
	I(k,\delta)\geq f_{5}(k,\delta)\geq f_{5}\Bigl(4,\log_{r}\Bigl(r\frac{r^{4}-1}{r-1}-1\Bigr)-4\Bigr)>0.566>g(r).\qedhere
	\]
\end{proof}

With these lemmas, we are ready to prove an upper bound of $1/g(r)<1.772$ on the randomized competitive ratio of \randalg.

\randUB*

\begin{proof}
	Let $(c_1,c_2,\dots)$ denote the solution obtained by \randalg.
	Let $\capacity\in\N$.
	For $\smash{\capacity<\sum_{i=0}^{3}r^{i}}$, we have $\capacity\leq c_{3}\leq t_{3}$.
	Let $S_{4}:=\{(\lfloor r^{x}\rfloor,\lfloor r^{1+x}\rfloor,\lfloor r^{2+x}\rfloor,\lfloor r^{3+x}\rfloor)\mid x\in(0,1)\}$ denote the set of all sequences that are possible for the sequence $(c_{0},c_{1},c_{2},c_{3})$.
	We denote $\sigma\in S_{4}$ as $\sigma=(\sigma_{0},\sigma_{1},\sigma_{2},\sigma_{3})$.
	Using $c_{-1}=0$, we have
	\begin{eqnarray*}
		&  & \frac{\E\bigl[f(X_{\alg}(\capacity))]}{v(\capacity)}\\
		& = & \frac{1}{v(\capacity)}\sum_{i=0}^{3}\E\bigl[f(X_{\alg}(\capacity))\mid\capacity\in(c_{i-1},c_{i}]\bigr]\cdot\mathbb{P}\bigl[\capacity\in(c_{i-1},c_{i}]\bigr]\\
		& \overset{\textrm{Lemma }\ref{lem:estimates_ratio_alg_opt}(i)}{\geq} & \sum_{i=0}^{3}\E\Bigl[\max\Bigl\{\frac{c_{i-1}}{\capacity},\frac{\capacity-t_{i-1}}{\max\{\capacity,c_{i}\}}\Bigr\}\Bigm|\capacity\in(c_{i-1},c_{i}]\Bigr]\cdot\mathbb{P}\bigl[\capacity\in(c_{i-1},c_{i}]\bigr]\\
		& = & \sum_{i=0}^{3}\mathbb{P}\bigl[\capacity\in(c_{i-1},c_{i}]\bigr] \sum_{\sigma\in S_{4}}\max\Bigl\{\frac{\sigma_{i-1}}{\capacity},\frac{\capacity-\sigma_{0}-\dots-\sigma_{i-1}}{\max\{\capacity,\sigma_{i}\}}\Bigr\}\\
		& & \quad\quad \cdot\frac{\mathbb{P}\bigl[\bigl((c_{0},c_{1},c_{2},c_{3})=\sigma\bigr)\cap\bigl(\capacity\in(c_{i-1},c_{i}]\bigr)\bigr]}{\mathbb{P}\bigl[\capacity\in(c_{i-1},c_{i}]\bigr]}\\
		& = & \sum_{i=0}^{3}\sum_{\sigma\in S_{4}}\max\Bigl\{\frac{\sigma_{i-1}}{\capacity},\frac{\capacity-\sigma_{0}-\dots-\sigma_{i-1}}{\max\{\capacity,\sigma_{i}\}}\Bigr\}\\
		& & \quad\quad \cdot\mathbb{P}\bigl[\bigl((c_{0},c_{1},c_{2},c_{3})=\sigma\bigr)\cap\bigl(\capacity\in(c_{i-1},c_{i}]\bigr)\bigr]\\
		& = & \sum_{i=0}^{3}\sum_{\substack{\sigma\in S_{4}\\
				\sigma_{i-1}<\capacity\leq\sigma_{i}
			}
		}\max\Bigl\{\frac{\sigma_{i-1}}{\capacity},\frac{\capacity-\sigma_{0}-\dots-\sigma_{i-1}}{\max\{\capacity,\sigma_{i}\}}\Bigr\}\cdot\mathbb{P}\bigl[(c_{0},c_{1},c_{2},c_{3})=\sigma\bigr].
	\end{eqnarray*}
	For $\sigma=(\sigma_{0},\sigma_{1},\sigma_{2},\sigma_{3})\in S_{4}$, we have
	\begin{eqnarray*}
		\mathbb{P}\bigl[(c_{0},c_{1},c_{2},c_{3})=\sigma\bigr] & = & \mathbb{P}\Bigl[\bigcap_{i\in\{0,1,2,3\}}\bigl(r^{i+\epsilon}\in[\sigma_{i},\sigma_{i}+1)\bigr)\Bigr]\\
		& = & \mathbb{P}\Bigl[\bigcap_{i\in\{0,1,2,3\}}\bigl(\epsilon\in[\log_{r}(\sigma_{i})-i,\log_{r}(\sigma_{i}+1)-i)\bigr)\Bigr]\\
		& = & \mathbb{P}\Bigl[\epsilon\in\bigl[\max_{i\in\{0,1,2,3\}}(\log_{r}(\sigma_{i})-i),\min_{i\in\{0,1,2,3\}}(\log_{r}(\sigma_{i}+1)-i)\bigr)\Bigr].
	\end{eqnarray*}
	This gives us an explicit formula to calculate the randomized competitive ratio.
	Doing this for all $\capacity\in\bigl\{1,\dots,\bigl\lfloor\sum_{i=0}^{3}r^{i}\bigr\rfloor\bigr\}$ (cf. Figure~\ref{fig:upper_bound_CR_small_k}), we obtain $\E\bigl[f(X_{\alg}(\capacity))]\geq0.569\cdot v_{\capacity}>g(r)\cdot v_{\capacity}$.
	
	\begin{figure}
		\includegraphics[width=7cm]{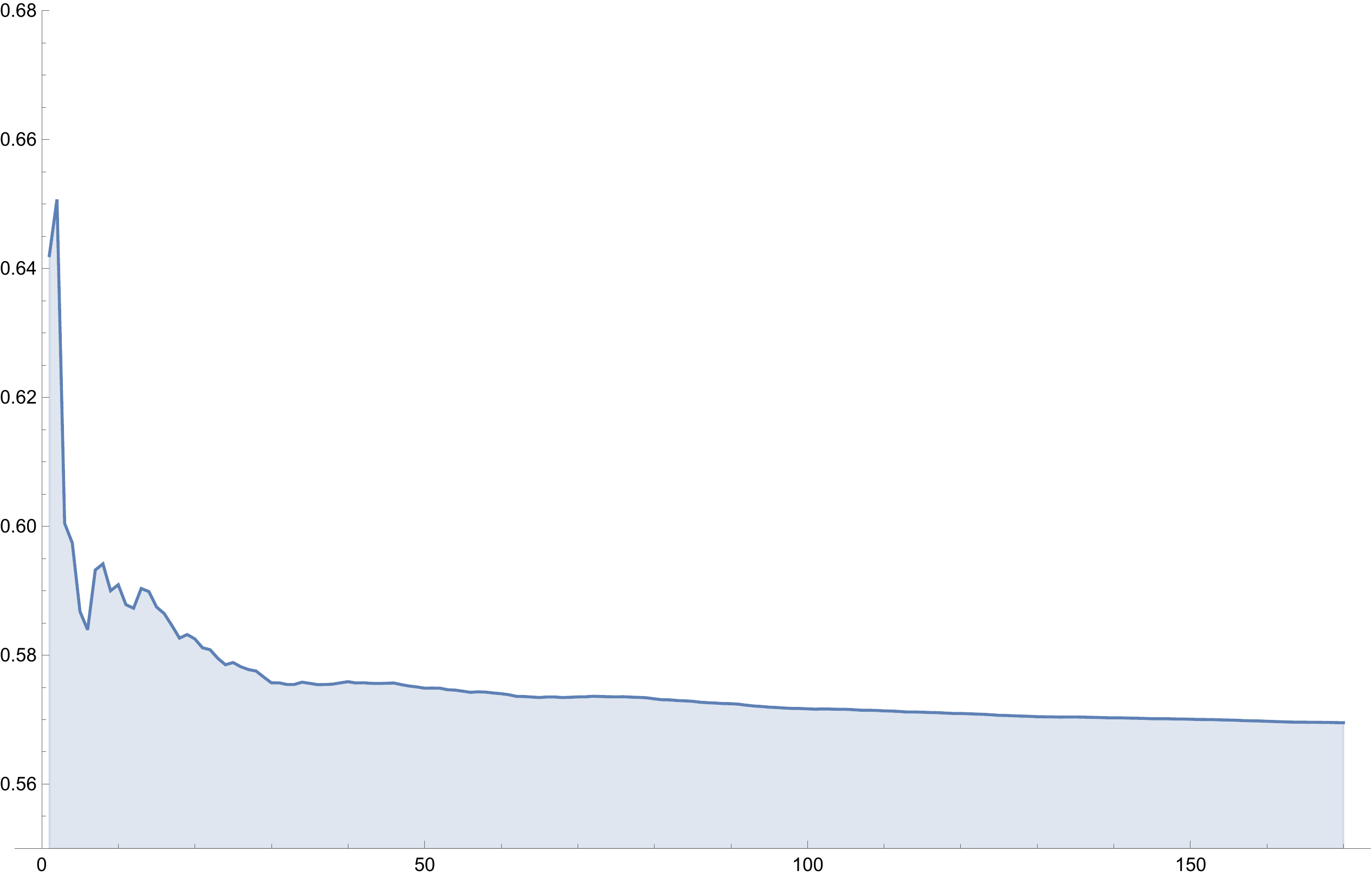}\includegraphics[width=7cm]{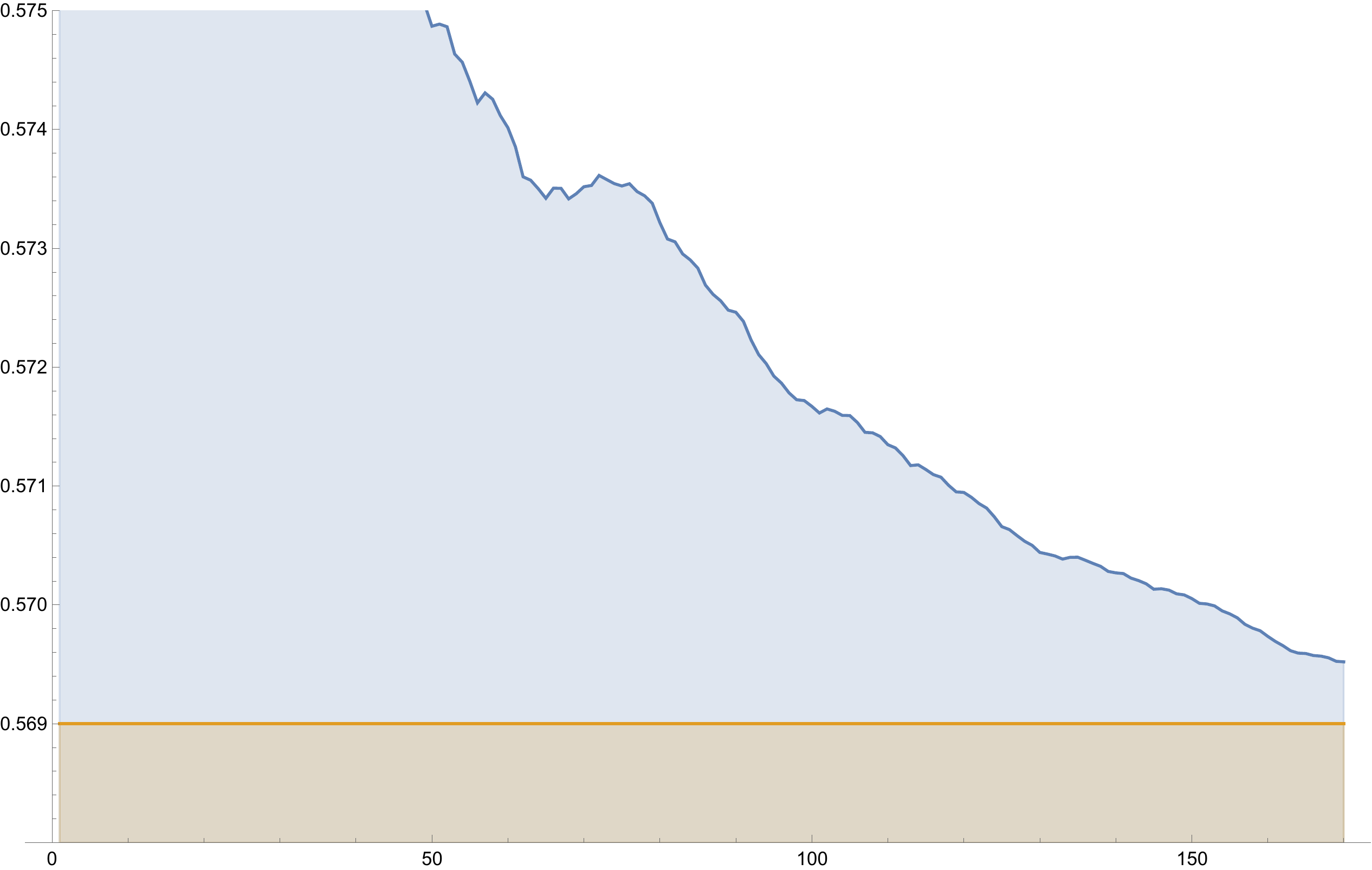}
		
		\caption{\label{fig:upper_bound_CR_small_k}Both graphics show an upper bound
			on the competitive ratio of $\protect\randalg$ for sizes less
			than $\sum_{i=0}^{3}r^{i}$. The horizontal line in the right graphic
			represents the value $0.569$.}
		
	\end{figure}
	
	Suppose $\capacity\geq\sum_{i=0}^{3}r^{i}$ and let $k\in\N$ and $\delta\in(0,1]$ be uniquely defined such that
	\begin{equation}
	\capacity=r^{k+\delta}\label{eq:def_k_delta}
	\end{equation}
	We have
	\[
	\tilde{c}_{k-1}=r^{k-1+\epsilon}\overset{\epsilon\leq1}{\leq}r^{k}\overset{\eqref{eq:def_k_delta}}{<}\capacity\overset{\eqref{eq:def_k_delta}}{\leq}r^{k+1}\overset{\epsilon\geq0}{\leq}r^{k+1+\epsilon}=\tilde{c}_{k+1}
	\]
	and, for all $i\in\N\cup\{0\}$,
	\[
	\tilde{t}_{i}-1<\tilde{c}_{i+1}<\tilde{t}_{i+1}-1
	\]
	where the first inequality follows from the fact that
	\[
	\frac{\tilde{t}_{i}-1}{\tilde{c}_{i+1}}<\frac{\tilde{t}_{i}}{\tilde{c}_{i+1}}=\frac{r^{\epsilon}}{r^{i+1+\epsilon}}\cdot\frac{r^{i+1}-1}{r-1}<\frac{1}{r-1}<1
	\]
	and the second inequality from the fact that $(\tilde{t}_{i+1}-1)-\tilde{c}_{i+1}=\tilde{t}_{i}-1\geq\tilde{c}_{0}-1=r^{\epsilon}-1>0$.
	This means, there are~$4$ intervals, in which $\capacity$ can fall:
	\[
	I_{1}=(\tilde{c}_{k-1},\tilde{t}_{k-1}-1],\quad I_{2}=(\tilde{t}_{k-1}-1,\tilde{c}_{k}],\quad I_{3}=(\tilde{c}_{k},\tilde{t}_{k}-1],\quad I_{4}=(\tilde{t}_{k}-1,\tilde{c}_{k+1}].
	\]
	
	We will calculate for which values of $\epsilon$ the value~$\capacity$ falls into which interval.
	For $i\in\N$, we have
	\begin{eqnarray}
	\capacity\leq\tilde{c}_{i} & \Leftrightarrow & r^{k+\delta}\leq r^{i+\epsilon}\nonumber \\
	& \Leftrightarrow & \epsilon\geq k-i+\delta\label{eq:boundary_ci}
	\end{eqnarray}
	and
	\begin{eqnarray}
	\capacity\leq\tilde{t}_{i}-1 & \Leftrightarrow & r^{k+\delta}\leq r^{\epsilon}\frac{r^{i+1}-1}{r-1}-1\nonumber \\
	& \Leftrightarrow & \epsilon\geq\log_{r}(r^{k+\delta}+1)+\log_{r}(r-1)-\log_{r}(r^{i+1}-1)=:\mu(i).\label{eq:boundary_ti-1}
	\end{eqnarray}
	
	The expected value of the algorithmic solution of size $\capacity$ is
	\begin{eqnarray*}
		&  & \E[f(X_{\alg}(\capacity))]\\
		& = & \sum_{i=1}^{5}\E[f(X_{\alg}(\capacity))\mid\capacity\in I_{i}]\cdot\mathbb{P}[\capacity\in I_{i}]\\
		& \overset{\textrm{Lemma }\ref{lem:estimates_ratio_alg_opt}(ii),(iii)}{\geq} & \Bigl[\E\Bigl[1-\frac{\tilde{t}_{k-2}}{\capacity}\Bigm|\capacity\in(\tilde{c}_{k-1},\tilde{t}_{k-1}-1]\Bigr]\cdot\mathbb{P}[\capacity\in I_{1}]\\
		&  & +\E\Bigl[\max\Bigl\{\frac{\tilde{c}_{k-1}-1}{\capacity},\frac{\capacity-\tilde{t}_{k-1}}{\tilde{c}_{k}}\Bigr\}\Bigm|\capacity\in(\tilde{t}_{k-1}-1,\tilde{c}_{k}]\Bigr]\cdot\mathbb{P}[\capacity\in I_{2}]\\
		&  & +\E\Bigl[1-\frac{\tilde{t}_{k-1}}{\capacity}\Bigm|\capacity\in(\tilde{c}_{k},\tilde{t}_{k}-1]\Bigr]\cdot\mathbb{P}[\capacity\in I_{3}]\\
		&  & +\E\Bigl[\max\Bigl\{\frac{\tilde{c}_{k}-1}{\capacity},\frac{\capacity-\tilde{t}_{k}}{\tilde{c}_{k+1}}\Bigr\}\Bigm|\capacity\in(\tilde{t}_{k}-1,\tilde{c}_{k+1}]\Bigr]\cdot\mathbb{P}[\capacity\in I_{4}]\Bigr]\cdot v_{\capacity}\\
		& \overset{\eqref{eq:boundary_ci},\eqref{eq:boundary_ti-1}}{=} & \Bigl[\int_{\min\bigl\{1,\mu(k-1)\bigr\}}^{1}1-\frac{\tilde{t}_{k-2}}{\capacity}\,d\epsilon\\
		&  & +\int_{\delta}^{\min\{1,\mu(k-1)\}}\max\Bigl\{\frac{\tilde{c}_{k-1}-1}{\capacity},\frac{\capacity-\tilde{t}_{k-1}}{\tilde{c}_{k}}\Bigr\}\,d\epsilon\\
		&  & +\int_{\max\bigl\{0,\mu(k)\bigr\}}^{\delta}1-\frac{\tilde{t}_{k-1}}{\capacity}\,d\epsilon\\
		&  & +\int_{0}^{\max\{0,\mu(k)\}}\max\Bigl\{\frac{\tilde{c}_{k}-1}{\capacity},\frac{\capacity-\tilde{t}_{k}}{\tilde{c}_{k+1}}\Bigr\}\,d\epsilon\Bigr]\cdot v_{\capacity}
	\end{eqnarray*}
	Furthermore, for $i\in\N$ we have
	\begin{eqnarray*}
		&  & \frac{\tilde{c}_{i}-1}{\capacity}\geq\frac{\capacity-\tilde{t}_{i}}{\tilde{c}_{i+1}}\\
		& \Leftrightarrow & r^{i+1+\epsilon}(r^{i+\epsilon}-1)\geq r^{k+\delta}\Bigl(r^{k+\delta}-r^{\epsilon}\frac{r^{i+1}-1}{r-1}\Bigr)\\
		& \Leftrightarrow & r^{2i+1}(r^{\epsilon})^{2}+\Bigl(r^{k+\delta}\frac{r^{i+1}-1}{r-1}-r^{i+1}\Bigr)r^{\epsilon}-r^{2k+2\delta}\geq0\\
		& \overset{r^{\epsilon}>0}{\Leftrightarrow} & r^{\epsilon}\geq-\frac{1}{2r^{i}}\Bigl(r^{k+\delta}\frac{1-r^{-(i+1)}}{r-1}-1\Bigr)\\
		& & \quad\quad +\sqrt{\frac{1}{4r^{2i}}\Bigl(r^{k+\delta}\frac{1-r^{-(i+1)}}{r-1}-1\Bigr)^{2}+r^{2k-2i+2\delta-1}}\\
		& \Leftrightarrow & \epsilon\geq\log_{r}\Bigl(\sqrt{\Bigl(r^{k+\delta}\frac{1-r^{-(i+1)}}{r-1}-1\Bigr)^{2}+4r^{2k+2\delta-1}}-r^{k+\delta}\frac{1-r^{-(i+1)}}{r-1}+1\Bigr)\\
		& & \quad\quad -\log_{r}(2)-i=:\nu(i),
	\end{eqnarray*}
	i.e., instead of one integral over a maximum, we can evaluate two separate integrals, which yields
	
	\begin{alignat*}{2}
		\E[f(X_{\alg}(\capacity))] \geq & \Bigl[\int_{\min\bigl\{1,\mu(k-1)\bigr\}}^{1}1-\frac{\tilde{t}_{k-2}}{\capacity}\,d\epsilon
		&& +\int_{\min\{1,\nu(k-1)\}}^{\min\{1,\mu(k-1)\}}\frac{\tilde{c}_{k-1}-1}{\capacity}\,d\epsilon\\
		& +\int_{\delta}^{\min\{1,\nu(k-1)\}}\frac{\capacity-\tilde{t}_{k-1}}{\tilde{c}_{k}}\,d\epsilon
		&& +\int_{\max\bigl\{0,\mu(k)\bigr\}}^{\delta}1-\frac{\tilde{t}_{k-1}}{\capacity}\,d\epsilon\\
		& +\int_{\max\{0,\nu(k)\}}^{\max\{0,\mu(k)\}}\frac{\tilde{c}_{k}-1}{\capacity}\,d\epsilon
		&& +\int_{0}^{\max\{0,\nu(k)\}}\frac{\capacity-\tilde{t}_{k}}{\tilde{c}_{k+1}}\,d\epsilon\Bigr]\cdot v_{\capacity}\\
		\overset{\textrm{Lemma }\ref{lem:estimate_sum-integrals_g(r)}}{\geq} & g(r)\cdot v_{\capacity} > 0.56437\cdot v_{\capacity} > \frac{1}{1.772} v_{\capacity}.
	\end{alignat*}
	This completes the proof.
\end{proof}

\subsection{Randomized Lower Bound}\label{sec:randLB}

We turn to proving the lower bound in Theorem~\ref{thm:lb} for \incmaxsep.

\randLB*

\begin{proof}
	We fix $N$ to be the number of sets $R_1,\dots,R_N$, leaving $d_1,\dots,d_N$ as parameters to determine the instance; we denote the resulting instance by $I(d_1,\dots,d_N)$.
	Note that, given a probability distribution $p_1,\dots,p_N$ over the elements $\{1,\dots,N\}$ in addition, Yao's principle~\cite{Yao1977} yields
	\[
	\inf_{\alg\in\mathcal{A}_N}\sum_{i=1}^{N}p_i\cdot\frac{i\cdot d_i}{\alg(I(d_1,\dots,d_N),i)}
	\]
	as a lower bound on the randomized competitive ratio of the problem.
	Here, $\alg(I,i)$ denotes the value of the first $i$ elements in the solution produced by $\alg$ on instance $I$, and $\mathcal{A}_N$ is the set of all deterministic algorithms on instances with $N$ sets $R_1,\dots,R_N$.
	As observed earlier, we may assume that
	\[
		\mathcal{A}_N:=\left\{\alg_{c_1,\dots,c_\ell} \;\Big|\; 1\leq c_1<\dots<c_\ell\leq N, \sum_{i=1}^\ell c_i\leq N\right\},
	\]
	where $\alg_{c_1,\dots,c_\ell}$ is the algorithm that first includes all elements of $R_{c_1}$ into the solution, then all elements of $R_{c_2}$, and so on.
	Once it has added the $c_\ell$ elements of $R_{c_\ell}$, it adds some arbitrary elements from then onwards.
	
	We can formulate the problem of maximizing the lower bound on the competitive ratio as an optimization problem:
	\begin{center}
		\begin{minipage}{.3\textwidth}
			\begin{align*}
				\max && \rho\\
				\text{s.t.} && \rho & \leq \sum_{i=1}^{N}p_i\cdot\frac{i\cdot d_i}{\alg(I(d_1,\dots,d_N),i)} && \forall\alg\in\mathcal{A}_N,\\
				&& \sum_{i=1}^{N}p_i & =  1, \\
				&& d_1,\dots,d_N & \geq  0, \\
				&& p_1,\dots,p_N & \geq  0.
			\end{align*}
		\end{minipage}
	\end{center}
	
	
	Note that the expression $\alg_{c_1,\dots,c_\ell}(I(d_1,\dots,d_N),i)$ can also be written as a function of $c_1,\dots,c_\ell$, $d_1,\dots,d_N$, and $i$ by taking the maximum over all sets from which $\alg_{c_1,\dots,c_\ell}$ selects elements:
	\[
		\alg_{c_1,\dots,c_\ell}(I(d_1,\dots,d_N),i)=\max_{1\leq j\leq\ell}\left\{\max\Bigg\{i-\sum_{1\leq j'<j}c_{j'},c_j\Bigg\}\cdot d_{c_j}\right\}.
	\]
	
	A feasible solution to the above optimization problem with $N=10$ is given by
	\begin{align*}
		&(\rho;d_1,\dots,d_{10};p_1,\dots,p_{10})\\
		&= (1.447; 1, 1/2, 1/2, 1/2, 2/5, 1/3, 1/3, 1/3, 1/3, 1/3; 0.132, 0, 0, 0.395, 0, 0, 0, 0, 0, 0.473),
	\end{align*}
	with objective value $1.447$.
\end{proof}



\bibliography{bibliography}

\begin{thebibliography}{10}

\bibitem{BernsteinDisserGrossHimburg/20}
Aaron Bernstein, Yann Disser, Martin Gro{\ss}, and Sandra Himburg.
\newblock General bounds for incremental maximization.
\newblock {\em Math. Program.}, 191(2):953--979, 2022.
\newblock \href {https://doi.org/10.1007/s10107-020-01576-0}
  {\path{doi:10.1007/s10107-020-01576-0}}.

\bibitem{BlumCCPRS94}
Avrim Blum, Prasad Chalasani, Don Coppersmith, William~R. Pulleyblank,
  Prabhakar Raghavan, and Madhu Sudan.
\newblock The minimum latency problem.
\newblock In {\em Proceedings of the Twenty-Sixth Annual {ACM} Symposium on
  Theory of Computing (STOC)}, pages 163--171. {ACM}, 1994.
\newblock \href {https://doi.org/10.1145/195058.195125}
  {\path{doi:10.1145/195058.195125}}.

\bibitem{ChrobakKNY08}
Marek Chrobak, Claire Kenyon, John Noga, and Neal~E. Young.
\newblock Incremental medians via online bidding.
\newblock {\em Algorithmica}, 50(4):455--478, 2008.
\newblock \href {https://doi.org/10.1007/s00453-007-9005-x}
  {\path{doi:10.1007/s00453-007-9005-x}}.

\bibitem{DisserKMS17}
Yann Disser, Max Klimm, Nicole Megow, and Sebastian Stiller.
\newblock Packing a knapsack of unknown capacity.
\newblock {\em {SIAM} J. Discret. Math.}, 31(3):1477--1497, 2017.
\newblock \href {https://doi.org/10.1137/16M1070049}
  {\path{doi:10.1137/16M1070049}}.

\bibitem{DisserKW21}
Yann Disser, Max Klimm, and David Weckbecker.
\newblock Fractionally subadditive maximization under an incremental knapsack
  constraint.
\newblock In {\em Proceedings of the 19th International Workshop on
  Approximation and Online Algorithms (WAOA)}, pages 206--223. Springer, 2021.
\newblock \href {https://doi.org/10.1007/978-3-030-92702-8\_13}
  {\path{doi:10.1007/978-3-030-92702-8\_13}}.

\bibitem{FujitaKM13}
Ryo Fujita, Yusuke Kobayashi, and Kazuhisa Makino.
\newblock Robust matchings and matroid intersections.
\newblock {\em {SIAM} J. Discret. Math.}, 27(3):1234--1256, 2013.
\newblock \href {https://doi.org/10.1137/100808800}
  {\path{doi:10.1137/100808800}}.

\bibitem{GoemansK98}
Michel~X. Goemans and Jon~M. Kleinberg.
\newblock An improved approximation ratio for the minimum latency problem.
\newblock {\em Math. Program.}, 82:111--124, 1998.
\newblock \href {https://doi.org/10.1007/BF01585867}
  {\path{doi:10.1007/BF01585867}}.

\bibitem{GoemansU17}
Michel~X. Goemans and Francisco Unda.
\newblock Approximating incremental combinatorial optimization problems.
\newblock In {\em Proceedings of the 20th International Workshop on
  Approximation Algorithms for Combinatorial Optimization Problems (APPROX)},
  pages 6:1--6:14, 2017.
\newblock \href {https://doi.org/10.4230/LIPIcs.APPROX-RANDOM.2017.6}
  {\path{doi:10.4230/LIPIcs.APPROX-RANDOM.2017.6}}.

\bibitem{Gonzalez85}
Teofilo~F. Gonzalez.
\newblock Clustering to minimize the maximum intercluster distance.
\newblock {\em Theor. Comput. Sci.}, 38:293--306, 1985.
\newblock \href {https://doi.org/10.1016/0304-3975(85)90224-5}
  {\path{doi:10.1016/0304-3975(85)90224-5}}.

\bibitem{HartlineS06}
Jeff Hartline and Alexa Sharp.
\newblock An incremental model for combinatorial maximization problems.
\newblock In {\em Proceedings of the 5th International Workshop on Experimental
  Algorithms (WEA)}, pages 36--48, 2006.
\newblock \href {https://doi.org/10.1007/11764298\_4}
  {\path{doi:10.1007/11764298\_4}}.

\bibitem{HartlineS07}
Jeff Hartline and Alexa Sharp.
\newblock Incremental flow.
\newblock {\em Networks}, 50(1):77--85, 2007.
\newblock \href {https://doi.org/10.1002/net.20168}
  {\path{doi:10.1002/net.20168}}.

\bibitem{HassinR02}
Refael Hassin and Shlomi Rubinstein.
\newblock Robust matchings.
\newblock {\em {SIAM} J. Discret. Math.}, 15(4):530--537, 2002.
\newblock \href {https://doi.org/10.1137/S0895480198332156}
  {\path{doi:10.1137/S0895480198332156}}.

\bibitem{HassinS06}
Refael Hassin and Danny Segev.
\newblock Robust subgraphs for trees and paths.
\newblock {\em {ACM} Trans. Algorithms}, 2(2):263--281, 2006.
\newblock \href {https://doi.org/10.1145/1150334.1150341}
  {\path{doi:10.1145/1150334.1150341}}.

\bibitem{KakimuraM13}
Naonori Kakimura and Kazuhisa Makino.
\newblock Robust independence systems.
\newblock {\em {SIAM} J. Discret. Math.}, 27(3):1257--1273, 2013.
\newblock \href {https://doi.org/10.1137/120899480}
  {\path{doi:10.1137/120899480}}.

\bibitem{KaoReifTate/93}
Ming-Yang Kao, John~H Reif, and Stephen~R Tate.
\newblock Searching in an unknown environment: An optimal randomized algorithm
  for the cow-path problem.
\newblock {\em Information and Computation}, 131(1):63--79, 1996.

\bibitem{KlimmK22}
Max Klimm and Martin Knaack.
\newblock Maximizing a submodular function with bounded curvature under an
  unknown knapsack constraint.
\newblock In {\em Proceedings of the 25th International Workshop on
  Approximation Algorithms for Combinatorial Optimization Problems (APPROX)},
  pages 49:1--49:19, 2022.
\newblock \href {https://doi.org/10.4230/LIPIcs.APPROX/RANDOM.2022.49}
  {\path{doi:10.4230/LIPIcs.APPROX/RANDOM.2022.49}}.

\bibitem{LinNRW10}
Guolong Lin, Chandrashekhar Nagarajan, Rajmohan Rajaraman, and David~P.
  Williamson.
\newblock A general approach for incremental approximation and hierarchical
  clustering.
\newblock {\em {SIAM} J. Comput.}, 39(8):3633--3669, 2010.
\newblock \href {https://doi.org/10.1137/070698257}
  {\path{doi:10.1137/070698257}}.

\bibitem{MatuschkeSS18}
Jannik Matuschke, Martin Skutella, and Jos{\'{e}}~A. Soto.
\newblock Robust randomized matchings.
\newblock {\em Math. Oper. Res.}, 43(2):675--692, 2018.
\newblock \href {https://doi.org/10.1287/moor.2017.0878}
  {\path{doi:10.1287/moor.2017.0878}}.

\bibitem{MegowM13}
Nicole Megow and Juli{\'{a}}n Mestre.
\newblock Instance-sensitive robustness guarantees for sequencing with unknown
  packing and covering constraints.
\newblock In {\em Proceedings of the 4th Innovations in Theoretical Computer
  Science Conference (ITCS)}, pages 495--504, 2013.
\newblock \href {https://doi.org/10.1145/2422436.2422490}
  {\path{doi:10.1145/2422436.2422490}}.

\bibitem{Mestre06}
Juli{\'{a}}n Mestre.
\newblock Greedy in approximation algorithms.
\newblock In {\em Proceedings of the 14th Annual European Symposium on
  Algorithms (ESA)}, pages 528--539, 2006.
\newblock \href {https://doi.org/10.1007/11841036\_48}
  {\path{doi:10.1007/11841036\_48}}.

\bibitem{MettuP03}
Ramgopal~R. Mettu and C.~Greg Plaxton.
\newblock The online median problem.
\newblock {\em {SIAM} J. Comput.}, 32(3):816--832, 2003.
\newblock \href {https://doi.org/10.1137/S0097539701383443}
  {\path{doi:10.1137/S0097539701383443}}.

\bibitem{Plaxton06}
C.~Greg Plaxton.
\newblock Approximation algorithms for hierarchical location problems.
\newblock {\em J. Comput. Syst. Sci.}, 72(3):425--443, 2006.
\newblock \href {https://doi.org/10.1016/j.jcss.2005.09.004}
  {\path{doi:10.1016/j.jcss.2005.09.004}}.

\bibitem{Yao1977}
Andrew Chi-Chin Yao.
\newblock Probabilistic computations: Toward a unified measure of complexity.
\newblock In {\em 18th Annual Symposium on Foundations of Computer Science
  (sfcs 1977)}, pages 222--227, 1977.
\newblock \href {https://doi.org/10.1109/SFCS.1977.24}
  {\path{doi:10.1109/SFCS.1977.24}}.

\end{thebibliography}

\end{document}